\def\anon{1}
\def\archiveauthor{1}
\def\preprint{1}

\PassOptionsToPackage{unicode}{hyperref}
\PassOptionsToPackage{hyphens}{url}
\PassOptionsToPackage{dvipsnames,svgnames,x11names}{xcolor}
\documentclass[12pt]{article}

\usepackage{amsmath,amssymb}
\usepackage{iftex}
\ifPDFTeX
  \usepackage[T1]{fontenc}
  \usepackage[utf8]{inputenc}
  \usepackage{textcomp}
\else
  \usepackage{unicode-math}
  \defaultfontfeatures{Scale=MatchLowercase}
  \defaultfontfeatures[\rmfamily]{Ligatures=TeX,Scale=1}
\fi
\usepackage{lmodern}
\IfFileExists{upquote.sty}{\usepackage{upquote}}{}
\IfFileExists{microtype.sty}{%
  \usepackage[]{microtype}
  \UseMicrotypeSet[protrusion]{basicmath}
}{}
\makeatletter
\@ifundefined{KOMAClassName}{%
  \IfFileExists{parskip.sty}{%
    \usepackage{parskip}
  }{%
    \setlength{\parindent}{0pt}
    \setlength{\parskip}{6pt plus 2pt minus 1pt}}
}{%
  \KOMAoptions{parskip=half}}
\makeatother
\usepackage{xcolor}
\makeatletter
\ifx\paragraph\undefined\else
  \let\oldparagraph\paragraph
  \renewcommand{\paragraph}{\@ifstar\xxxParagraphStar\xxxParagraphNoStar}
  \newcommand{\xxxParagraphStar}[1]{\oldparagraph*{#1}\mbox{}}
  \newcommand{\xxxParagraphNoStar}[1]{\oldparagraph{#1}\mbox{}}
\fi
\ifx\subparagraph\undefined\else
  \let\oldsubparagraph\subparagraph
  \renewcommand{\subparagraph}{\@ifstar\xxxSubParagraphStar\xxxSubParagraphNoStar}
  \newcommand{\xxxSubParagraphStar}[1]{\oldsubparagraph*{#1}\mbox{}}
  \newcommand{\xxxSubParagraphNoStar}[1]{\oldsubparagraph{#1}\mbox{}}
\fi
\makeatother

\usepackage{longtable,booktabs,array}
\usepackage{calc}
\usepackage{etoolbox}
\makeatletter
\patchcmd\longtable{\par}{\if@noskipsec\mbox{}\fi\par}{}{}
\makeatother
\IfFileExists{footnotehyper.sty}{\usepackage{footnotehyper}}{\usepackage{footnote}}
\makesavenoteenv{longtable}
\usepackage{graphicx}
\makeatletter
\def\maxwidth{\ifdim\Gin@nat@width>\linewidth\linewidth\else\Gin@nat@width\fi}
\def\maxheight{\ifdim\Gin@nat@height>\textheight\textheight\else\Gin@nat@height\fi}
\makeatother
\setkeys{Gin}{width=\maxwidth,height=\maxheight,keepaspectratio}
\makeatletter
\def\fps@figure{htbp}
\makeatother

\makeatletter
\@ifpackageloaded{caption}{}{\usepackage{caption}}
\AtBeginDocument{%

  }
\@ifpackageloaded{float}{}{\usepackage{float}}
\floatstyle{ruled}
\@ifundefined{c@chapter}{\newfloat{codelisting}{h}{lop}}{\newfloat{codelisting}{h}{lop}[chapter]}
\floatname{codelisting}{Listing}

\@ifpackageloaded{subcaption}{}{\usepackage{subcaption}}
\makeatother

\ifLuaTeX
  \usepackage{selnolig}
\fi
\usepackage[]{natbib}
\usepackage{bookmark}
\IfFileExists{xurl.sty}{\usepackage{xurl}}{}
\usepackage{amsthm,xr-hyper,bm,ifsym,dsfont,multirow,rotating,lscape,tikz,bbm}
\usepackage[normalem]{ulem}
\usepackage{enumerate,algorithm,algorithmic,bigstrut,epstopdf,accents,makecell}
\usepackage[graphicx]{realboxes}
\usepackage{cleveref}

\providecommand{\anon}{0}
\providecommand{\preprint}{0}
\providecommand{\preprintdate}{September 4, 2026}
\if1\anon
  \newcommand{\pdfauthorfield}{Ye Shi; Xiao Jin; Chung-Piaw Teo}
\else
  \newcommand{\pdfauthorfield}{Anonymous}
\fi

\hypersetup{
  pdftitle={Optimizer as Detector: Stochastic Gradient Descent for Latent Mixture Models},
  pdfauthor={\pdfauthorfield},
  pdfkeywords={homogeneity testing; competitive loss; diffusion approximation; latent heterogeneity},
  colorlinks=true,
  linkcolor={blue},
  filecolor={Maroon},
  citecolor={Blue},
  urlcolor={Blue},
  pdfcreator={LaTeX}}

\bibpunct[, ]{(}{)}{,}{a}{}{,}%
\theoremstyle{plain}
\newtheorem{theorem}{Theorem}
\newtheorem{proposition}{Proposition}
\newtheorem{lemma}{Lemma}
\newtheorem{corollary}{Corollary}
\theoremstyle{definition}
\newtheorem{assumption}{Assumption}

\theoremstyle{remark}
\newtheorem{remark}{Remark}

\crefname{theorem}{Theorem}{Theorems}
\Crefname{theorem}{Theorem}{Theorems}
\crefname{proposition}{Proposition}{Propositions}
\Crefname{proposition}{Proposition}{Propositions}
\crefname{lemma}{Lemma}{Lemmas}
\Crefname{lemma}{Lemma}{Lemmas}
\crefname{corollary}{Corollary}{Corollaries}
\Crefname{corollary}{Corollary}{Corollaries}
\crefname{assumption}{Assumption}{Assumptions}
\Crefname{assumption}{Assumption}{Assumptions}
\crefname{definition}{Definition}{Definitions}
\Crefname{definition}{Definition}{Definitions}
\crefname{remark}{Remark}{Remarks}
\Crefname{remark}{Remark}{Remarks}

\makeatletter
\newcommand*\bigcdot{\mathpalette\bigcdot@{.75}}
\newcommand*\bigcdot@[2]{\mathbin{\vcenter{\hbox{\scalebox{#2}{$\m@th#1\bullet$}}}}}
\makeatother

\newcommand{\R}{\mathbb{R}}
\newcommand{\E}{\mathbb{E}}
\newcommand{\Var}{\mathrm{Var}}
\newcommand{\sgn}{\mathrm{sgn}}
\newcommand{\bfe}{\bar{e}}
\newcommand{\ind}{\mathbf{1}}

\begin{document}
\def\spacingset#1{\renewcommand{\baselinestretch}%
{#1}\small\normalsize} \spacingset{1}

\providecommand{\archiveauthor}{0}
\if1\anon
{
  \title{\bf Optimizer as Detector: Stochastic Gradient Descent for Latent Mixture Models}
  \if1\preprint
    \date{\preprintdate}
  \else
    \date{}
  \fi
  \if1\archiveauthor
  \author{
    Ye Shi$^{*}$, Xiao Jin$^{+}$, and Chung-Piaw Teo$^{\dagger}$\\[0.5em]
    {\footnotesize $^{*}$School of Management, University of Science and Technology of China,}\\
    {\footnotesize $^{+}$Institute of Operations Research and Analytics,
    National University of Singapore}\\
    {\footnotesize $^{\dagger}$NUS Business School, National University of Singapore}}
  \else
  \author{
    Ye Shi\thanks{Corresponding author: Ye Shi; E-mail: hecules@ustc.edu.cn}\\
    School of Management, University of Science and Technology of China\\
    Hefei, Anhui 230026, China
    \and
    Xiao Jin\\ Institute of Operations Research and Analytics\\
    National University of Singapore, Singapore 117602
    \and
    Chung-Piaw Teo\\ NUS Business School\\
    National University of Singapore, Singapore 119245}
  \fi
  \maketitle
} \fi

\if0\anon
{
  \bigskip
  \bigskip
  \bigskip
  \begin{center}
    {\LARGE\bf Optimizer as Detector: Stochastic Gradient Descent for Latent Mixture Models}
  \end{center}
  \medskip
} \fi

\bigskip
\begin{abstract}
Pooling latent subpopulations can obscure relationships and yield misleading regression conclusions, including Simpson's paradox (SP).  We propose a detector based on the steady-state dynamics of constant-step stochastic gradient descent (SGD).  Unlike likelihood-based mixture tests and confounder-search methods, it requires neither a normal-mixture specification nor observed candidate confounders.  Two linear pieces compete under a winner-take-all squared-error loss, and their normalized terminal separation forms the test statistic.  Using diffusion approximations, we derive its asymptotic null distribution for general centered scalar covariates and for Gaussian multivariate covariates under symmetric noise.  The distribution-dependent null center reduces to the dimension-free constant $4/\pi$ when both covariates and noise are Gaussian.  We establish asymptotic size control and consistency against fixed mixture alternatives under regularity conditions.  We extend the method to intercept and partial-mixture heterogeneity and study endogeneity, heteroskedasticity, and nonlinear misspecification.  Simulations examine calibration, power, and robustness.  Finally, a three-stage Detect--Screen--Verify toolkit separates evidence of heterogeneity from its substantive explanation.  Across eight public datasets, it recovers four established SP benchmarks and identifies four cases that, to our knowledge, have not been documented previously.  The detector requires neither latent-group labels nor a prespecified number of mixture components.
\end{abstract}

\vfill

\newpage
\spacingset{1.8} 

\section{Introduction}\label{sec:intro}

Many real-world regression analyses pool observations drawn from different subpopulations.  When the subgroup labels are unobserved, a single fitted relationship can conceal meaningful differences across the underlying groups and lead to misleading conclusions.  Simpson's paradox (SP) is an extreme manifestation of this problem: the aggregate and conditional associations have opposite signs after stratification~\citep{blyth1972simpson,Pearl2009}.  For the empirical comparisons, motivated by the practical relaxation of the all-strata condition in \citet{alipourfard2018simpsons}, we record SP when a strict majority of eligible strata reverse the pooled sign.  Pooling can therefore reverse a substantive scientific or policy conclusion, a risk that continues to attract attention in statistics, economics, and public policy~\citep{mittal1991homogeneity,albers2015dutch}.

One natural response to latent heterogeneity is to estimate a mixture regression.  Statistical approaches typically rely on likelihood-based estimation~\citep{quandt1972switching}, most often implemented through the expectation-maximization (EM) algorithm or variants such as alternating minimization and online EM~\citep{dempster1977,yi2014alternating,liu2024convergence}.  In particular, \citet{kasahara2015testing} develop a formal modified penalized-EM test of $m=m_0$ against $m=m_0+1$ components in normal mixture regression.  These methods impose parametric distributional assumptions, and estimation generally conditions on a candidate number of mixture components; moreover, most theoretical guarantees focus on two-component mixtures and establish local or global convergence~\citep{balakrishnan2017statistical,kwon2019global}.  Optimization-based alternatives, including mixed-integer programming (MIP), jointly assign observations to clusters and estimate their regression parameters, but computational scalability can become a central concern~\citep{park2017algorithms,wang2019convergence}.  Several studies instead search directly for SP using exhaustive enumeration or classification-tree methods, both of which require observed candidate confounders~\citep{xu2018detecting,shmueli2018forest}.

\subsection{Our Approach}

We approach this question through the dynamics of an optimizer.  Given observations $\{(y_t,x_t)\}_{t=1}^T$, we define a nonconvex \textit{competitive loss} as the smaller of the squared residuals from two linear models with parameters $\theta_1$ and $\theta_2$.  Applying constant-step stochastic gradient descent (SGD) to this loss creates winner-take-all competition: at each round, only the better-fitting model is updated.  Our key observation is that the resulting competition encodes information about whether the data are homogeneous.  \Cref{sec:setup} states the loss, algorithm, and hypotheses formally, with $H_0$ denoting a single linear model and $H_1$ a latent mixture of distinct linear models.

\begin{figure}[ht]
	\centering
	\begin{subfigure}{0.4\textwidth}
		\centering
		\includegraphics[width=\linewidth]{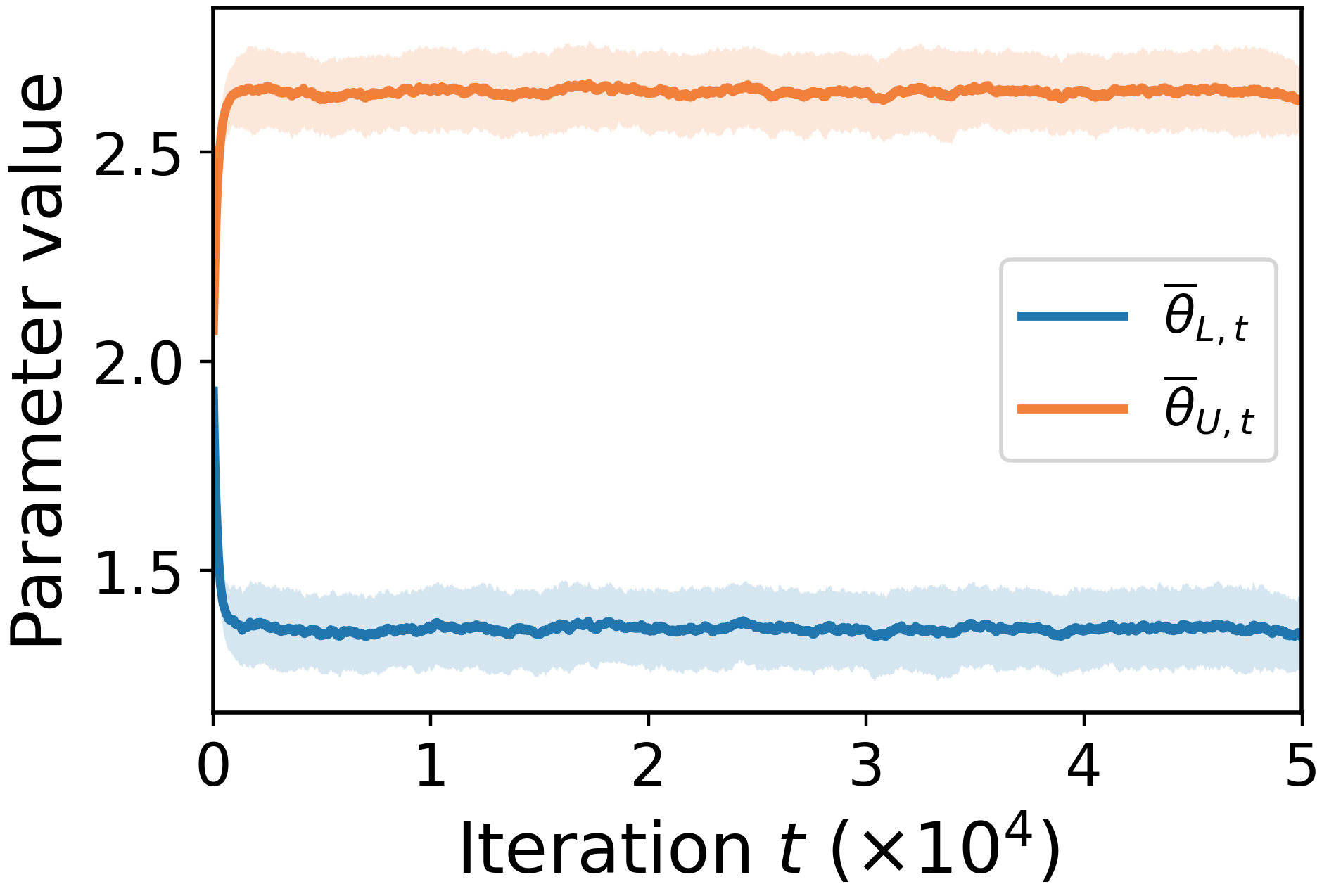}
		\caption{Mean ordered trajectories}
		\label{fig:detect-traj}
	\end{subfigure}
	\qquad
	\begin{subfigure}{0.4\textwidth}
		\centering
		\includegraphics[width=\linewidth]{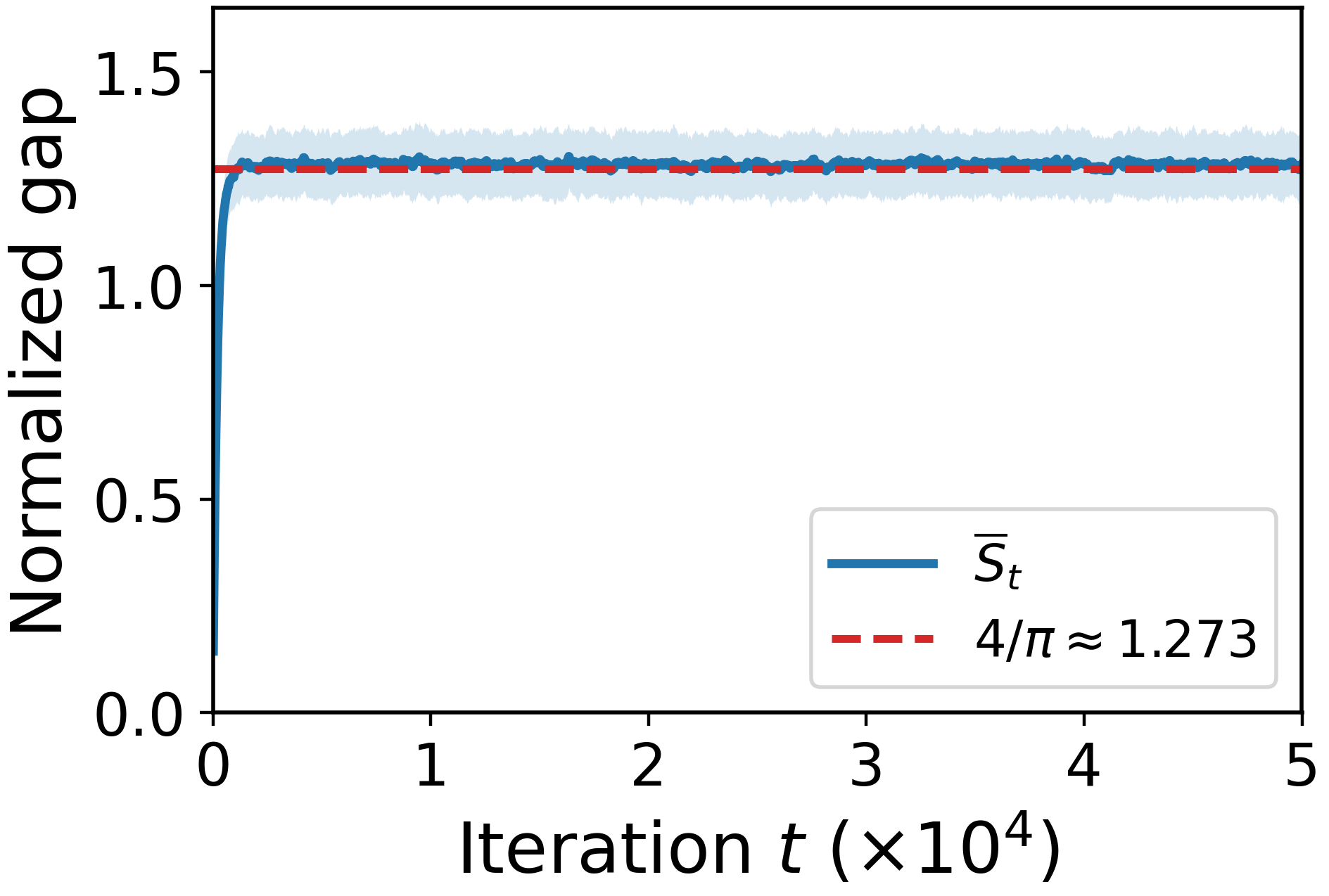}
		\caption{Mean normalized gap}
		\label{fig:detect-ST}
	\end{subfigure}
	\caption{Constant-step SGD dynamics under $H_0$.  Within each path, the iterates are ordered as lower ($L$) and upper ($U$) before averaging.  Solid curves are pointwise across-path means, and shaded bands show one across-path standard deviation.  The red dashed line in panel~(b) marks the later-derived Gaussian null benchmark $4/\pi$.}
	\label{fig:detection-example}
\end{figure}

\Cref{fig:detection-example} illustrates the mechanism in a homogeneous scalar example under $H_0$: $y_t=x_t^\top\theta^{\mathrm{true}}+\varepsilon_t$, where $\theta^{\mathrm{true}}=2$, $x_t,\varepsilon_t\sim N(0,1)$, and $\eta=0.01$.  The figure averages 200 independent paths after ordering the two iterates within each path.  Although the data contain only one regression model, the mean lower and upper iterates do not collapse to the same value.  Instead, they \textit{spontaneously separate} and approach a stable, symmetric long-run pattern around $\theta^{\mathrm{true}}$ (panel~(a)), while their mean normalized gap settles near the later-derived Gaussian benchmark $4/\pi\approx1.273$ (panel~(b)).  Constant-step updates also generate persistent fluctuations around this pattern, as shown by the shaded bands.  Thus, even homogeneous data produce a positive and random terminal gap.

This positive gap under homogeneity motivates our test statistic.  Rather than asking whether the two SGD pieces differ, we compare their \textit{normalized gap} at the terminal iteration with a fitted upper null critical value.  A gap that is unusually large relative to this benchmark provides evidence against a homogeneous regression.  In this way, the optimizer itself supplies the test statistic without fitting a finite-mixture likelihood, specifying the true number of latent groups, or observing their labels.  Making this rule inferentially valid, however, requires us to characterize the null dynamics, calibrate their persistent stochastic fluctuations, and establish that latent mixtures produce a distinguishable signal.

 \subsection{Contributions}

First, we provide the theoretical foundation for our method through a diffusion-approximation analysis \citep{fan2018statistical,hu2019diffusion,wang2020asymptotic,liu2021diffusion}.  The analysis has two layers.  A deterministic ODE describes the mean-field dynamics and identifies where the competing models settle, whereas a local SDE characterizes their stochastic fluctuations.  Under the null hypothesis $H_0$, the ODE has stable antisymmetric equilibria: with a scalar covariate, the two estimator locations are symmetrically positioned around the homogeneous coefficient under general centered covariates and symmetric noise; with Gaussian multivariate covariates, the equilibria form a Mahalanobis ellipsoid.  Under the alternative $H_1$, by contrast, no coincident state is stable, and the separated equilibria covered by our analysis have component-noise-normalized gaps strictly larger than their matched null references.  The SDE analysis then yields an asymptotically Gaussian law for the null gap, whose distribution-dependent center reduces to the dimension-free constant $4/\pi$ when both covariates and noise are Gaussian.  This contrast between the null and alternative dynamics establishes consistency under explicit terminal-concentration and calibration conditions.

Second, we translate this dynamical contrast into a feasible testing procedure.  Symmetric initialization directs the iterates toward the desired null regime, and a fitted upper null quantile yields asymptotic Type~I error control.  Because power depends on standardizing by the noise within mixture components rather than by the pooled residual variation, we develop a variance function regression that estimates the required component-noise scale under both hypotheses when its quadratic terms are identified.  We also extend the procedure to intercept heterogeneity and partial mixture structure.  Finally, we examine robustness to endogeneity, heteroskedasticity, and nonlinear misspecification.

Third, we demonstrate the empirical behavior of the detector.  Synthetic experiments examine its null calibration, power, and robustness across a range of data-generating processes.  We then embed the test in a three-stage Detect--Screen--Verify (DSV) toolkit that deliberately separates the existence of heterogeneity from its substantive explanation.  Across eight datasets spanning several disciplines, DSV reproduces four established SP benchmarks and uncovers four additional cases that, to our knowledge, have not been documented previously.  

\paragraph*{Literature Review.}
Our work connects two streams of literature: mixture regression and heterogeneity detection, and stochastic gradient descent and diffusion approximation.  To conserve space, we provide a detailed review in Section~\ref{app:literature-review} of the Supplementary Materials.

\paragraph*{Organization and Notation.}
\Cref{sec:setup} introduces the model and procedure, \Cref{sec:detection} develops the theory, \Cref{sec:extensions} studies extensions and robustness, and \Cref{sec:framework} presents simulations and applications; concluding remarks are shown in \Cref{sec:conclusion} .
Throughout this paper, $Z\sim N(0,1)$ denotes a standard normal random variable, and $x^\top y$ denotes the inner product.  The Mahalanobis norm is $\|e\|_{\Sigma}=\sqrt{e^\top\Sigma\, e}$, $\|A\|_{\mathrm{op}}$ denotes the spectral (operator) norm of a matrix~$A$, the indicator function is $\ind\{\cdot\}$, and $\dot{f} = df/d\tau$ denotes the time derivative.  We write $[K] = \{1,\ldots,K\}$.  We use $X_n = O_p(a_n)$ to mean $X_n/a_n$ is bounded in probability, i.e., for every $\epsilon > 0$ there exists $M$ such that $\Pr(|X_n/a_n| > M) < \epsilon$ for all~$n$.  Similarly, $X_n = o_p(a_n)$ means $X_n/a_n \to 0$ in probability.  We write $a_n \asymp b_n$ if $a_n/b_n$ is bounded away from $0$ and $\infty$.



\section{Problem Setup}\label{sec:setup}


Consider a setting in which a data analyst observes regression data $\{(y_t, x_t)\}_{t\in [T]}$, where $y_t\in \R$ is the outcome and $x_t\in \mathcal{X} \subseteq \R^d$ is the covariate vector.
We formalize the detection question as a hypothesis test.  Under the null hypothesis, the data arise from a homogeneous linear model:
\[
	H_0:\; y_t = x_t^\top\theta^{\mathrm{true}} + \varepsilon_t \quad \text{for all } t,
	\qquad \theta^{\mathrm{true}}\in\R^d.
\]
Under the alternative, the triples $(x_t,j_t,\varepsilon_t)$ are i.i.d., where the latent group $j_t\in[K^*]$ ($K^*\ge 2$) has positive marginal probabilities $\pi_z=\Pr(j_t=z)>0$.  The outcome follows a latent mixture model, also known as mixture regression:
\begin{equation}\label{eq:dgp}
	H_1:\; y_t = x_t^\top\theta_{j_t}^{\mathrm{true}} + \varepsilon_t,
	\qquad \theta_i^{\mathrm{true}}\neq\theta_j^{\mathrm{true}} \text{ for some } i\neq j,
\end{equation}
where $\theta_1^{\mathrm{true}},\ldots,\theta_{K^*}^{\mathrm{true}}\in\R^d$ are the true group-specific parameters and $\varepsilon_t$ is additive noise independent of $(x_t,j_t)$.

Throughout, the analyst does not know $K^*$, $\pi$, or any confounder or group labels $j_t$ (under $H_1$) to carry out the test; the detection procedure requires only the regression data.

We now describe the algorithmic primitive that underlies the test.  The analyst maintains two parameter vectors $\theta=(\theta_1,\theta_2)\in\R^{2d}$ and minimizes the \textit{competitive} loss
\begin{equation}\label{eq:competitive-loss}
		f_t(\theta)=\min_{k\in\{1,2\}}\{g_{k,t}(\theta)\}, t\in[T],
\end{equation}
where 	$g_{k,t}(\theta) =(y_t-\theta_k^\top x_t)^2$ for $k\in\{1,2\}$.   We call \eqref{eq:competitive-loss} the \textit{competitive loss} because the $\min$ operator selects the better-fitting of the two models at each round, creating a winner-take-all competition between the two pieces; it also renders the composite loss non-convex.  

\begin{algorithm}[ht]
	\caption{SGD with the Competitive Loss} \footnotesize
	\label{alg:sgd}
	\begin{algorithmic}[1]
		\REQUIRE Data $\{(y_t, x_t)\}_{t\in [T]}$; step size $\eta  > 0$; initial estimates $\theta_{1,0}, \theta_{2,0}$
		\FOR{$t = 1, \dots, T$}
		\STATE $r_{k,t} \leftarrow y_t - \theta_{k,t}^\top x_t$ for $k = 1, 2$
		\STATE $w_t \leftarrow \arg\min_{k \in \{1,2\}}\{ r_{k,t}^2\}$ (break ties uniformly at random)
		\STATE $\theta_{w_t,t+1} \leftarrow \theta_{w_t,t} + \eta \, r_{w_t,t}\, x_t$;\quad $\theta_{k,t+1} \leftarrow \theta_{k,t}$ for $k \neq w_t$
		\ENDFOR
	\end{algorithmic}
\end{algorithm}

The data analyst minimizes this loss via SGD, as shown in \Cref{alg:sgd}.  Let $r_{k,t} = y_t - \theta_{k,t}^\top x_t$ denote the residual for piece~$k\in\{1,2\}$ at round~$t$, and let $w_t$ be a single minimizer of $r_{k,t}^2$, chosen uniformly at random when there is a tie.  The update uses the constant residual-update coefficient $\eta>0$:
\begin{equation}\label{eq:sgd-update}
	\theta_{k,t+1} = \theta_{k,t} + \eta\, \ind\{k=w_t\}\, r_{k,t}\, x_t,
	\qquad k \in \{1, 2\}.
\end{equation}
At each round, only the winning piece $w_t$ takes a gradient step and the other is left unchanged.  Because $\eta$ is held constant and the data $(x_t,\varepsilon_t)$ are i.i.d., the iterate sequence $(\theta_{1,t},\theta_{2,t})_{t\ge1}$ is a time-homogeneous \textit{Markov chain}.  Unlike decaying-step SGD, its contracting coordinates retain fluctuations of order $\sqrt\eta$ after the chain enters a locally attracting separated regime.  The null law in \Cref{prop:null-distribution} is the equilibrium diffusion law of the gap functional in that localized regime.


Because the competitive loss is nonconvex and nonsmooth, we analyze~\eqref{eq:sgd-update} using the weak-convergence framework for constant-step stochastic approximation~\citep{kushner2003}.  The resulting ODE identifies the mean-field equilibria and their stability, while a local SDE describes the $\sqrt\eta$ fluctuations around a selected equilibrium.  

Under $H_0$, let $e_{k,t}=\theta^{\mathrm{true}}-\theta_{k,t}$ denote the error of piece~$k$; thus $e_{1,t}=-e_{2,t}$ means that the two estimators are symmetric around $\theta^{\mathrm{true}}$.  The error recursion is
\begin{equation}\label{eq:error-recursion}
	e_{k,t+1} = e_{k,t} - \eta\,\ind\{k=w_t\}\,r_{k,t}\,x_t, \qquad k \in\{1,2\},
\end{equation}
where $r_{k,t} = \varepsilon_t + e_{k,t}^\top x_t$.  Write $I_{k,t}=\ind\{k=w_t\}$.  Given the current errors, decompose the increment $\Delta e_{k,t}=e_{k,t+1}-e_{k,t}$ into its conditional mean and a martingale difference:
\begin{equation}\label{eq:decomposition}
	\Delta e_{k,t} = \eta\,\mu_k(e_{1,t}, e_{2,t}) + \eta\,\xi_{k,t}.
\end{equation}
For a generic state $(e_1,e_2)$ and fresh observation $(x,\varepsilon)$, let $I_k(e_1,e_2;x,\varepsilon)$ denote the winner indicator under the same uniform tie rule.  The conditional mean
\[
	\mu_k(e_1,e_2)=-\E\!\left[I_k(e_1,e_2;x,\varepsilon)(\varepsilon+e_k^\top x)x\right]
\]
is the \textit{drift}.  Define the martingale-difference term
$
	\xi_{k,t}:=-I_{k,t}r_{k,t}x_t-\mu_k(e_{1,t},e_{2,t}),
$
which has conditional mean zero given the current errors.  Let $\Gamma(e)\Gamma(e)^\top$ denote the full joint covariance of the two-piece SGD noise vector $(\xi_1,\xi_2)$ at $e=(e_1,e_2)$.

On the rescaled time $\tau=\eta t$, the mean-field limit is
\begin{equation}\label{eq:mean-ode}
	\dot{e}_k = \mu_k(e_1, e_2), \quad k \in \{1, 2\},
\end{equation}
whose stable equilibria satisfy $\mu_k(e_1,e_2)=0$.  Retaining the leading stochastic term gives
\begin{equation}\label{eq:sde}
	de = \mu(e)\,d\tau + \sqrt{\eta}\,\Gamma(e)\,dW_\tau,
\end{equation}
where $\mu=(\mu_1,\mu_2)$ and $W_\tau$ is standard Brownian motion.  Linearizing this diffusion at a separated attracting equilibrium yields an Ornstein--Uhlenbeck process.  Under the iteration regime in \Cref{prop:null-distribution}, its stationary fluctuations in the contracting directions determine the null law of the gap statistic.  Technical justification and verification of this localized approximation are provided in Section~\ref{app:diffusion-approximation} of the Supplementary Materials.

We impose the following regularity conditions for the detection analysis.  Under $H_0$, $\varepsilon_t=y_t-x_t^\top\theta^{\mathrm{true}}$; under $H_1$, $\varepsilon_t=y_t-x_t^\top\theta_{j_t}^{\mathrm{true}}$ as in~\eqref{eq:dgp}.

\begin{assumption}\label{ass:diffusion}
	Under the hypothesis being considered:		
	(a)~The pairs $(x_t,\varepsilon_t)$ are i.i.d.\ under $H_0$, and the triples $(x_t,j_t,\varepsilon_t)$ are i.i.d.\ under $H_1$.  Moreover, $\E[x_t]=0$ and $\Sigma:=\E[x_tx_t^\top]\succ0$.  Under $H_0$, $\varepsilon_t$ is independent of $x_t$.  Under $H_1$, $\varepsilon_t$ is independent of $(x_t,j_t)$.	
	(b)~The noise $\varepsilon_t$ has a distribution symmetric about zero with a bounded density $f_\varepsilon$ that is continuous at zero, $f_\varepsilon(0)>0$, and $\E[\varepsilon_t^2]=\sigma_\varepsilon^2<\infty$.	
	(c)~$\E[\varepsilon_t^2 \|x_t\|^2] < \infty$ and $\E[\|x_t\|^4] < \infty$.	
	(d)~Under $H_1$, $\inf_{b\in\R^d}\E[\{x_t^\top(\theta_{j_t}^{\mathrm{true}}-b)\}^2]>0$.
	(e)~Under either $H_0$ or $H_1$, the plug-in nuisance estimators $(\hat\sigma_\varepsilon,\hat\Sigma)$ of $(\sigma_\varepsilon,\Sigma)$ satisfy $\hat\sigma_\varepsilon\xrightarrow{p}\sigma_\varepsilon$ and $\|\hat\Sigma-\Sigma\|_{\mathrm{op}}\xrightarrow{p}0$.
\end{assumption}

In \Cref{ass:diffusion}, Part~(a) centers the marginal covariate distribution while allowing dependence between $x_t$ and $j_t$; accordingly, the conditional assignment probabilities $\pi_z(x):=\Pr(j_t=z\mid x_t=x)$ may vary with $x$.  The analysis is stated for the centered covariates used by the detector, while intercept heterogeneity is treated separately below.  Parts~(b)--(c) support the explicit null equilibrium and the localized diffusion argument, as verified in the Supplementary Materials.  Part~(d) requires the mixture predictors to retain a positive mean-square separation from every homogeneous linear predictor on the realized covariate distribution.  The condition is analogous to the support condition in \citet{kasahara2015testing}, which requires distinct regression parameters to yield different linear predictors with positive probability.  Part~(e) requires consistency of the plug-in nuisance estimators; their identification under $H_1$ is discussed with \Cref{alg:test} and \Cref{sec:scale-estimation}.



\section{ Theoretical Results}\label{sec:detection}

In this section, we analyze SGD with the competitive loss and formalize observations from \Cref{fig:detection-example}.  We first characterize the stable equilibrium under $H_0$, derive the null distribution of the gap statistic, and establish the asymptotic validity of the resulting test with proper initialization.  We then collect the dynamics and power analysis under $H_1$.  We present the main results and relegate proofs to the Supplementary Materials.

 Define $c_0=\E[|\varepsilon|]\sqrt{2/\pi}/\sigma_\varepsilon$, $\alpha_\varepsilon=2f_\varepsilon(0)\E[|\varepsilon|]$, and $\alpha_x=\E[|x|]\E[|x|^3]/(\E[x^2])^2$ for a scalar covariate. We first present the equilibrium results for the ODE.

\begin{theorem}\label{thm:eqm}
Under $H_0$ and Parts~(a)--(c) of \Cref{ass:diffusion}:

	(i) In the scalar case, the antisymmetric point $e_1^* = -e_2^* =e^*$, where
	$
	 e^*=  {\E[|\varepsilon|]\,\E[|x|]}/{\E[x^2]},
	$
	is an equilibrium.  If $\alpha_\varepsilon\alpha_x<1$, it is locally exponentially stable.

	(ii) In the multivariate case with $x \sim N(0, \Sigma)$, any antisymmetric point $e_1^* = -e_2^* = e^*$ with $e^* \in \mathcal{E}_2$ is an equilibrium, where
	$
		\mathcal{E}_2 = \{e \in \R^d : \|e\|_\Sigma = c_0\,\sigma_\varepsilon\}
	$
	is a Mahalanobis ellipsoid.  If $\alpha_\varepsilon < \pi/4$, the full-state equilibrium set $\mathcal M_2=\{(e,-e):e\in\mathcal E_2\}$ is locally exponentially stable.
\end{theorem}

\Cref{thm:eqm} characterizes the locally stable equilibria in error coordinates. Specifically, when the covariate is scalar, \Cref{thm:eqm}(i) gives the equilibrium $e_1^*=-e_2^*=e^*$, where $e^*={\E[|\varepsilon|]\,\E[|x|]}/{\E[x^2]}$.  Gaussian data satisfy the stability condition, with $\alpha_\varepsilon\alpha_x=8/\pi^2\approx0.81$.  When the covariates are multivariate, \Cref{thm:eqm}(ii) reveals a remarkable structure: the equilibrium set $\mathcal{E}_2$ is a Mahalanobis ellipsoid $\|e\|_\Sigma = c_0\,\sigma_\varepsilon$ rather than a single point.  The particular equilibrium reached may depend on initialization, but every point in the set has the same Mahalanobis radius.

Based on \Cref{thm:eqm}, we return to the SGD estimators through $\theta_k=\theta^{\mathrm{true}}-e_k$.  Any error equilibrium $(e^*,-e^*)$ corresponds to the estimator locations $\theta_1^*=\theta^{\mathrm{true}}-e^*$ and $\theta_2^*=\theta^{\mathrm{true}}+e^*$.  Thus, the constant-step iterates fluctuate around two locations symmetric about $\theta^{\mathrm{true}}$.  Their equilibrium gap is $\theta_2^*-\theta_1^*=2e^*$ in the scalar case, while the multivariate normalized gap is $\|\theta_1^*-\theta_2^*\|_\Sigma/\sigma_\varepsilon=2\|e^*\|_\Sigma/\sigma_\varepsilon=2c_0$ for every $e^*\in\mathcal E_2$.  This explains the separation shown in \Cref{fig:detection-example}.
The stability has a simple force-balance interpretation.  Winner-take-all competition creates an outward push: once the two pieces differ slightly, each preferentially wins observations whose residual shocks favor its current side, and these selective updates reinforce the separation.  Ordinary regression creates an inward pull toward the common true parameter $\theta^{\mathrm{true}}$.  When the pieces are closer than the equilibrium radius, the competitive push dominates and drives them apart; when they are farther away, the regression pull dominates and draws them back.  At the equilibrium radius, the two forces balance and sustain a stable, nonzero separation.


\begin{remark}
	The coincident state $e_1=e_2=0$ is also an equilibrium, but it is unstable; see the proof of \Cref{thm:eqm} in the Supplementary Materials.  The detector uses the separated locally stable equilibria characterized above.
\end{remark}

Based on \Cref{thm:eqm}, define the normalized gap statistic as
\begin{equation}\label{eq:gap-statistic}
	S_T:=\frac{\|\theta_{1,T}-\theta_{2,T}\|_\Sigma}{\sigma_\varepsilon},
\end{equation}
which measures the separation between the two SGD estimators, $\theta_{1,T}$ and $\theta_{2,T}$, after $T$ iterations.  Because each SGD update depends on random covariates and errors, $S_T$ is a random variable that fluctuates around its equilibrium value from ex ante perspective.  We next characterize its asymptotic distribution under $H_0$.  Define $\kappa_\varepsilon=\E[|\varepsilon|]/\sigma_\varepsilon$ and, for a scalar covariate, $\kappa_x=\E[|x|]/\sqrt{\E[x^2]}$.  

\begin{theorem}\label{prop:null-distribution}
	Under $H_0$ and Parts~(a)--(c) of \Cref{ass:diffusion}, suppose the initialization is constructed independently of the detection stream and converges to a point $(0,-c)$ in midpoint--half-gap error coordinates for a fixed $c\ne0$, as verified for the proposed initialization in \Cref{prop:init}.  Let $S_\infty^{H_0}:=\|\theta_1^*-\theta_2^*\|_\Sigma/\sigma_\varepsilon$ denote the normalized gap at the separated equilibrium selected by the unique mean-ODE trajectory.  As $\eta\to0^+$, let $T\asymp\eta^{-p}$ for some fixed $p\in(1,2)$, so that $\eta T/\log(1/\eta)\to\infty$ and $\eta^2T\to0$.  Then:

	(i) In the scalar case with $\alpha_\varepsilon\alpha_x<1$, $S_\infty^{H_0}=2\kappa_\varepsilon\kappa_x$ and $(S_T-S_\infty^{H_0})/\sqrt{2\eta D/\sigma_\varepsilon^2}\xrightarrow{d}N(0,1)$, where $D=\E[(\varepsilon-e^*x)^2x^2\ind\{\varepsilon x\ge0\}]$ and $e^*$ is given in \Cref{thm:eqm}(i).
 
	(ii) In the multivariate case with $x\sim N(0,\Sigma)$ and $\alpha_\varepsilon<\pi/4$, $S_\infty^{H_0}=2\kappa_\varepsilon\sqrt{2/\pi}$ and $(S_T-S_\infty^{H_0})/\sqrt{\eta V(e^*)}\xrightarrow{d}N(0,1)$, where $e^*=e^*(c)$ is the deterministic equilibrium direction selected by the mean-ODE trajectory from the fixed initialization direction $c$, and $V(e^*)>0$ is the radial variance coefficient given in \eqref{eq:multi-radial-variance} of the Supplementary Materials.  This convergence is unconditional.
	
	(iii) Consequently, in both cases, $S_T\xrightarrow{p}S_\infty^{H_0}$.
\end{theorem}

\Cref{prop:null-distribution}(i) gives the following asymptotic representation in the scalar case:
$
	S_T=S_\infty^{H_0}
	+\sqrt{ {2\eta D}/{\sigma_\varepsilon^2}}\,Z+o_p(\sqrt\eta).
$
The null equilibrium gap is $S_\infty^{H_0}=2\kappa_\varepsilon\kappa_x$, where $\kappa_x=\E|x|/\sqrt{\E[x^2]}$, and the variance is $2\eta D/\sigma_\varepsilon^2+o(\eta)$.  Thus the stochastic fluctuation is of order $\sqrt\eta$: a smaller step size yields tighter concentration but requires more iterations to reach the equilibrium regime. In the multivariate case, \Cref{prop:null-distribution}(ii) gives, for the deterministic equilibrium direction $e^*(c)\in\mathcal E_2$ selected by the fixed initialization,
$
	S_T=S_\infty^{H_0}
	+\sqrt{\eta \,V(e^*(c))}\,Z+o_p(\sqrt\eta),
$
The leading center is common to all points in $\mathcal E_2$, whereas $V(e^*(c))$ may depend on the deterministic direction selected by the fixed initialization.  If $\Sigma=cI$ is isotropic, this variance coefficient is independent of the initialization direction.

Particularly, when both the covariates and noise are Gaussian, $\kappa_\varepsilon=\E|Z|=\sqrt{2/\pi}$, and the relevant standardized covariate projection has the same absolute mean.  Therefore the leading center in both the scalar and multivariate cases becomes the \textit{competitive splitting constant}
\[
	\frac{4}{\pi}=2\!\left(\E|Z|\right)^2.
\]
This value is invariant to $d$, $\Sigma$, and $\sigma_\varepsilon$ within the Gaussian model.  However, it is not distribution-free: outside the Gaussian setting, the leading center generally depends on the standardized absolute moments of the noise and covariates.

\begin{remark} 
	For a general centered multivariate covariate distribution, $S_\infty^{H_0}$ need not have a direction-independent closed-form expression.  Any antisymmetric equilibrium $e^*$ must satisfy
$
		\|e^*\|_\Sigma^2=\E|\varepsilon|\,\E|{e^*}^\top x|,
$
	and hence the corresponding normalized equilibrium gap is
$
	\frac{2\|e^*\|_\Sigma}{\sigma_\varepsilon}
	=2\kappa_\varepsilon
	\frac{\E|{e^*}^\top x|}{\|e^*\|_\Sigma}.
$
	In general, this value depends on both the covariate distribution and the selected equilibrium direction.  Gaussian covariates are special because the standardized absolute projection moment is the same in every direction.
\end{remark}

\Cref{thm:eqm} characterizes the antisymmetric equilibria, but whether the SGD iterates approach one of them depends on the initialization, since the competitive loss is nonconvex. The following result shows that a suitable initialization places the iterates in the basin of an antisymmetric equilibrium, so the regime of \Cref{thm:eqm} is reached.

\begin{proposition} \label{prop:init}
	Under $H_0$ and Parts~(a)--(c) of \Cref{ass:diffusion}, let $\hat\theta_0$ be computed from a preliminary sample independent of the detection stream and satisfy
	$\hat\theta_0\xrightarrow{p}\theta^{\mathrm{true}}$, and set
	$\theta_{1,0}=\hat\theta_0+c$ and
	$\theta_{2,0}=\hat\theta_0-c$ for a fixed
	$c\in\R^d\setminus\{0\}$.  With probability tending to one, this starting point lies in the basin of an antisymmetric equilibrium of \Cref{thm:eqm}: in the scalar case, one of the two antisymmetric equilibrium points in part~(i) when $\alpha_\varepsilon\alpha_x<1$; in the multivariate case, the set $\mathcal M_2$ of part~(ii) when $x\sim N(0,\Sigma)$ and $\alpha_\varepsilon<\pi/4$.
\end{proposition}


\begin{algorithm}[ht]
	\caption{Detection Test} \footnotesize
	\label{alg:test}
	\begin{algorithmic}[1]
		\REQUIRE Independent preliminary sample $\mathcal D_0$ of size $n_0\asymp T$; detection stream $\{(y_t,x_t)\}_{t\in[T]}$; step size $\eta>0$; significance level $\alpha$
		\STATE Using $\mathcal D_0$, estimate the means and center the detection stream; compute $\hat\theta_0$, $\hat\Sigma$, preliminary residuals $\hat\varepsilon_s$, and $\hat\sigma_\varepsilon$
		\STATE Initialize $\theta_{1,0} = \hat\theta_0 + c$, $\theta_{2,0} = \hat\theta_0 - c$ for a fixed $c\in\R^d\setminus\{0\}$
		\STATE Run \Cref{alg:sgd} to obtain $\theta_{1,T}, \theta_{2,T}$; compute plug-in statistic $\hat S_T$ shown in \eqref{eq:hst}
		\STATE Form the fitted upper null quantile $\hat q_{1-\alpha,\eta}$ using the calibration procedure in \Cref{sec:scale-estimation}
		\STATE \textbf{Reject} $H_0$ if $\hat S_T > \hat q_{1-\alpha,\eta}$
	\end{algorithmic}
\end{algorithm}

Together, \Cref{thm:eqm} and \Cref{prop:init} motivate a one-sided test based on the observed gap between the two fitted pieces: under $H_0$ the iterates settle at an antisymmetric equilibrium with a calibrated separation, so a systematically larger gap is evidence against homogeneity.  

\Cref{alg:test} operationalizes this logic, and its main steps correspond directly to the preceding results.  From an independent preliminary sample $\mathcal D_0$ of size $n_0\asymp T$, it first computes a consistent preliminary estimator $\hat\theta_0$ together with the plug-in covariance $\hat\Sigma$ and noise scale $\hat\sigma_\varepsilon$ (Step~1); it then forms the antisymmetric initialization $\theta_{1,0}=\hat\theta_0+c$ and $\theta_{2,0}=\hat\theta_0-c$ of \Cref{prop:init}, which places the null ODE in the basin of an antisymmetric equilibrium (Step~2); and it runs the SGD algorithm in \Cref{alg:sgd} on the detection stream, whose null terminal distribution is characterized in \Cref{prop:null-distribution} (Step~3).  The test rejects for large values of the plug-in gap statistic
\begin{equation}\label{eq:hst}
\hat S_T = \frac{\|\theta_{1,T} - \theta_{2,T}\|_{\hat\Sigma}}{\hat\sigma_\varepsilon},
\end{equation}
the empirical counterpart of the normalized equilibrium gap.  To turn $\hat S_T$ into a calibrated decision, we first estimate the upper null quantile $q_{1-\alpha,\eta}$ by $\hat q_{1-\alpha,\eta}$ (Step~4; see \Cref{sec:scale-estimation}).  We then reject $H_0$ when $\hat S_T$ exceeds $\hat q_{1-\alpha,\eta}$ (Step~5).  The following result establishes asymptotic size control for this test.

\begin{proposition}\label{cor:CI}
	Consider the test in \Cref{alg:test} under $H_0$, initialized as in \Cref{prop:init}, and suppose the conditions of \Cref{prop:null-distribution} hold as $\eta\to0^+$, with $T\asymp\eta^{-p}$ for some fixed $p\in(1,2)$.  Let $\hat\mu_x$ and $\hat\mu_y$ be the preliminary-sample means used to center the detection stream, let $\hat\sigma_\varepsilon$ and $\hat\Sigma$ be plug-in estimators of $\sigma_\varepsilon$ and $\Sigma$, and let $\hat q_{1-\alpha,\eta}$ be the fitted counterpart of the unconditional oracle $(1-\alpha)$ null quantile $q_{1-\alpha,\eta}$ of $S_T$.  Suppose $\|\hat\mu_x\|$, $|\hat\mu_y|$, $\hat\sigma_\varepsilon-\sigma_\varepsilon$, $\|\hat\Sigma-\Sigma\|_{\mathrm{op}}$, and $\hat q_{1-\alpha,\eta}-q_{1-\alpha,\eta}$ are all $o_p(\sqrt\eta)$, and additionally
	$(\|\hat\mu_x\|+|\hat\mu_y|)\eta T=o_p(1)$.  If the null law of $S_T$ is continuous, then $\Pr(\textup{reject}\mid H_0)\to\alpha$.
\end{proposition}

\Cref{cor:CI} establishes asymptotic size control under explicit estimator-rate and distributional conditions.  Specifically, root-$n_0$-consistent sample means and estimators of $\sigma_\varepsilon$, $\Sigma$, and the other nuisance moment functionals satisfy the required $o_p(\sqrt\eta)$ rates because $n_0\asymp T$ and $T\eta\to\infty$ imply $n_0^{-1/2}=o(\sqrt\eta)$.  They also satisfy the accumulated-centering condition because $n_0\asymp T$ and $p<2$ give $n_0^{-1/2}\eta T\asymp\eta\sqrt T\to0$.  The required rate for $\hat q_{1-\alpha,\eta}$ depends on the chosen calibration method and is discussed in \Cref{sec:scale-estimation}.  Continuity of the finite-$\eta$ null law is imposed separately in the proposition.  The proof shows that the scalar and multivariate limits in \Cref{prop:null-distribution} imply the corresponding oracle-quantile expansion directly; in the multivariate case, the fixed initialization selects the deterministic direction $e^*(c)$, so no conditional-to-unconditional mixture argument is needed.

The preceding results characterize the null dynamics and calibration.  We now study the test under $H_1$, where the relevant long-run regime is a \textit{separated attracting equilibrium} whose midpoint is not fixed a priori.  For $a^*=(\theta_1^*,\theta_2^*)$, define the \textit{terminal-concentration condition}, as $\eta\to0^+$ and $T\asymp\eta^{-p}$ for fixed $p\in(1,2)$, by
\begin{equation}\label{eq:terminal-concentration}
 \bigl\|(\theta_{1,T}-\theta_1^*,\theta_{2,T}-\theta_2^*)\bigr\|=o_p(1),
\end{equation}
where the norm is the Euclidean product norm.  Condition~\eqref{eq:terminal-concentration} connects the selected attracting equilibrium to the terminal statistic.  Standard local stochastic-approximation theory yields this convergence when $a^*$ is isolated and locally exponentially attracting, the initialization lies in its ODE basin, the drift and noise satisfy local regularity and moment conditions, and a Lyapunov or confinement condition keeps the iterates in that basin with probability tending to one~\citep{kushner2003}.  The regime $\eta T/\log(1/\eta)\to\infty$ removes the deterministic transient, and the constant-step fluctuations are $O_p(\sqrt\eta)=o_p(1)$.  We now compare the alternative gap with its null reference.

\begin{proposition}\label{prop:H1-gap}
	Under $H_1$ and \Cref{ass:diffusion}, no coincident point $(\theta,\theta)$ is a stable equilibrium.  Suppose there is an attracting set $\mathcal A$ whose ODE basin contains the initialization.  At any separated ODE equilibrium $(\theta_1,\theta_2)\in\mathcal A$, define the normalized gap
	$
		S_\infty^{H_1}
		:=\frac{\|\theta_1-\theta_2\|_\Sigma}{\sigma_\varepsilon}.
$
	Then:

	(i) In the scalar case,
	$
		S_\infty^{H_1}
		>S_\infty^{H_0}=2\kappa_\varepsilon\kappa_x.
	$

	(ii) In the multivariate case,
	$S_\infty^{H_1}
	>S_\infty^{H_0}=2\kappa_\varepsilon\sqrt{2/\pi}$ when $x\sim N(0,\Sigma)$.

	(iii) In either setting of parts~(i)--(ii), let $a^*=(\theta_1^*,\theta_2^*)$ be the separated attracting equilibrium selected by the initialization and suppose that the terminal-concentration condition \eqref{eq:terminal-concentration} holds at $a^*$.  Suppose in addition that the fitted null critical value satisfies
	$\hat q_{1-\alpha,\eta}\xrightarrow{p}S_\infty^{H_0}$.  Then
	$
		\Pr(\textup{reject}\mid H_1)\to1.
	$
  
\end{proposition}

Under $H_0$, the normalized gap concentrates at $S_\infty^{H_0}$, which equals $4/\pi$ under Gaussianity, and \Cref{cor:CI} gives asymptotic size $\alpha$.  Under $H_1$, covariate-dependent group assignment changes the distribution of observations across components, while the predictor-separation condition in \Cref{ass:diffusion}(d) makes the oracle component-scale gap strictly exceed its null reference.  At the attracting equilibrium selected by the dynamics, terminal concentration, consistent component-scale estimation, and consistent calibration then give power approaching one.


\subsection{Intercept Heterogeneity Test}
The above results test slope heterogeneity using demeaned data.  However, a slope-only statistic can miss pure intercept heterogeneity: if subgroups share the same slope and differ only in intercept, demeaning removes only the marginal intercept and does not create slope separation.  We therefore consider $H_0^{(\alpha)}:y_t=x_t^\top\beta^{\mathrm{true}}+\alpha^*+\varepsilon_t$ against $H_1^{(\alpha)}:y_t=x_t^\top\beta^{\mathrm{true}}+\alpha_{j_t}^{\mathrm{true}}+\varepsilon_t$, where at least two latent intercepts differ.  Given partial residuals $r_t = y_t - \hat\beta_{\mathrm{OLS}}^\top x_t$, run SGD with the competitive loss on scalar intercepts~$\alpha_k$ and form the feasible intercept-gap statistic
\begin{equation}\label{eq:intercept-stat}
	\hat S_T^{(\alpha)}:=|\alpha_{1,T}-\alpha_{2,T}|/\hat\sigma_\varepsilon.
\end{equation}
This is a fixed-design scalar gap calculation with predictor one, so $\kappa_x=1$ and the null level is $2\kappa_\varepsilon$.  We reject at level $\alpha_{\mathrm I}$ when $\hat S_T^{(\alpha)}$ exceeds the one-sided critical value
\begin{equation}\label{eq:intercept-crit}
	\hat q_{1-\alpha_{\mathrm I},\eta}^{(\alpha)}:=2\hat\kappa_\varepsilon+z_{1-\alpha_{\mathrm I}}\sqrt{\eta(1-\hat\kappa_\varepsilon^2)},
\end{equation}
where $2\hat\kappa_\varepsilon$ is the plug-in null level and $\sqrt{\eta(1-\hat\kappa_\varepsilon^2)}$ is the plug-in diffusion scale from the constant-design calculation.

\begin{corollary}\label{prop:intercept}
	Suppose Parts~(a)--(c) of \Cref{ass:diffusion} hold for the original covariate--noise pair, the distribution of $x_t$ is symmetric about zero, $f_\varepsilon$ is uniformly continuous, and $2f_\varepsilon(0)\E|\varepsilon|<1$.  The centering clause applies to $x_t$, whereas the intercept recursion uses the constant predictor one.  Let the preliminary slope estimator be computed from an independent sample and satisfy $\|\hat\beta_{\mathrm{OLS}}-\beta^{\mathrm{true}}\|=o_p(\sqrt\eta)$.  Conditional on this sample, put $a=\beta^{\mathrm{true}}-\hat\beta_{\mathrm{OLS}}$, $\varepsilon_a=\varepsilon+a^\top x$, $\sigma_a^2=\E(\varepsilon_a^2\mid a)$, and $\kappa_a=\E(|\varepsilon_a|\mid a)/\sigma_a$.  Suppose the feasible residual estimators satisfy $\hat\sigma_\varepsilon-\sigma_a=o_p(\sqrt\eta)$ and $\hat\kappa_\varepsilon-\kappa_a=o_p(\sqrt\eta)$.  As $\eta\to0^+$ with $T\asymp\eta^{-p}$ for fixed $p\in(1,2)$, the test that rejects when $\hat S_T^{(\alpha)}>\hat q_{1-\alpha_{\mathrm I},\eta}^{(\alpha)}$ satisfies $\Pr(\textup{reject}\mid H_0^{(\alpha)})\to\alpha_{\mathrm I}$.  Under $H_1^{(\alpha)}$, if $\hat\sigma_\varepsilon\xrightarrow{p}\sigma_\varepsilon$, $\hat q_{1-\alpha_{\mathrm I},\eta}^{(\alpha)}\xrightarrow{p}2\kappa_\varepsilon$, and \eqref{eq:terminal-concentration} holds at a separated attracting equilibrium, then $\Pr(\textup{reject}\mid H_1^{(\alpha)})\to1$.
\end{corollary}

\Cref{prop:intercept} is the intercept-only counterpart of \Cref{prop:null-distribution}(i) and \Cref{cor:CI}.  Its proof derives the constant-design CLT directly and applies it conditionally to the exact effective-noise representation of the feasible residuals.  Symmetry of the covariate distribution is required only for this intercept extension and is weaker than Gaussianity.  The power statement is conditional on consistent component-noise calibration, which generally requires additional identification under a pure intercept mixture.  The combined procedure sets $\alpha_{\mathrm I}=\alpha/2$ for both the slope and intercept tests, controlling family-wise error at level $\alpha$ by Bonferroni.

\subsection{Estimating the Null Law and Component-Noise Scale}\label{sec:scale-estimation}

To effectively implement \Cref{alg:test}, we need to estimate the null law used for calibration and the component-noise scale used to standardize the gap statistic.  Although the scale also enters the null law, its construction merits separate treatment because it must remain valid under $H_1$ to preserve power.  We discuss these two tasks as follows.

\paragraph*{Estimating the null law and its upper quantile.}
Step~4 of \Cref{alg:test} requires the fitted upper null quantile $\hat q_{1-\alpha,\eta}$.  In the scalar case, we obtain it by substituting preliminary-sample estimates of $\kappa_\varepsilon$, $\kappa_x$, $D$, and $\sigma_\varepsilon$ into the center and variance in \Cref{prop:null-distribution}(i).  This plug-in normal approximation yields a valid fitted quantile when the resulting center error is $o_p(\sqrt\eta)$ and the diffusion-scale estimator is ratio-consistent.  In the multivariate case, the radial variance can depend on the equilibrium direction selected by the fixed initialization and can be difficult to estimate analytically.  We therefore use an empirical-covariate fitted-null bootstrap: draw independent preliminary and detection resamples under the fitted homogeneous model and rerun the complete preliminary-fit--center--initialize--SGD pipeline.  The empirical $(1-\alpha)$ quantile of the resulting gap statistics is $\hat q_{1-\alpha,\eta}$.  This procedure preserves the direction-dependent projection ratio and finite-step radial effects.  The number of bootstrap replications is chosen so that the Monte Carlo quantile error is negligible relative to $\sqrt\eta$.  A parametric Gaussian Monte Carlo calibration is valid only when a Gaussian covariate model is justified.

\paragraph*{Estimating the component-noise scale.}
The OLS residual variance consistently estimates $\sigma_\varepsilon^2$ under $H_0$; under $H_1$, it also captures between-component variation.  We isolate the component noise through a variance function regression of squared residuals on the covariates \citep{davidian1987variance,chaganty2013spectral}. 
The construction has three steps.  First, cross-fit the pooled regression, whose population coefficient under the general alternative is $\theta_{\mathrm{pool}}=\Sigma^{-1}\E[xy]=\Sigma^{-1}\E[x x^\top\theta_j^{\mathrm{true}}]$, and form the cross-fitted residuals $\hat R_t=y_t-\hat\theta_{\mathrm{pool}}^{(-t)\top}x_t$, where $\hat\theta_{\mathrm{pool}}^{(-t)}$ is computed without observation~$t$.  Second, fit the conditional second moment of these residuals.  The quadratic specification used under constant assignment probabilities is
\begin{equation}\label{eq:quadratic-variance-regression}
	\hat R_t^2=a+\langle B,x_tx_t^\top\rangle+u_t.
\end{equation}
Third, take the fitted intercept as the estimator, $\hat\sigma_\varepsilon^2:=\hat a$; we refer to $\hat\sigma_\varepsilon$ as the \textit{component-scale estimator}.  Writing $R=y-\theta_{\mathrm{pool}}^\top x$ for the population pooled residual, covariate-dependent assignment gives the variance function
\begin{equation}\label{eq:dependent-variance-function}
\E[R^2\mid x]
=\sigma_\varepsilon^2
+\sum_{z=1}^{K^*}\pi_z(x)
\{x^\top(\theta_z^{\mathrm{true}}-\theta_{\mathrm{pool}})\}^2.
\end{equation}
The component-noise variance is the constant term in~\eqref{eq:dependent-variance-function}.  When the assignment probabilities are constant in $x$, the second term reduces to $x^\top\Omega_\theta x$, where $\Omega_\theta=\operatorname{Var}(\theta_j^{\mathrm{true}})$, and the quadratic regression in~\eqref{eq:quadratic-variance-regression} is exact.  With covariate-dependent assignment, the second term is a weighted quadratic function that vanishes at $x=0$.  A flexible variance-function estimator identifies $\sigma_\varepsilon^2$ from this constant under continuity and local-support conditions at zero, providing the consistent estimator in \Cref{ass:diffusion}(e) used for power.  Under $H_0$, the variance function is constant and the ordinary root-$T$ residual-scale estimator has error $o_p(\sqrt\eta)$ when $T\eta\to\infty$, as required by \Cref{cor:CI}.

\paragraph*{Pooled-residual fallback.}
When this identification condition fails, so that the variance function
regression is degenerate and $\hat\sigma_\varepsilon^2$ is unavailable, we
fall back to an alternative statistic built directly from the pooled
residuals.  Define
$\theta_{\mathrm{pool}}:=\Sigma^{-1}\E[xy]$ and
$\sigma_{\mathrm{pool}}^2:=\E[(y-\theta_{\mathrm{pool}}^\top x)^2]$, and form
the pooled-residual statistic
\begin{equation}\label{eq:pooled-stat}
	\hat S_T^{\mathrm{pool}}
	:=\frac{\|\theta_{1,T}-\theta_{2,T}\|_{\hat\Sigma}}{\hat\sigma_{\mathrm{pool}}},
\end{equation}
where $\hat\sigma_{\mathrm{pool}}$ is the ordinary OLS-residual estimator.  We
reject at level $\alpha$ when $\hat S_T^{\mathrm{pool}}$ exceeds the fitted
pooled-residual critical value $\hat q_{1-\alpha,\eta}^{\mathrm{pool}}$, the
counterpart of \Cref{cor:CI} with $\sigma_\varepsilon$ replaced by the pooled
scale $\sigma_{\mathrm{pool}}$.  Under $H_0$,
$\sigma_{\mathrm{pool}}=\sigma_\varepsilon$, so this statistic preserves the
size result of \Cref{cor:CI}; under $H_1$, the pooled scale may be inflated,
so power is no longer automatic and instead requires the strict-margin
condition of the following corollary.

\begin{corollary}\label{cor:pooled-power}
	Suppose Parts~(a)--(d) of \Cref{ass:diffusion} hold under $H_1$, and the terminal-concentration
	condition \eqref{eq:terminal-concentration} holds at the separated attracting
	equilibrium $(\theta_1^*,\theta_2^*)$ selected by the initialization.  Define
	$S_\infty^{\mathrm{pool}}:=\|\theta_1^*-\theta_2^*\|_\Sigma/\sigma_{\mathrm{pool}}$,
	and suppose $\hat\sigma_{\mathrm{pool}}\xrightarrow{p}\sigma_{\mathrm{pool}}$,
	$\hat\Sigma\xrightarrow{p}\Sigma$, and
	$\hat q_{1-\alpha,\eta}^{\mathrm{pool}}\xrightarrow{p}q_\infty^{\mathrm{pool}}$.
	If $S_\infty^{\mathrm{pool}}>q_\infty^{\mathrm{pool}}$, then the test that
	rejects when $\hat S_T^{\mathrm{pool}}>\hat q_{1-\alpha,\eta}^{\mathrm{pool}}$
	satisfies $\Pr(\textup{reject}\mid H_1)\to1$.
\end{corollary}

For the pooled-residual statistic, consistency requires the strict margin in
\Cref{cor:pooled-power}.  For a given application, this margin can be assessed
by estimating the selected pooled-normalized equilibrium and its
fitted-null benchmark, with a row bootstrap used to check that the lower
confidence bound for their difference is positive.  This verification is a
power certificate; null validity continues to follow from the consistency of
the pooled-residual scale under $H_0$.

\paragraph*{Impact of Higher-Order Competitive Loss.}
The detection theory above uses two competing pieces.  Adding a third piece does not uniformly improve power.  In the scalar Gaussian setting, the $K=3$ dynamics converge under $H_0$ to a symmetric equilibrium with normalized outer gap $2c_3$, where $c_3\approx1.320$.  For all sufficiently small positive balanced separations, the selected locally attracting alternative equilibrium instead has a strictly smaller outer gap, so a correctly calibrated outer-gap test is asymptotically blind in this local regime (\Cref{prop:K3-scalar} in the Supplementary Materials).  This structural failure persists even when $\sigma_\varepsilon$ is known and justifies retaining the two-piece loss in \Cref{alg:test}; \Cref{sec:K-general} of the Supplementary Materials gives the proof, simulations, and a $K=4$ comparison.

\section{Extensions and Robustness}\label{sec:extensions}

\subsection{Detection for Partial Mixture}
In many applications, domain knowledge localizes where heterogeneity can reside: most covariates are believed to affect the response uniformly, and only one focal covariate, typically the exposure or treatment variable of interest, may carry a group-dependent coefficient driven by an unobserved confounder. Next, we show that our approach can be modified to handle detection in this setting.  Fix a focal coordinate $\ell\in[d]$, partition $x_t=(x_{-\ell,t}^\top,x_{\ell,t})^\top$, and consider the partial mixture model
\begin{equation}\label{eq:partial-mixture}
y_t=\psi^\top x_{-\ell,t}+\theta_{j_t}\,x_{\ell,t}+\varepsilon_t,
\end{equation}
where $\psi\in\R^{d-1}$ is common to all latent groups and only the coefficient of $x_{\ell,t}$ may be mixed.  Let
$B_\ell:=\E(x_{-\ell}x_\ell)/\E(x_\ell^2)$ and $z_{-\ell}:=x_{-\ell}-B_\ell x_\ell$.  Write $\vartheta_j:=\theta_j+\psi^\top B_\ell$ and let $\bar\vartheta_w:=\E(x_\ell^2\vartheta_j)/\E(x_\ell^2)$ denote the focal coefficient targeted by the transformed regression.  On an independent preliminary sample, estimate $B_\ell$, use the resulting $\hat z_{-\ell}$ to regress $y$ jointly on $(\hat z_{-\ell},x_\ell)$ and estimate the common coefficient $\psi$ and focal center, and freeze these estimates.  On the detection stream, form $\hat z_{-\ell,t}=x_{-\ell,t}-\hat B_\ell x_{\ell,t}$ and $\tilde y_t=y_t-\hat\psi^\top\hat z_{-\ell,t}$, then apply the scalar test to $(x_{\ell,t},\tilde y_t)$.  
The following result establishes the effectiveness of this detection stream:

\begin{proposition} \label{prop:partial}
Suppose the partial-mixture model in~\eqref{eq:partial-mixture} satisfies \Cref{ass:diffusion}, with $f_\varepsilon$ uniformly continuous.  Suppose $z_{-\ell}$ is independent of $(x_\ell,j)$, has a distribution symmetric about zero, and satisfies $\E\|z_{-\ell}\|^4<\infty$.  From an independent sample of size $n_0\asymp T$, suppose $\|\hat\psi-\psi\|=o_p(\sqrt\eta)$, $\|\hat B_\ell-B_\ell\|=o_p(1)$, and the preliminary focal estimator is $o_p(\sqrt\eta)$-consistent for $\bar\vartheta_w-\psi^\top(B_\ell-\hat B_\ell)$, the focal coefficient in the regression using $\hat z_{-\ell}$.  As $\eta\to0^+$ with $T\asymp\eta^{-p}$ for fixed $p\in(1,2)$:

	(i) Under $H_0$, if the conditions of \Cref{cor:CI} hold for the resulting feasible scalar model, then $\Pr(\textup{reject}\mid H_0)\to\alpha$.

	(ii) Under $H_1$, if the scalar gap condition of \Cref{prop:H1-gap}(i), terminal concentration \eqref{eq:terminal-concentration} for the feasible recursion, and consistent matched-null calibration hold, then $\Pr(\textup{reject}\mid H_1)\to1$.
\end{proposition}

\subsection{Robustness to Endogeneity}
The preceding analysis assumes $\varepsilon_t$ is independent of $x_t$.  Here we allow linear endogeneity, $\mathrm{Cov}(x_t,\varepsilon_t)=\gamma\ne0$, under joint Gaussianity.  Specifically, consider $y_t=x_t^\top\theta_{j_t}^{\mathrm{true}}+\varepsilon_t$ for general $d\ge1$, where the component coefficients are identical under $H_0$ and at least two are distinct under $H_1$.  Retain \Cref{ass:diffusion} except for its noise--covariate independence clause, which is replaced by $x_t\sim N(0,\Sigma)$, $\Sigma\succ0$, $\varepsilon_t\sim N(0,\sigma_\varepsilon^2)$, and $\mathrm{Cov}(x_t,\varepsilon_t)=\gamma$, where $\gamma^\top\Sigma^{-1}\gamma<\sigma_\varepsilon^2$; under $H_1$, the latent group remains independent of $(x_t,\varepsilon_t)$.  Define the projected residual $\tilde\varepsilon_t:=\varepsilon_t-\gamma^\top\Sigma^{-1}x_t$ and its variance $\tilde\sigma_\varepsilon^2:=\sigma_\varepsilon^2-\gamma^\top\Sigma^{-1}\gamma$.

\begin{proposition}
\label{prop:endogeneity}
	Under the endogeneity setup above:

	(i) Under $H_0$, the conclusions of \Cref{prop:null-distribution} and \Cref{cor:CI} continue to hold with $(\varepsilon,\sigma_\varepsilon)$ replaced by $(\tilde\varepsilon,\tilde\sigma_\varepsilon)$.  In particular, the Gaussian null center remains $4/\pi$ and $\Pr(\textup{reject}\mid H_0)\to\alpha$.

	(ii) Under $H_1$, if the conditions of \Cref{prop:H1-gap}(iii) hold with $(\varepsilon,\sigma_\varepsilon)$ replaced by $(\tilde\varepsilon,\tilde\sigma_\varepsilon)$, then $\Pr(\textup{reject}\mid H_1)\to1$.
\end{proposition}

\Cref{prop:endogeneity} follows from the common shift
$
	y_t=x_t^\top\{\theta_{j_t}^{\mathrm{true}}+\Sigma^{-1}\gamma\}
	+\tilde\varepsilon_t.
$
Joint Gaussianity makes $\tilde\varepsilon_t$ independent of $x_t$, while the shift $\Sigma^{-1}\gamma$ is shared by every component and therefore leaves their contrasts unchanged.  The relevant scale is thus $\tilde\sigma_\varepsilon$, not the marginal standard deviation $\sigma_\varepsilon$.  If $\bar\theta=\E[\theta_j^{\mathrm{true}}]$ and $\Omega_\theta=\Var(\theta_j^{\mathrm{true}})$, the residual from the pooled BLP satisfies
$
	\E[R^2\mid x]=\tilde\sigma_\varepsilon^2+x^\top\Omega_\theta x.
$
Thus the component-scale estimator continues to isolate the component-noise variance: under $H_0$, $\Omega_\theta=0$, while under $H_1$ the quadratic term absorbs the between-component variation.  Joint Gaussian endogeneity therefore changes only the common coefficient and the relevant noise scale, leaving the normalized null calibration and conditional fixed-alternative consistency unchanged.  As $\gamma^\top\Sigma^{-1}\gamma \to \sigma_\varepsilon^2$, the projected noise variance approaches zero, so the diffusion approximation becomes nearly degenerate and finite-sample calibration may become unstable, although the stated reduction remains valid whenever $\gamma^\top\Sigma^{-1}\gamma < \sigma_\varepsilon^2$.  The synthetic experiment in \Cref{sec:synthetic} finds power $1.00$ in all nine endogeneity settings considered.


\subsection{Effects of Heteroskedasticity}
We next characterize how heteroskedasticity affects the null calibration of the gap statistic $\hat S_T$ in \Cref{prop:null-distribution}(i).  For analytical tractability, consider the scalar model $y_t=x_t\theta_{j_t}+\varepsilon_t$.  Retain \Cref{ass:diffusion} except for its noise--covariate independence clause: under $H_1$, $j_t$ is independent of $(x_t,\varepsilon_t)$, and under both hypotheses $\varepsilon_t\mid x_t$ has mean zero and is symmetric about zero, while its variance may depend on $x_t$.  Write $\sigma_x^2:=\E[x^2]>0$, $\sigma_\varepsilon^2:=\E[\varepsilon^2]>0$, $\kappa_\varepsilon:=\E|\varepsilon|/\sigma_\varepsilon$, and $\kappa_x:=\E|x|/\sigma_x$.

\begin{assumption}\label{ass:heterosk}
	The conditional densities $f_{\varepsilon\mid x}(\cdot\mid x)$ are uniformly bounded and continuous at zero for $P_x$-almost every $x$, and $\mu\mapsto\E[|\mu+\varepsilon|\mid x]$ is uniquely minimized at zero.
\end{assumption}

\begin{proposition}\label{prop:heterosk}
	Under the heteroskedastic scalar model above and \Cref{ass:heterosk}, let $\eta\to0^+$ with $T\asymp\eta^{-p}$ for some fixed $p\in(1,2)$.

	(i) Under $H_0$, suppose \eqref{eq:terminal-concentration} holds at the attracting null equilibrium and the homoskedastic reference is consistently fitted.  If $\mathrm{Cov}(|\varepsilon|,|x|)>0$ by a fixed amount, then $\Pr(\textup{reject}\mid H_0)\to1$; if the covariance is negative by a fixed amount, then $\Pr(\textup{reject}\mid H_0)\to0$.

	(ii) Under $H_1$, suppose \eqref{eq:terminal-concentration} holds at a separated attracting equilibrium and the fitted heteroskedastic-null critical value converges to the corresponding null equilibrium gap.  Then $\Pr(\textup{reject}\mid H_1)\to1$.
\end{proposition}

\Cref{prop:heterosk} shows conditional consistency under fixed alternatives when $j_t$ is independent of $(x_t,\varepsilon_t)$, the conditional noise is centered and symmetric with a unique median at zero, \eqref{eq:terminal-concentration} holds at a separated attracting equilibrium, and the critical value consistently estimates the matched heteroskedastic-null quantile.
The original homoskedastic calibration may instead be anti-conservative or conservative; even $\operatorname{Cov}(|\varepsilon|,|x|)=0$ aligns only the leading centers and does not establish exact size.
Moreover, under unrestricted heteroskedasticity the variance function regression in \Cref{sec:scale-estimation} does not by itself identify the marginal noise scale under $H_1$, because $\E[R^2\mid x]=\sigma_\varepsilon^2(x)+x^2\Omega_\theta$, where $\sigma_\varepsilon^2(x):=\Var(\varepsilon\mid x)$.  Feasible matched calibration therefore requires additional structure or a separately justified estimator of the heteroskedastic null law.

A related calibration issue arises under quadratic model misspecification: when the data follow a quadratic nonlinearity, the linear model optimized by SGD leaves an $x$-dependent residual.  After demeaning ($\tilde{x} = x - \bar{x}$), an omitted even term $\beta\tilde{x}^2$ is orthogonal to $\tilde{x}$ under symmetric covariates (since $\E[\tilde{x}^3]=0$), so it need not bias the linear OLS slope, but it can still alter the gap statistic through residual variation tied to $x$.  The same orthogonality extends to Gaussian multivariate covariates because all third central moments vanish ($\E[\tilde{x}_i\tilde{x}_j\tilde{x}_k]=0$).  Unlike \Cref{prop:heterosk}, the omitted even term induces an $x$-dependent conditional \textit{mean} rather than pure conditional heteroskedasticity, so it is not a literal special case; the precise equilibrium shift is characterized here only for the scalar case.

 In applications, Breusch--Pagan and Ramsey RESET diagnostics flag these departures when interpreting the test result.

\section{Experiment and Applications}\label{sec:framework}

In this section, we develop a comprehensive set of numerical results that complement the theoretical analysis.  First, we validate the null calibration, detection power, and step-size choice on synthetic data (\S\ref{sec:synthetic}).  Then, we present a three-stage discovery toolkit for real-world datasets to validate known SPs and even identify underreported SPs (\S\ref{sec:framework-desc}).

\subsection{Synthetic Experiment}\label{sec:synthetic}

\paragraph*{Synthetic Experiment Setup.} Throughout this subsection, for $H_0$, we generate data from the linear model $y_t = x_t^\top\theta^{\mathrm{true}} + \varepsilon_t$ with $x_t \sim N(0,\Sigma)$ and $\varepsilon_t \sim N(0,\sigma_\varepsilon^2)$, independently at each round~$t$.  Unless stated otherwise, we set $\sigma_\varepsilon = 1$, $\Sigma = I_d$, and run SGD with constant step size $\eta = 0.005$ for $T = 200{,}000$ iterations over 30 independent seeds.  Under both $H_0$ and $H_1$, we estimate $\sigma_\varepsilon$ by the cross-fitted variance function regression in \Cref{sec:scale-estimation}.  Because the multivariate synthetic designs have diagonal $\Sigma$ and diagonal mixture covariance, the auxiliary regression uses an intercept and $(x_1^2,\ldots,x_d^2)$.  Each experiment below specifies any departures from these defaults.

We validate the null distribution of $S_T$ predicted by \Cref{prop:null-distribution}, both in mean and in tail probability.  For the mean, we run SGD under $H_0$ and vary dimension ($d = 1,2,5,10,20$), noise level ($\sigma_\varepsilon = 0.5,1,2,5$), and covariance structure (identity, diagonal, correlated), using different $\theta^{\mathrm{true}}$ values to confirm invariance.  

\begin{figure}[!t]
	\centering
	\begin{subfigure}[t]{0.485\textwidth}
		\centering
		\includegraphics[width=\linewidth]{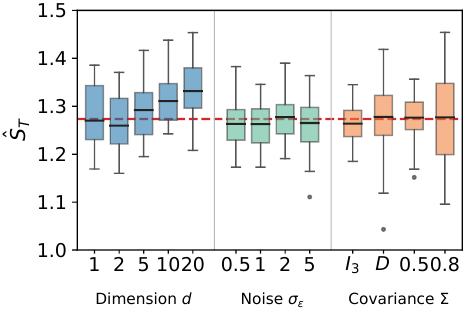}
		\caption{Null calibration}
		\label{fig:synthetic-null}
	\end{subfigure}
	\hfill
	\begin{subfigure}[t]{0.47\textwidth}
		\centering
		\includegraphics[width=\linewidth]{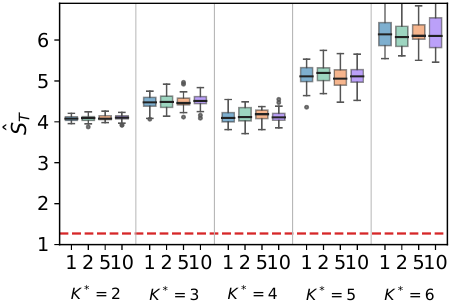}
		\caption{Varying $K^*$ and $d$}
		\label{fig:synthetic-groups}
	\end{subfigure}
	\par\medskip
	\begin{subfigure}[t]{0.47\textwidth}
		\centering
		\includegraphics[width=\linewidth]{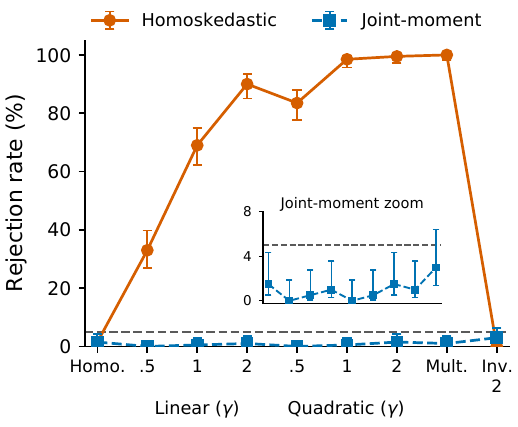}
		\caption{Heteroskedasticity}
		\label{fig:heterosked-rejection}
	\end{subfigure}
	\hfill
	\begin{subfigure}[t]{0.47\textwidth}
		\centering
		\includegraphics[width=\linewidth]{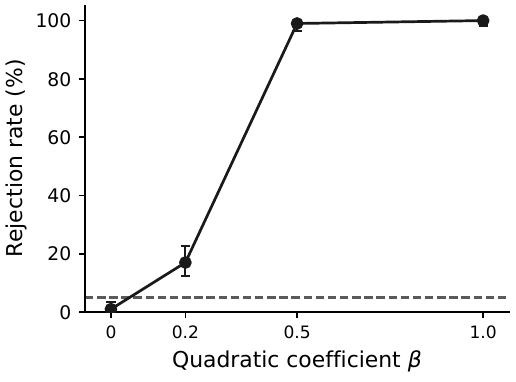}
		\caption{Quadratic misspecification}
		\label{fig:misspec-rejection}
	\end{subfigure}
	\caption{Synthetic calibration, alternative separation, and robustness.  Panels~(a)--(b) show seed-level distributions of $\hat S_T$; dashed lines mark $4/\pi\approx1.273$, and $D=\operatorname{diag}(1,2,5)$ in panel~(a).  Panels~(c)--(d) report null rejection rates over 200 trials per setting; error bars are Wilson 95\% confidence intervals, and dashed lines mark the nominal 5\% level.  Panel~(c) compares homoskedastic and joint-moment calibration.}
	\label{fig:synthetic-robustness}
\end{figure}

Panel~(a) of \Cref{fig:synthetic-robustness} displays the results. In every $H_0$ configuration, the mean $\hat S_T$ falls within 5\% of $4/\pi \approx 1.273$.  The small positive bias for larger $d$ is consistent with finite-step effects omitted by the first-order limit in \Cref{prop:null-distribution}(ii).  The distributions across noise levels and covariance structures remain centered close to $4/\pi$, confirming invariance of the leading center to $\sigma_\varepsilon$ and~$\Sigma$.   Under this fresh-stream design, empirical Type~I error is also close to nominal: at $\alpha=0.05$, scalar closed-form calibration gives a rejection rate of 5.5\% for Gaussian covariates, while setting-level simulation calibration gives 6.0\% for both $d=2$ and $d=5$.  These calculations use the component-scale estimator in every sample.  They do not validate the real-data protocol, which reuses the same finite sample for 1{,}000 passes.  Repeating that protocol under a Gaussian scalar null for 300 independently generated datasets produces rejection rates of 58.7\%, 55.3\%, 46.0\%, and 30.3\% at $n=150$, 333, 777, and 1{,}840, respectively, with Monte Carlo standard errors 2.8, 2.9, 2.9, and 2.7 percentage points.  The matched binary-design check at $n=777$ rejects 2.7\% of the time (standard error 0.9 percentage points).  \Cref{tab:finite-n-reuse} in the Supplementary Materials reports the full audit.  Consequently, repeated-data analytic calibration at these sample sizes is not justified by the current asymptotic theory.

We next examine gap behavior under $H_1$ across $K^*\in\{2,3,4,5,6\}$ and $d\in\{1,2,5,10\}$, placing the signal in the first coordinate when $d>1$.  In \Cref{fig:synthetic-robustness}(b), the component-scale estimate remains close to one (mean $0.987$--$1.006$), whereas the pooled-residual scale ranges from $2.236$ to $3.563$; the resulting mean $\hat S_T$ ranges from $4.078$ to $6.183$ and is nearly invariant to $d$ for fixed $K^*$.  The complementary signal-strength sweep is reported in \Cref{fig:signal-strength} of the Supplementary Materials.

Finally, we assess power and robustness.  At $\Delta=0.5$, power at level $0.05$ is $1.00$ in 12 of 13 dimension--endogeneity settings and $0.99$ in the remaining setting.  Under joint Gaussian endogeneity, projected-noise calibration keeps the scalar mean $\hat S_T$ within $2\%$ of $4/\pi$; Type~I error is $5.2\%$ overall, and multivariate rejection rates range from $2\%$ to $8\%$.

Panels~(c)--(d) of \Cref{fig:synthetic-robustness} report the robustness results.  With linear heteroskedasticity at $\lambda_{\mathrm{het}}=1$, the homoskedastic plug-in rejection rate is $0.69$, versus $0.005$ using the joint-moment null center from \Cref{prop:heterosk}; the inverse-quadratic design instead makes the plug-in test conservative.  Under $H_1$, the mean gap exceeds its matched $H_0$ value in every heteroskedastic design examined.  An omitted quadratic term yields rejection rates of $0.17$, $0.99$, and $1.00$ at $\beta=0.2$, $0.5$, and $1.0$.  Full designs and the endogeneity plot appear in \Cref{app:robustness-experiments} of the Supplementary Materials.

\subsection{A Discovery Toolkit for Real-World Datasets}\label{sec:framework-desc}

To further demonstrate the value of our detection test, we propose a three-stage toolkit, \textit{Detect--Screen--Verify} (DSV), for identifying SPs from real-world datasets as follows:
\begin{enumerate}
	\item[Stage~1.]   Run \Cref{alg:test}.
	\item[Stage~2.]   After a Stage~1 rejection, run SGD on the original (non-demeaned) data with an intercept.  The two fitted pieces induce a binary partition: each observation is assigned to the piece with the smaller squared residual at the fitted parameters.  Then, screen all observed covariates against this partition, using the correlation ratio for continuous variables or Cram\'{e}r's~$V$ for categorical variables, and select the covariate with the strongest association as the candidate confounder.
	\item[Stage~3.]   Stratify the data by the identified confounder and fit within-group regressions; a strict majority of eligible within-group slopes opposite to the pooled slope confirms SP.
\end{enumerate}
 
The stages serve distinct roles.  Stage~1 (``detect'') is an omnibus detector: it tests departure from the calibrated homogeneous model without requiring a candidate confounder or the true number of mixture components.  Its statistic is normalized by the component-scale estimator for continuous focal covariates, and by the pooled-residual scale with its matched calibration for binary focal covariates, for which the variance function regression may be degenerate (\Cref{sec:scale-estimation}).  A rejection alone, however, is not yet a discovery: it signals a departure compatible with latent subpopulations without naming them or establishing a paradox.  Stages~2 and~3 provide this specificity.  Stage~2 (``screen'') exploits a byproduct of the optimizer: the two fitted pieces already partition the observations.  It then links this partition to an observed covariate, turning an anonymous rejection into a candidate confounder.  Stage~3 (``verify'') checks the defining condition of SP, a within-group sign reversal under the selected stratification.  Thus, an end-to-end DSV discovery requires all three stages: rejection, recovery of an interpretable grouping, and verification of the reversal.  Stages~2--3 are diagnostic procedures and do not inherit the inferential guarantees of Stage~1.  Finally, as discussed in \Cref{sec:extensions}, Breusch--Pagan and Ramsey RESET tests on the OLS residuals are used before Stage~1 to flag possible heteroskedasticity and functional-form misspecification.

We apply DSV to eight datasets spanning healthcare, biology, labor, energy, housing, and ecology: four established SP benchmarks and four applications in which, to our knowledge, the reported sign reversals have not previously been documented.  For mixture detection, we compare SGD with the modified penalized-EM homogeneity test of \citet{kasahara2015testing} and conventional EM with BIC model selection~\citep{dempster1977,mclachlan2000finite}.  For SP discovery, we also report Xu-style enumeration~\citep{xu2018detecting} and X-terminal CART screening~\citep{shmueli2018forest}.  These discovery methods receive the nominated focal relationship and observed candidate confounders and therefore are not direct benchmarks for the candidate-free detection stage.  Dataset-level results, sources, implementation details, and mechanisms are reported in \Cref{tab:confounder-all,app:realworld-details} of the Supplementary Materials.

\paragraph*{Results.}
Across the eight datasets, the originally used SGD statistic exceeds its fitted critical value in 8/8 cases; in light of the finite-sample audit, these are computed outputs rather than validated 5\%-level rejections.  The KS detector rejects homogeneity in 7/8 cases, while EM--BIC selects two components in 7/8 cases.  Among the SP-discovery procedures, all of which use observed candidate covariates, DSV Screen top-ranks the reference grouping in 8/8 cases, and the complete DSV procedure recovers and verifies the reference reversal in 8/8 cases.  Xu-style enumeration also surfaces the reference reversal in 8/8 cases, while X-terminal CART does so in 5/8.  These discovery counts address a different task from candidate-free mixture detection.  Dataset-level results and qualifications are provided in \Cref{tab:confounder-all,app:realworld-details} of the Supplementary Materials.

\paragraph*{An Illustrative Example.}
We illustrate how the toolkit finds both a confounder and its confounding mechanism using global soil-nematode data; the dataset source and construction are documented in \Cref{app:realworld-datasets} of the Supplementary Materials.  We ask whether human pressure is negatively associated with bacterivore abundance within ecological regions.  The pooled regression of log bacterivore abundance on human footprint is positive ($+0.0064$), suggesting the opposite.  The component-normalized statistic exceeds its originally fitted threshold ($\hat S_T=1.419$, analytic $p<10^{-10}$; Breusch--Pagan $p=0.51$, RESET $p=0.28$), but the Stage~1 $p$-value is provisional for the repeated-data reason above.  The subsequent descriptive stages identify WWF biome as the confounder: after stratifying, the slope is negative in 7 of 9 biomes (\Cref{fig:nematode-gap,fig:nematode-slopes}).  The mechanism is that biomes differ in both baseline abundance and typical human footprint; a pooled regression conflates these distinct ecological regimes, and the between-biome variation overwhelms the within-biome effect (\Cref{fig:nematode-pooled-map,fig:nematode-biome-map}).  Without stratification, the pooled positive sign could mislead analyses of human impact on soil biodiversity; the within-biome negative association is consistent with the ecological expectation that land-use disturbance is associated with lower bacterivore abundance.

 \begin{figure}[!t]
	\centering
	\begin{subfigure}{0.45\textwidth}
		\centering
		\includegraphics[width=\linewidth]{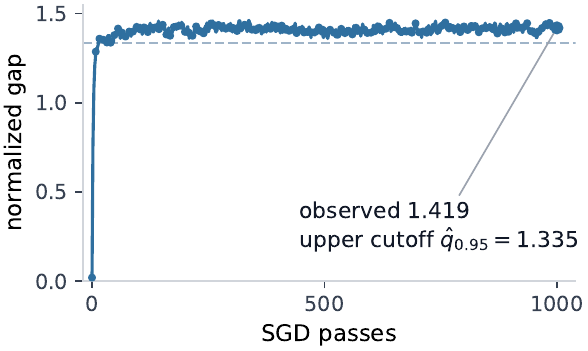}
		\caption{Component-normalized gap and upper cutoff}
		\label{fig:nematode-gap}
	\end{subfigure}
	\qquad
	\begin{subfigure}{0.45\textwidth}
		\centering
		\includegraphics[width=\linewidth]{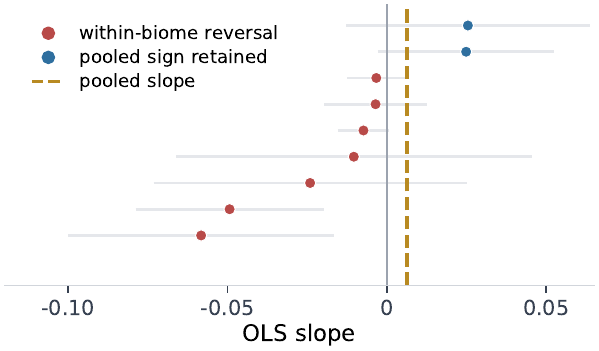}
		\caption{Pooled and biome-specific slopes}
		\label{fig:nematode-slopes}
	\end{subfigure}
	\par\medskip
	\begin{subfigure}{0.45\textwidth}
		\centering
		\includegraphics[width=\linewidth]{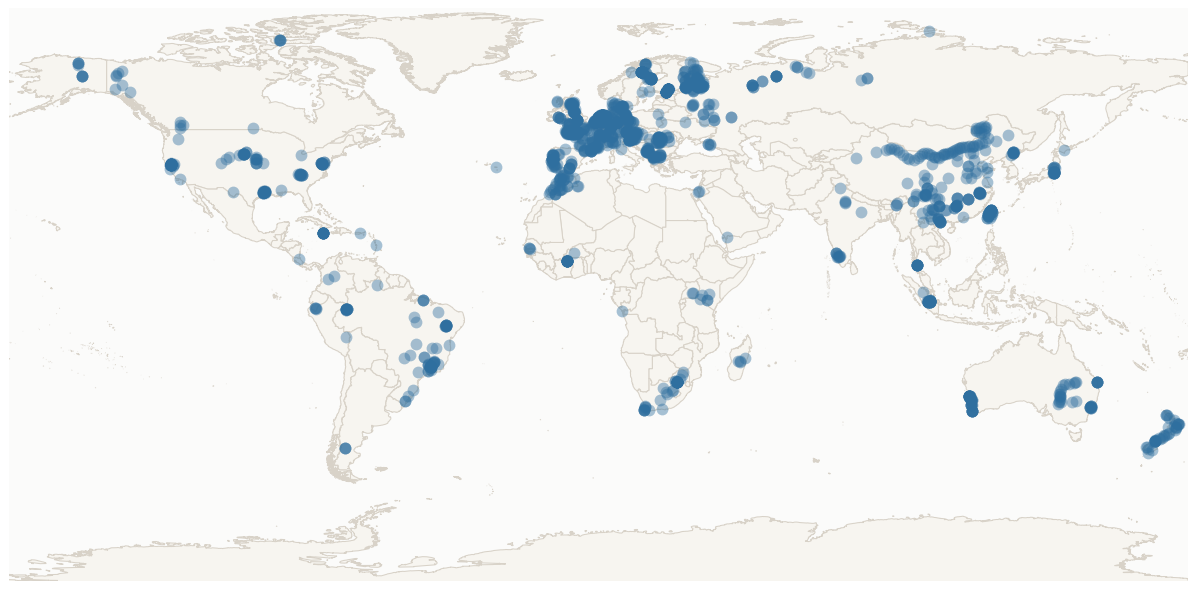}
		\caption{Pooled spatial view}
		\label{fig:nematode-pooled-map}
	\end{subfigure}
	\qquad
	\begin{subfigure}{0.45\textwidth}
		\centering
		\includegraphics[width=\linewidth]{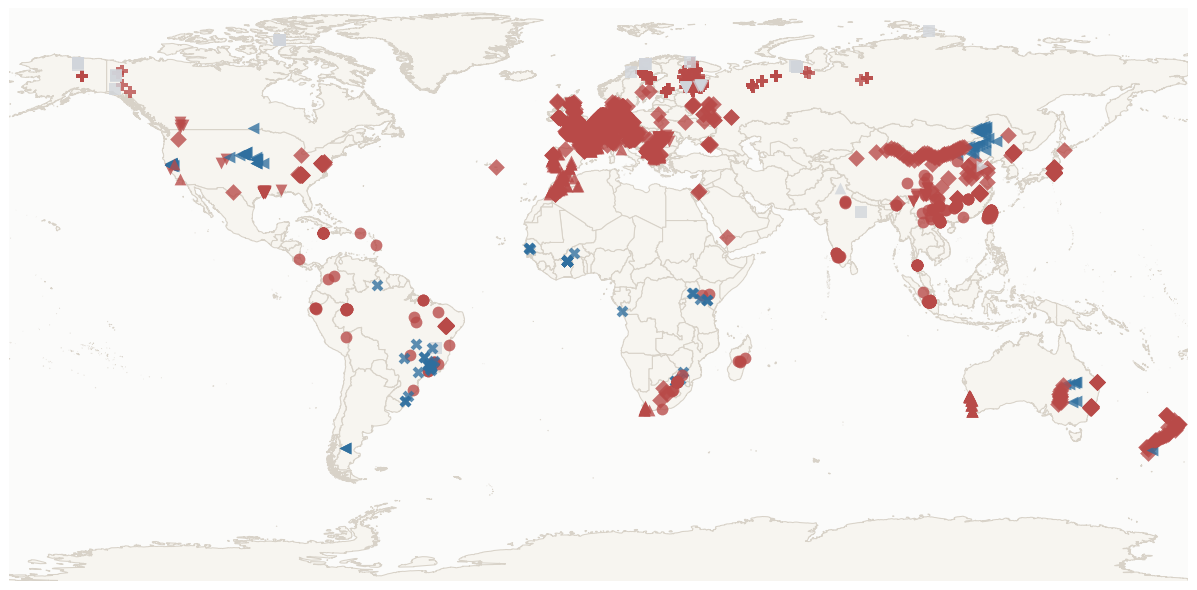}
		\caption{Biome-stratified spatial view}
		\label{fig:nematode-biome-map}
	\end{subfigure}
	\caption{Detection and interpretation for the soil-nematode example.  (a)~The SGD gap exceeds the fitted 5\% upper critical value.  (b)~Seven of nine biome-specific slopes are negative despite a positive pooled slope; points and bars show estimates and 95\% confidence intervals, and the dashed gold line shows the pooled slope.  (c)--(d)~WWF-biome stratification reveals distinct regimes (blue: pooled sign retained; red: reversed).}
	\label{fig:realworld-applications}
\end{figure}

This example also makes the distinction between detection and discovery concrete.  KS rejects homogeneity for the nominated focal regression, providing evidence of mixture structure without identifying WWF biome or establishing the reversal.  DSV recovers WWF biome from the SGD partition and verifies the pooled-to-within sign change.  In an unguided scan in which neither the focal relationship nor WWF biome is supplied, the Xu categorical enumeration contains the correct human-footprint--bacterivore--biome triplet, but ranks it only 69th among 18{,}960 candidates; a quantile-based extension accommodating continuous variables ranks it 18{,}110th among 774{,}200 candidates.  It therefore recovers the pattern only after producing many higher-ranked alternatives rather than surfacing it as a leading discovery.  The automatic tree scan does not recover the relevant triplet among its top 80 candidate focal relationships.  Implementation details and the unguided-search convention for these comparisons are given in \Cref{app:implementation}.

The descriptions, focal variables, candidate confounders, and confounding mechanisms for the remaining datasets are shown in \Cref{app:realworld-details}.


\section{Concluding Remarks}\label{sec:conclusion}
This work shows that SGD with competitive loss and a constant step size can serve as a detector of latent heterogeneity.  We establish a universality result for nonconvex SGD dynamics: winner-take-all competition spontaneously performs unsupervised mixture detection.  In the Gaussian design-and-noise setting, the equilibrium is governed by $4/\pi$, a dimensionless universal quantity arising from the half-normal mean.  Building on this mechanism, the DSV toolkit combines detection, screening, and verification to identify and validate SP in real-world datasets.  These contributions connect the implicit bias of SGD, mixture modeling, and hypothesis testing.

These results suggest several directions for future research.  First, \Cref{prop:heterosk} establishes conditional fixed-alternative consistency with matched heteroskedastic-null calibration in the scalar case; proving terminal concentration and extending feasible calibration to multivariate heteroskedastic designs are important next steps.  Second, the multivariate equilibrium and null theories currently assume Gaussian covariates, while the power guarantee additionally assumes Gaussian noise.  Developing equilibrium, null-calibration, and detection results for non-Gaussian multivariate covariates and non-Gaussian noise, together with finite-sample Type~I and Type~II error bounds, would broaden the inferential theory.  In the Gaussian benchmark, a nonasymptotic bound on $|S_T-4/\pi|$ would provide explicit guidance for sample size and step-size selection.  Third, replacing squared error with a general loss $\ell(y,\theta^\top x)$ would extend the framework to settings such as classification, although the diffusion analysis and symmetry argument would need to be adapted.

\if0\preprint
\paragraph*{Declaration of Generative AI Use.}
The authors used Claude Code (Fable) for wording edits and to assist with implementing standard methods, like SGD, exhaustive enumeration, classification trees, expectation-maximization,
and likelihood-ratio tests in our experiments. This assistance does not affect the core ideas, scientific rigor, or originality of the research.

\paragraph*{Disclosure Statement.}
The authors declare no conflicts of interest.

\paragraph*{Data Availability Statement.}
The empirical datasets analyzed in this study are publicly available from the original sources listed in the Supplementary Materials.
\if0\anon
An anonymized reproducibility code package is included with the submission.
\else
Versioned reproducibility code is available under the MIT License at \url{https://github.com/hecralies/optimizer-as-detector/releases/tag/v1.0-oad}.
\fi
\fi

\phantomsection\label{supplementary-material}

%
%

\bibliography{reference}

\def\spacingset#1{\renewcommand{\baselinestretch}%
{#1}\small\normalsize} \spacingset{1.8} 

\clearpage
\setcounter{section}{0}
\setcounter{equation}{0}
\setcounter{theorem}{0}
\setcounter{proposition}{0}
\setcounter{lemma}{0}
\setcounter{corollary}{0}
\setcounter{definition}{0}
\setcounter{assumption}{0}
\setcounter{remark}{0}
\setcounter{figure}{0}
\setcounter{table}{0}
\renewcommand{\thesection}{S\arabic{section}}
\renewcommand{\theequation}{S\arabic{equation}}
\renewcommand{\thetheorem}{S\arabic{theorem}}
\renewcommand{\theproposition}{S\arabic{proposition}}
\renewcommand{\thelemma}{S\arabic{lemma}}
\renewcommand{\thecorollary}{S\arabic{corollary}}
\renewcommand{\thedefinition}{S\arabic{definition}}
\renewcommand{\theassumption}{S\arabic{assumption}}
\renewcommand{\theremark}{S\arabic{remark}}
\renewcommand{\thefigure}{S\arabic{figure}}
\renewcommand{\thetable}{S\arabic{table}}
\def\theHsection{S\arabic{section}}
\def\theHequation{S\arabic{equation}}
\def\theHtheorem{S\arabic{theorem}}
\def\theHproposition{S\arabic{proposition}}
\def\theHlemma{S\arabic{lemma}}
\def\theHcorollary{S\arabic{corollary}}
\def\theHdefinition{S\arabic{definition}}
\def\theHassumption{S\arabic{assumption}}
\def\theHremark{S\arabic{remark}}
\def\theHfigure{S\arabic{figure}}
\def\theHtable{S\arabic{table}}

\begin{center}
\textbf{\large Supplementary Material for ``Optimizer as Detector: Stochastic Gradient Descent for Latent Mixture Models''}
\if1\anon

\medskip
\if1\archiveauthor
Ye Shi$^{*}$, Xiao Jin$^{+}$, and Chung-Piaw Teo$^{\dagger}$
\else
Ye Shi, Xiao Jin, and Chung-Piaw Teo
\fi
\if1\preprint

\medskip
\preprintdate
\fi
\fi
\end{center}
\medskip

\section{Literature Review}\label{app:literature-review}

Our work connects two streams of literature: the estimation and detection of mixtures of linear regressions, and the analysis of stochastic gradient descent through diffusion approximation.

\subsection{Mixture Regression and Heterogeneity Detection}

Modeling data as a mixture of linear relationships dates to the classical switching-regression and finite-mixture literature~\citep{quandt1972switching,Quandt1978,mclachlan2000finite}, in which the EM algorithm~\citep{dempster1977} is widely used for estimation.  Because the mixture likelihood is nonconvex, a large body of work asks when estimation can be guaranteed to succeed.  \citet{kwon2019global} establish global convergence of EM for a two-component mixed linear regression model, \citet{wang2019convergence} formulate mixed linear regression as a mixed-integer program and prove consistency in the noiseless case, and more recent studies consider high-dimensional~\citep{wang2024leveraging} and fully online~\citep{liu2025convergence} settings.  Most of these estimation methods assume a specified component count.  Closest to our testing objective, \citet{kasahara2015testing} study heteroscedastic normal mixture regression models and develop a modified penalized-EM test of $m_0$ against $m_0+1$ components.  In particular, their $m_0=1$ case is a formal homogeneity test against a two-component mixture-regression alternative.

Our approach differs from the existing methods as follows.  First, our statistic is generated by the stationary dynamics of constant-step SGD rather than a Gaussian mixture likelihood and batch penalized EM.  Second, although our algorithm uses two competing fitted pieces, the omnibus alternative in \eqref{eq:dgp} allows any unknown number $K^*\ge2$ of true latent components; the test asks whether heterogeneity exists without first estimating $K^*$.  Third, our matched null analysis covers general symmetric noise in the scalar setting as well as the multivariate Gaussian setting under the stated assumptions.  Finally, the subsequent Screen and Verify stages seek an observed explanation and check the sign reversal required for SP, tasks not addressed by a component-number test.

Research on Simpson's paradox asks a related but different question: it searches for a sign reversal associated with latent heterogeneity~\citep{blyth1972simpson,Pearl2009}.  Enumeration and tree-based methods~\citep{xu2018detecting,shmueli2018forest,alipourfard2018simpsons} require observed candidate confounders, whereas neural-network partitioning~\citep{wang2023simnet} learns a heterogeneity-revealing partition without providing calibrated inference.  Our Detect--Screen--Verify construction combines an SGD-based omnibus statistic with subsequent descriptive recovery of an observed explanation and direct verification of sign reversal.

\subsection{Stochastic Gradient Descent and Diffusion Approximation}

Our analysis also builds on the classical theory of stochastic approximation~\citep{robbins1951stochastic}, whose asymptotic behavior is commonly studied through ODE and dynamical-systems arguments~\citep{kushner2003,borkar2008,benveniste1990adaptive}.  One branch of this literature uses iterate averaging~\citep{ruppert1988efficient,polyak1992acceleration} for online statistical inference: \citet{chen2020statistical} construct confidence intervals from averaged SGD, \citet{zhu2021online} estimate its limiting covariance fully online, and \citet{chen2023online} extend inference to gradient-free Kiefer--Wolfowitz updates.  A complementary branch studies constant-step-size SGD through a diffusion limit~\citep{fan2018statistical,hu2019diffusion,wang2020asymptotic,liu2021diffusion}, whose stationary law captures persistent stochastic motion around attracting sets.  The cited diffusion-approximation studies primarily use these tools to quantify uncertainty around an estimator or describe convergence toward a minimizer.  We instead use the stationary diffusion law to derive the null distribution of a latent-heterogeneity test under the winner-take-all loss.

\section{Diffusion Approximation Analysis}\label{app:diffusion-approximation}

Our analysis of the SGD dynamics~\eqref{eq:sgd-update} under the competitive loss uses the weak-convergence method for constant-step stochastic approximation \citep[Sections~8.2.1 and~10.1.1]{kushner2003}.  The argument proceeds from the mean ODE to a local SDE and then to the long-horizon terminal law used for inference.  Because the sample update contains a discontinuous winner rule, the required regularity and localization conditions require separate verification.  \Cref{lem:local-diffusion-regularity} provides this verification, and \Cref{prop:localized-diffusion} establishes the resulting ODE, OU, and terminal approximations.  Their statements and proofs are deferred to \Cref{app:localized-diffusion-results} to ensure readability.  The related smooth-SGD analyses in \citet{fan2018statistical,wang2020asymptotic,liu2021diffusion} provide useful context.

\paragraph*{Step~1: Mean ODE.}
Under constant step size $\eta$ and the time rescaling $\tau = \eta t$, define the piecewise-constant interpolation $e^\eta(\tau) := e_{\lfloor \tau/\eta \rfloor}$.  On compact time intervals along trajectories separated from coincidence, this interpolation converges weakly, as established in \Cref{prop:localized-diffusion}(i), to the mean ODE $\dot e_k = \mu_k(e_1,e_2)$, where $e_k = \theta^{\mathrm{true}} - \theta_k$ and
\begin{equation}\label{eq:ode-drift}
	\mu_k(e_1,e_2) = -\E\bigl[I_k\,r_k\,x\bigr], \quad r_k = \varepsilon + e_k^\top x, \quad I_k=\ind\{k=w\}, \quad k = 1,2,
\end{equation}
where $w$ is the winning piece selected by the uniform tie rule specified in \Cref{alg:sgd} and its surrounding text.  Equivalently, after averaging over the tie randomization, define the selection weights $p_k:=\Pr(w=k\mid x,\varepsilon,e_1,e_2)$.  Since $p_k=\E[I_k\mid x,\varepsilon,e_1,e_2]$, the drift can also be written as $\mu_k=-\E[p_k r_kx]$, with $p_1+p_2=1$ everywhere.
Setting the drift to zero identifies the equilibria.  We focus on separated equilibria, meaning $e_1^*\ne e_2^*$, or equivalently $\delta^*\ne0$.  The prescribed nonzero half-gap initialization selects a trajectory that remains separated from coincidence, and the stability of its limiting equilibrium is established in the proof of \Cref{thm:eqm}.

\paragraph*{Step~2: Local SDE and OU approximation.}
The corresponding small-noise diffusion is
\begin{equation}\label{eq:sde-general}
	de = \mu(e)\,d\tau + \sqrt{\eta}\,\Gamma(e)\,dW_\tau,
\end{equation}
where $e=(e_1,e_2)$ and $\Gamma(e)\Gamma(e)^\top$ is the conditional covariance of the joint SGD noise.  Linearizing around a separated equilibrium $(e_1^*,e_2^*)$ and rescaling deviations by $\eta^{-1/2}$ gives, on fixed time intervals, the OU approximation established in \Cref{prop:localized-diffusion}(ii),
\begin{equation}\label{eq:ou-general}
	d\xi = J(e_1^*,e_2^*)\,\xi\,d\tau + \Gamma(e_1^*,e_2^*)\,dW_\tau,
\end{equation}
where $J(e_1^*,e_2^*)$ is the drift Jacobian.  

\paragraph*{Step~3: Long-horizon terminal law.}
The terminal regime $T\asymp\eta^{-p}$, $p\in(1,2)$, extends beyond fixed diffusion-time intervals.  We therefore show that, after burn-in, the iterates remain near the selected attractor, the contracting modes reach their stationary OU law, and tangent motion does not affect the gap statistic at first order.  These long-horizon steps are formalized in \Cref{prop:localized-diffusion}(iii) and applied in the proof of \Cref{prop:null-distribution}.

In the following proofs, we analyze the SGD dynamics through the error vectors $e_k$ and their midpoint--half-gap coordinates $\bfe:=(e_1+e_2)/2$ and $\delta:=(e_1-e_2)/2$; we refer to the $\bfe$-direction as the symmetric (midpoint) mode and to the $\delta$-direction as the antisymmetric (half-gap) mode.  \Cref{tab:notation} collects these and the remaining recurring symbols. 

\begin{table}[H]
	\centering
	\caption{Summary of notation.}
	\label{tab:notation}
	\renewcommand{\baselinestretch}{1}\footnotesize
	\begin{tabular}{@{}lp{0.75\textwidth}@{}}
		\toprule
		Notation & Meaning \\
		\midrule
		$e_k$ & Null error of piece $k$: $\theta^{\mathrm{true}}-\theta_k$ \\
		$r_k$;\; $w$;\; $p_k$ & Residual $r_k=\varepsilon+e_k^\top x$; winner $w$; tie-aware weight $p_k=\Pr(w=k\mid x,\varepsilon,e_1,e_2)$ \\
		$\mu_k$ & Mean-ODE drift: $\mu_k(e_1,e_2)=-\E[p_kr_kx]$; see \eqref{eq:ode-drift} \\
		$\bfe$ & Midpoint $(e_1+e_2)/2$  \\
		$\delta$ & Half-gap $(e_1-e_2)/2$ \\
		$\tau$ & Diffusion time $\eta t$ \\
		$\Sigma$;\; $\|e\|_\Sigma$ & Second moment $\Sigma=\E[xx^\top]$; norm $\|e\|_\Sigma=(e^\top\Sigma e)^{1/2}$ \\
		$\kappa_\varepsilon$;\; $\kappa_x$ & Normalized absolute moments: $\kappa_\varepsilon=\E[|\varepsilon|]/\sigma_\varepsilon$; for $d=1$, $\kappa_x=\E[|x|]/\sqrt{\E[x^2]}$ \\
		$c_0$ & Equilibrium-radius constant: $c_0\sigma_\varepsilon=\E[|\varepsilon|]\sqrt{2/\pi}$; $c_0=2/\pi$ for Gaussian noise \\
		$e^*$ & Equilibrium half-gap: $(e_1^*,e_2^*)=(e^*,-e^*)$; for $d=1$, $e^*=\E[|\varepsilon|]\E[|x|]/\E[x^2]$; for $d\ge2$, $e^*\in\mathcal E_2$ \\
		$\mathcal E_2$ & Half-gap equilibrium ellipsoid $\{e\in\R^d:\|e\|_\Sigma=c_0\sigma_\varepsilon\}$ ($d\ge2$) \\
		$\mathcal M_2$ & Equilibrium manifold $\{(e,-e):e\in\mathcal E_2\}$ \\
		$J^{\mathrm{sym}}$;\; $J^{\mathrm{anti}}$ & Symmetric and antisymmetric Jacobian blocks at a separated equilibrium  \\
		$\alpha_\varepsilon$;\; $\alpha_x$ & Stability constants: $\alpha_\varepsilon=2f_\varepsilon(0)\E[|\varepsilon|]$; for $d=1$, $\alpha_x=\E[|x|]\E[|x|^3]/(\E[x^2])^2$ and stability requires $\alpha_\varepsilon\alpha_x<1$; for Gaussian covariates with $d\ge2$, stability requires $\alpha_\varepsilon<\pi/4$ \\
		$\xi_\delta$ & Rescaled half-gap fluctuation $\eta^{-1/2}(\delta-e^*)$ \\
		$D(e^*)$ & Common piece-noise covariance at a separated equilibrium ($D$ when $d=1$) \\
		$S_T$;\; $\hat S_T$;\; $S_\infty^{H_0}$, $S_\infty^{H_1}$ & Statistic $S_T=\|\theta_{1,T}-\theta_{2,T}\|_\Sigma/\sigma_\varepsilon$; plug-in $\hat S_T$; limits under $H_0$ and $H_1$ \\
		$m$;\; $\delta_\theta$;\; $\alpha_z$ (under $H_1$) & Parameter midpoint $(\theta_1+\theta_2)/2$; parameter half-gap $\delta_\theta=(\theta_1-\theta_2)/2$; group-specific offset from the fitted midpoint $\theta_z^{\mathrm{true}}-m$ \\
		\bottomrule
	\end{tabular}
\end{table}

\section{Proofs for Results in \Cref{sec:detection}}\label{app:proofs}

Here, we provide proofs for the analytical results in \Cref{sec:detection}. These proofs invoke the localized diffusion-approximation results and several auxiliary lemmas; for readability, their statements and proofs are collected in \Cref{app:localized-diffusion-results,app:auxiliary-lemmas}.

\subsection{Proof of \Cref{thm:eqm}}

For analytical convenience, we first rewrite the mean ODE~\eqref{eq:ode-drift} in the midpoint--half-gap coordinates $(\bfe,\delta)$.  Since $r_1^2-r_2^2=4\delta^\top x(\varepsilon+\bfe^\top x)$, piece~1 wins exactly when $\delta^\top x(\varepsilon+\bfe^\top x)<0$, and the uniform tie rule assigns weight $\tfrac12$ to each piece when this quantity vanishes.  The mean ODE can therefore be written as
\begin{equation}\label{eq:ode-drift-expanded}
	\begin{aligned}
		\dot e_1 &= -\E\Bigl[\bigl(\ind\{\delta^\top x(\varepsilon+\bfe^\top x)<0\}+\tfrac12\ind\{\delta^\top x(\varepsilon+\bfe^\top x)=0\}\bigr)(\varepsilon+e_1^\top x)\,x\Bigr],\\
		\dot e_2 &= -\E\Bigl[\bigl(\ind\{\delta^\top x(\varepsilon+\bfe^\top x)>0\}+\tfrac12\ind\{\delta^\top x(\varepsilon+\bfe^\top x)=0\}\bigr)(\varepsilon+e_2^\top x)\,x\Bigr].
	\end{aligned}
\end{equation}

\paragraph*{Scalar case ($d=1$).}  We first consider the scalar case.  On the antisymmetric regime $e_1=-e_2=\delta$ and $\bfe=0$, the selection event for $\delta>0$ reduces to $\{\varepsilon x\le0\}$: ties contribute no drift, because observations with $x=0$ zero the integrand and $\Pr(\varepsilon=0)=0$ under \Cref{ass:diffusion}.  Conditioning on $x$ and using the symmetry of $\varepsilon$, we have $\E[x^2\ind\{\varepsilon x\le0\}]=\tfrac12\E[x^2]$ and $\E[\varepsilon x\ind\{\varepsilon x\le0\}]=-\tfrac12\E[|\varepsilon|]\E[|x|]$.  Therefore,
\begin{equation}
	\begin{aligned}
		\mu_1(\delta,-\delta)
		= -\E\bigl[\ind\{\varepsilon x \le 0\}(\varepsilon + \delta x)x\bigr]
		&= -\E[\varepsilon x\,\ind\{\varepsilon x \le 0\}] - \delta\,\E[x^2\,\ind\{\varepsilon x \le 0\}]\\
		&
		= \tfrac{1}{2}\E[|\varepsilon|]\E[|x|] - \tfrac{1}{2}\E[x^2]\,\delta.
	\end{aligned}
\end{equation}
When $\delta<0$, the selection event becomes $\{\varepsilon x\ge0\}$, and the same calculation gives $\mu_1(\delta,-\delta)=-\tfrac12\E[|\varepsilon|]\E[|x|]-\tfrac12\E[x^2]\delta$.  Combining the two cases and applying \Cref{lem:drift-symmetry}(a,b), we obtain, for $\delta\ne0$,
\begin{equation}\label{eq:delta-ode}
	\dot{\delta} = \tfrac{1}{2}(\mu_1 - \mu_2) = f(\delta) = \tfrac{1}{2}\E[|\varepsilon|]\E[|x|]\,\sgn(\delta) - \tfrac{1}{2}\E[x^2]\,\delta, \quad f(0) = 0.
\end{equation}
Solving $f(\delta)=0$ on the positive branch gives $\delta^*=e^*=\E[|\varepsilon|]\E[|x|]/\E[x^2]$.  Because $f$ is odd, $-e^*$ is the other antisymmetric equilibrium point.  It remains to establish local full-state stability.  By \Cref{lem:equiv}, $\psi'(0)=2f_\varepsilon(0)\E[|x|^3]$, while symmetry and independence of $\varepsilon$ give
\[
	\left.\frac{d}{d\bfe}\E\bigl[|\varepsilon+\bfe x|\,|x|\bigr]\right|_{\bfe=0}
	=\E[\sgn(\varepsilon)x|x|]=0.
\]
Consequently, the full Jacobian in $(\bfe,\delta)$ coordinates at $(0,e^*)$ is
\begin{equation}\label{eq:scalar-full-jacobian}
	J_{(\bfe,\delta)}(0,e^*)
	=\begin{pmatrix}
	-\tfrac12\E[x^2](1-\alpha_\varepsilon\alpha_x)&0\\
	0&-\tfrac12\E[x^2]
	\end{pmatrix}.
\end{equation}
When $\alpha_\varepsilon\alpha_x<1$, both eigenvalues are strictly negative, so $J:=J_{(\bfe,\delta)}(0,e^*)$ is Hurwitz.  Put $u=(\bfe,\delta-e^*)^\top$.  By \Cref{lem:local-diffusion-regularity}(i)--(ii), the drift is locally Lipschitz near $(0,e^*)$ and
\[
\dot u=Ju+r(u),
\qquad
\frac{\|r(u)\|}{\|u\|}\longrightarrow0
\quad\text{as }u\to0.
\]
Let $P\succ0$ be the symmetric solution of the Lyapunov equation $J^\top P+PJ=-I$, and define $L(u)=u^\top Pu$.  Along the unique local solutions,
\[
\dot L(u)=-\|u\|^2+2u^\top P r(u).
\]
For all sufficiently small $u$, the Fr\'echet remainder satisfies $2|u^\top P r(u)|\le\tfrac12\|u\|^2$.  Hence
\[
\dot L(u)\le-\tfrac12\|u\|^2
\le-\frac{1}{2\lambda_{\max}(P)}L(u),
\]
and Gronwall's inequality gives local exponential stability of $(e_1,e_2)=(e^*,-e^*)$.  The same conclusion holds for $(-e^*,e^*)$ by label symmetry.  This proves~(i).

\paragraph*{Multivariate case ($d \ge 2$).}  We next consider the multivariate case.  On the antisymmetric regime $\delta=e$ and $\bfe=0$, we decompose $\mu_1(e,-e)$ into its regression and noise terms.  Symmetry of $\varepsilon$ gives $\E[\ind\{\varepsilon(e^\top x)\le0\}\mid x]=\tfrac12$, so the regression term is $-\tfrac12\Sigma e$.  For the noise term, conditioning on $x$ and applying \Cref{lem:sign-proj}(i), together with $c_0\sigma_\varepsilon=\E[|\varepsilon|]\sqrt{2/\pi}$, gives $\tfrac12c_0\sigma_\varepsilon\Sigma e/\|e\|_\Sigma$.  Combining the two terms yields
\begin{equation}\label{eq:multi-drift}
	\mu_1(e,-e) = \frac{1}{2}\left(\frac{c_0\sigma_\varepsilon}{\|e\|_\Sigma}-1\right)\Sigma e.
\end{equation}
By \Cref{lem:drift-symmetry}(a,b), $\mu_2(e,-e)=-\mu_1(e,-e)$.  Since $\Sigma\succ0$, the equilibrium condition reduces to $\|e^*\|_\Sigma=c_0\sigma_\varepsilon$.  Hence the half-gap equilibrium set is the Mahalanobis ellipsoid $\mathcal{E}_2=\{e\in\R^d:\|e\|_\Sigma=c_0\sigma_\varepsilon\}$.

It remains to establish local full-state stability.  By \Cref{lem:local-diffusion-regularity}(ii), the full drift is jointly Fr\'echet differentiable at $(\bfe,\delta)=(0,e^*)$.  Hence \Cref{lem:jacobian-block} applies, and the $2d\times2d$ Jacobian at $(e_1,e_2)=(e^*,-e^*)$ is block diagonal in the $(\bfe,\delta)$ coordinates:
$J_{(\bfe,\delta)}=\operatorname{diag}(J^{\mathrm{sym}},J^{\mathrm{anti}})$.  We consider the two blocks separately.

\textit{Symmetric block.}  By \Cref{lem:ebar-multi}, $\dot{\bfe}=-\tfrac12\Sigma\bfe+\tfrac12\Psi(\bfe,\delta)$, where $\Psi(\bfe,\delta):=\E[|\delta^\top x|\sgn(\varepsilon+\bfe^\top x)x]$.  Conditional on $x$, the expectation over $\varepsilon$ is $h(\bfe^\top x)$, where $h(t):=2F_\varepsilon(t)-1=2f_\varepsilon(0)t+o(t)$ at zero.  Differentiating at $(\bfe,\delta)=(0,e^*)$, as justified in \Cref{lem:ebar-multi}, gives
\[
	\left.\frac{\partial\Psi}{\partial\bfe}\right|_{(0,e^*)}
	=2f_\varepsilon(0)\E[|e^{*\top}x|xx^\top].
\]
Substituting this derivative into $J^{\mathrm{sym}}$ and applying \Cref{lem:sign-proj}(ii), we obtain
\begin{equation}\label{eq:Jsym}
	J^{\mathrm{sym}}
	= -\tfrac{1}{2}\Sigma
	+ f_\varepsilon(0)\,\|e^*\|_\Sigma\,\E[|Z|]
	\left(\Sigma + \frac{\Sigma e^* e^{*\top}\Sigma}{e^{*\top}\Sigma e^*}\right).
\end{equation}
To determine its sign, define $S:=\Sigma^{-1/2}J^{\mathrm{sym}}\Sigma^{-1/2}$ and $\hat v:=\Sigma^{1/2}e^*/\|e^*\|_\Sigma$.  Using $\alpha_\varepsilon=2f_\varepsilon(0)\E[|\varepsilon|]$, we obtain $S=(-\tfrac12+\tfrac{\alpha_\varepsilon}{\pi})I+\tfrac{\alpha_\varepsilon}{\pi}\hat v\hat v^\top$.  By Sylvester's law of inertia, $J^{\mathrm{sym}}\prec0$ if and only if $S\prec0$.  The eigenvalue of $S$ on the orthogonal complement $\hat v^\perp$ is $s_\perp=-\tfrac12(1-\tfrac{2}{\pi}\alpha_\varepsilon)$, with multiplicity $d-1$, while its eigenvalue along $\operatorname{span}\{\hat v\}$ is $s_\parallel=-\tfrac12(1-\tfrac{4}{\pi}\alpha_\varepsilon)$.  Both are negative when $\alpha_\varepsilon<\pi/4$.  In particular, Gaussian noise gives $\alpha_\varepsilon=2/\pi<\pi/4$, and hence $J^{\mathrm{sym}}\prec0$.

\textit{Antisymmetric block.}  On $\bfe=0$, \eqref{eq:multi-drift} and $\mu_2(\cdot)=-\mu_1(\cdot)$ give $\dot\delta=\mu_1(\delta,-\delta)=:g(\delta)=\tfrac12(c_0\sigma_\varepsilon/\|\delta\|_\Sigma-1)\Sigma\delta$.  Define $r(\delta):=\|\delta\|_\Sigma$.  Since $\partial r/\partial\delta=\Sigma\delta/r(\delta)$, differentiation gives
$$\nabla_\delta g = \frac{c_0\sigma_\varepsilon}{2r}\Sigma - \frac{c_0\sigma_\varepsilon}{2r^3}(\Sigma\delta)(\Sigma\delta)^\top - \frac{1}{2}\Sigma.$$
At the equilibrium radius $r=c_0\sigma_\varepsilon$, the first and third terms cancel.  Therefore,
\begin{equation}\label{eq:jacobian}
	J^{\mathrm{anti}} = -\frac{1}{2(c_0\sigma_\varepsilon)^2}(\Sigma e^*)(\Sigma e^*)^\top.
\end{equation}
This matrix is rank-one negative semidefinite, with nonzero eigenvalue $-e^{*\top}\Sigma^2 e^*/(2e^{*\top}\Sigma e^*)$.  Its $d-1$ zero eigenvalues correspond exactly to the tangent space of $\mathcal E_2$ at $e^*$.

We now combine the two blocks.  The symmetric block contributes $d$ strictly negative eigenvalues, while the antisymmetric block contributes one negative radial eigenvalue and $(d{-}1)$ zero eigenvalues tangent to $\mathcal{E}_2$.  In $(\bfe,\delta)$ coordinates, the full equilibrium manifold is $\{(0,e):e\in\mathcal E_2\}$; equivalently, in $(e_1,e_2)$ coordinates it is $\mathcal M_2=\{(e,-e):e\in\mathcal E_2\}$.  Its tangent kernel is exactly $\{0\}\times T_e\mathcal E_2$.  Moreover, the block matrices vary continuously with $e\in\mathcal E_2$, and compactness of $\mathcal E_2$ makes the $d+1$ negative transverse eigenvalues uniformly bounded away from zero.

We next establish local attraction using an explicit Lyapunov function.  Put $r_0:=c_0\sigma_\varepsilon$, $q:=\|\delta\|_\Sigma$, and consider the quadratic Lyapunov function
\begin{equation}\label{eq:explicit-transverse-lyapunov}
L(\bfe,\delta):=\|\bfe\|^2+(q-r_0)^2.
\end{equation}
On a sufficiently small tube around $\{0\}\times\mathcal E_2$, $q$ is bounded away from zero and $L$ is uniformly equivalent to the squared distance from the equilibrium manifold.  For such $\delta$, define the radial retraction $e(\delta):=r_0\delta/q\in\mathcal E_2$.  The uniform Fr\'echet expansion in \Cref{lem:local-diffusion-regularity}(ii) gives
\[
\dot\bfe=J^{\mathrm{sym}}(e(\delta))\bfe+\rho_{\bfe},
\qquad
\|\rho_{\bfe}\|
\le \omega(\sqrt L)\sqrt L,
\]
for a deterministic modulus $\omega(r)\to0$, uniformly along the manifold.  Because $J^{\mathrm{sym}}(e)$ is symmetric negative definite for every $e\in\mathcal E_2$ and depends continuously on $e$, compactness supplies $a_1>0$ such that
\[
2\bfe^\top\dot\bfe
\le -a_1\|\bfe\|^2+\omega_1(\sqrt L)L,
\qquad \omega_1(r)\to0.
\]

For the radial component, \eqref{eq:multi-drift} gives the exact restricted identity
\[
\dot q
=-\frac12(q-r_0)
\frac{\delta^\top\Sigma^2\delta}{\delta^\top\Sigma\delta}
\qquad\text{when }\bfe=0.
\]
The quadratic bound for $\Phi(\bfe,\delta)-\Phi(0,\delta)$ in the proof of \Cref{lem:local-diffusion-regularity}(ii) shows that, in the full system, the right-hand side has an additional remainder $\rho_q$ satisfying $|\rho_q|\le C\|\bfe\|^2$.  Since the Rayleigh quotient is bounded below by $\lambda_{\min}(\Sigma)>0$,
\[
\frac{d}{d\tau}(q-r_0)^2
\le-\lambda_{\min}(\Sigma)(q-r_0)^2
+C|q-r_0|\,\|\bfe\|^2.
\]
Combining the last two displays and using
$|q-r_0|\,\|\bfe\|^2\le\sqrt L\,L$ yields constants $a>0$ and a modulus $\widetilde\omega(r)\to0$ such that
\[
\dot L\le-aL+\widetilde\omega(\sqrt L)L.
\]
The drift is locally Lipschitz on this tube by \Cref{lem:local-diffusion-regularity}(i), so the inequality holds along its unique solutions.  Shrinking the tube makes the second term at most $aL/2$; Gronwall's inequality then gives $L(\tau)\le L(0)e^{-a\tau/2}$.  Thus $\mathcal M_2$ is uniformly locally exponentially attracting and locally exponentially stable as a set.  This proves~(ii).

\paragraph*{Instability of the coincident origin.}
We finally verify the instability of the coincident origin.  Under the uniform tie rule, $(\bfe,\delta)=(0,0)$ is an equilibrium.  In the scalar case, take any $\delta \in (0, \bar\delta)$ with $\bar\delta \le e^*$.  The scalar drift satisfies $f(\delta) = \tfrac{1}{2}\E[|\varepsilon|]\E[|x|] - \tfrac{1}{2}\E[x^2]\delta > 0$ because $\delta < e^*$.  Similarly, $f(\delta) < 0$ for $\delta \in (-\bar\delta, 0)$.  Thus perturbations along the antisymmetric line are amplified.  In the multivariate Gaussian case, define $a:=c_0\sigma_\varepsilon$ and $r:=\|\delta\|_\Sigma$.  For $\delta\ne0$, the antisymmetric restriction of \eqref{eq:multi-drift} gives $\dot r=\tfrac{1}{2r}(a/r-1)\delta^\top\Sigma^2\delta>0$ whenever $0<r<a$.  Thus the coincident origin is also repelling along every nonzero antisymmetric ray in the multivariate setting.
\hfill$\square$

\subsection{Proof of \Cref{prop:null-distribution}}

By \Cref{thm:eqm,prop:init}, the unique mean-ODE trajectory selected under $H_0$ approaches a separated antisymmetric equilibrium in the scalar case or the equilibrium manifold in the Gaussian multivariate case.  By label symmetry, we relabel the scalar pieces when necessary so that the selected equilibrium is $(e_1^*,e_2^*)=(e^*,-e^*)$ with $e^*>0$; this relabeling leaves the gap statistic and the variance coefficient $D$ unchanged.  In the multivariate case, the fixed initialization direction $c$ and uniqueness of the ODE determine a nonrandom asymptotic phase $e^*=e^*(c)$ on the manifold.  The trajectory remains separated from coincidence, so \Cref{lem:local-diffusion-regularity} verifies all local hypotheses of \Cref{prop:localized-diffusion}.  Because $T\asymp\eta^{-p}$ with $p\in(1,2)$ gives $\eta T/\log(1/\eta)\to\infty$ and $\eta^2T\to0$, part~(iii) of that proposition applies: the contracting coordinates attain their stationary OU law and, in the multivariate case, the selected equilibrium direction changes by only $o_p(1)$ on the original scale.  The terminal martingale-array central limit theorem in the proof of that proposition is unconditional and freezes its coefficients at the deterministic phase $e^*(c)$.

We now compute the stationary normal fluctuation in midpoint--half-gap coordinates $(\bfe,\delta)$.  The joint Fr\'echet differentiability required by \Cref{lem:jacobian-block} is verified in \Cref{lem:local-diffusion-regularity}(ii).  Thus both the drift Jacobian and the SGD-noise covariance are block-diagonal at a separated equilibrium, so \eqref{eq:ou-general} splits into independent midpoint and half-gap components.  The statistic $S_T=2\|\delta_T\|_\Sigma/\sigma_\varepsilon$ depends only on the half-gap and is constant along the multivariate equilibrium manifold.  Therefore \Cref{prop:localized-diffusion}(iii) removes the neutral tangent coordinates at first order, and its $\sqrt\eta$-order variance is determined by the stationary normal OU component.  The smooth normal-coordinate remainder is $o_p(\sqrt\eta)$.

At the equilibrium exactly one piece is updated per iteration, so the two piece noises are pointwise mutually exclusive and, their means vanishing there, uncorrelated.  By label symmetry they share a common covariance $D(e^*)$, and in $(\bfe,\delta)$ coordinates the midpoint and half-gap noises each have covariance $D(e^*)/2$ with zero cross-covariance.  Since the drift Jacobian is also block-diagonal (\Cref{lem:jacobian-block}), the fluctuation~\eqref{eq:ou-general} decouples as claimed.  Writing $\xi_\delta:=\eta^{-1/2}(\delta-e^*)$ for the rescaled half-gap fluctuation, we carry out the scalar and multivariate calculations separately.

\paragraph*{Scalar case.}
We first consider the scalar case.  By noise symmetry, the common piece-noise variance is
\[
D=\E[(\varepsilon+e^*x)^2x^2\ind\{\varepsilon x\le0\}]
=\E[(\varepsilon-e^*x)^2x^2\ind\{\varepsilon x\ge0\}],
\]
whereas the antisymmetric drift eigenvalue is $\lambda_\delta=-\E[x^2]/2$.  By \Cref{lem:equiv}, the symmetric eigenvalue is
$
\lambda_{\mathrm{sym}}
=-\frac{\E[x^2]}{2}(1-\alpha_\varepsilon\alpha_x).
$
Under $\alpha_\varepsilon\alpha_x<1$, both eigenvalues are negative.  Since the half-gap noise has variance $D/2$, the rescaled half-gap fluctuation $\xi_\delta$ satisfies
\begin{equation}\label{eq:ou3}
d\xi_\delta=-\tfrac12\E[x^2]\xi_\delta\,d\tau+\sqrt{D/2}\,dB_\tau,
\end{equation}
whose stationary variance is $D/(2\E[x^2])$.  The long-horizon terminal law in \Cref{prop:localized-diffusion}(iii) therefore gives the distributional statement
\[
\frac{\delta_T-e^*}
{\sqrt{\eta D/(2\E[x^2])}}
\xrightarrow{d}N(0,1).
\]
Since $e^*>0$ and $\delta_T-e^*=O_p(\sqrt\eta)$, $\Pr(\delta_T>0)\to1$.  Hence, with probability tending to one,
\[
S_T-S_\infty^{H_0}
=\frac{2\sqrt{\E[x^2]}}{\sigma_\varepsilon}(\delta_T-e^*).
\]
Combining these two displays with Slutsky's theorem gives the scalar convergence in \Cref{prop:null-distribution}(i); equivalently, $2D/\sigma_\varepsilon^2$ is the variance coefficient of the limiting law of $(S_T-S_\infty^{H_0})/\sqrt\eta$.  This proves part~(i).

For Gaussian noise and covariates, we further expand $D$ as
(a)~$\E[\varepsilon^2x^2\ind\{\varepsilon x\ge0\}]=\sigma_\varepsilon^2\E[x^2]/2$;
(b)~$-2e^*\E[\varepsilon x^3\ind\{\varepsilon x\ge0\}]=-8\sigma_\varepsilon^2\E[x^2]/\pi^2$; and
(c)~$(e^*)^2\E[x^4\ind\{\varepsilon x\ge0\}]=6\sigma_\varepsilon^2\E[x^2]/\pi^2$.
Therefore, $D=\sigma_\varepsilon^2\E[x^2](1/2-2/\pi^2)$ and the leading mean is $4/\pi$.

\paragraph*{Multivariate case.}
We next consider the multivariate case.  The covariance identity above gives
\[
D(e^*)=\E[(\varepsilon-e^{*\top}x)^2xx^\top
\ind\{\varepsilon e^{*\top}x\ge0\}],
\qquad Q_\delta(e^*)=\frac12D(e^*).
\]
By the proof of \Cref{thm:eqm}, the antisymmetric Jacobian is
\[
J^{\mathrm{anti}}=-\frac{(\Sigma e^*)(\Sigma e^*)^\top}{2(c_0\sigma_\varepsilon)^2}.
\]
The full fixed-point OU process has no stationary distribution.  The matrix $J^{\mathrm{anti}}$ has one negative radial eigenvalue and $d-1$ zero eigenvalues tangent to $\mathcal E_2$.  For any nonzero tangent vector $t$, Gaussian nondegeneracy gives $t^\top Q_\delta(e^*)t>0$, while $J^{\mathrm{anti}}t=0$.  If a stationary covariance solved $J^{\mathrm{anti}}\Sigma_\xi+\Sigma_\xi J^{\mathrm{anti}\top}+Q_\delta=0$, contraction by $t$ would give the contradiction $0=t^\top Q_\delta t>0$.  Hence the tangent coordinates diffuse in the local tangent directions.

It remains to study the radial normal coordinate, because the gap depends only on the Mahalanobis radius and is constant along $\mathcal E_2$.  Let $e^*=e^*(c)\in\mathcal E_2$ denote the deterministic direction selected by the ideal mean-ODE trajectory from the fixed initialization direction $c$, and define
\[
r_0=c_0\sigma_\varepsilon,\,
u_r=\frac{\Sigma e^*}{\|\Sigma e^*\|},\,
|\lambda_{\mathrm{rad}}|=\frac{\|\Sigma e^*\|^2}{2r_0^2},
\, q_r=u_r^\top Q_\delta(e^*)u_r=\frac12u_r^\top D(e^*)u_r.
\]
On fixed intervals, projecting the local fluctuation onto $u_r$ gives a scalar OU process with rate $|\lambda_{\mathrm{rad}}|$, diffusion variance $q_r$, and stationary variance $q_r/(2|\lambda_{\mathrm{rad}}|)$.  Over the longer regime used in the theorem, \Cref{prop:localized-diffusion}(iii) applies this OU limit to the moving normal coordinate while the unscaled tangent coordinate changes by $o_p(1)$; continuity then permits freezing the coefficients at the selected limit $e^*$.  For $g(\delta)=\|\delta\|_\Sigma$,
\[
\nabla g(e^*)=\frac{\Sigma e^*}{r_0}
=\frac{\|\Sigma e^*\|}{r_0}u_r.
\]
Therefore, the first-order radial fluctuation $Y=\nabla g(e^*)^\top\xi_\delta$ has variance
\[
\Var(Y)
=\frac{\|\Sigma e^*\|^2}{r_0^2}
\frac{q_r}{2|\lambda_{\mathrm{rad}}|}
=q_r.
\]
Since $S_T=2g(\delta_T)/\sigma_\varepsilon$, define the radial variance coefficient
\begin{equation}\label{eq:multi-radial-variance}
V(e^*):=\frac{4q_r}{\sigma_\varepsilon^2}
=\frac{2\,u_r^\top D(e^*)u_r}{\sigma_\varepsilon^2}.
\end{equation}
 
Use a tubular neighborhood of $\mathcal E_2$ and let $\pi(\delta_T)\in\mathcal E_2$ be the local projection of $\delta_T$ onto the manifold.  By \Cref{prop:localized-diffusion}(iii), the normal displacement $\delta_T-\pi(\delta_T)$ is $O_p(\sqrt\eta)$, whereas $\pi(\delta_T)-e^*=o_p(1)$.  A uniform normal-coordinate Taylor expansion therefore gives
\[
g(\delta_T)=r_0+\sqrt\eta\,Y+o_p(\sqrt\eta),
\]
Because $g$ is constant on $\mathcal E_2$, only the $O_p(\sqrt\eta)$ normal displacement enters the first-order term.  The delta method \citep[Chapter~3]{vaart1998} and the stationary normal OU law now give the unconditional convergence in \Cref{prop:null-distribution}(ii), with leading center $2r_0/\sigma_\varepsilon=2\kappa_\varepsilon\sqrt{2/\pi}$.  This proves part~(ii).

Under anisotropic $\Sigma$, the radial variance can depend on the selected equilibrium direction.  The test therefore calibrates the multivariate critical value by simulation, as specified in \Cref{sec:scale-estimation}.

Finally, parts~(i)--(ii) imply $S_T\xrightarrow{p}S_\infty^{H_0}$, which proves part~(iii).

The two restrictions on the iteration regime are used through \Cref{prop:localized-diffusion}(iii).  The polynomial burn-in gives $\eta T/\log(1/\eta)\to\infty$, which removes the initial condition even after $\sqrt\eta$ normalization and yields the stationary OU variance.  The condition $\eta^2T\to0$ makes the accumulated unscaled tangent martingale and the second-order phase drift $o_p(1)$, so the multivariate equilibrium direction remains locally fixed while the radial functional reaches its equilibrium fluctuation law.
\hfill$\square$

\subsection{Proof of \Cref{prop:init}}

In midpoint--half-gap coordinates, the proposed initialization gives
$\bfe_0=\theta^{\mathrm{true}}-\hat\theta_0\xrightarrow{p}0$ and
$\delta_0=-c$, and hence $(\bfe_0,\delta_0)\xrightarrow{p}(0,-c)$.
The limiting point is separated because $c\ne0$.  We first analyze its deterministic trajectory.  In the scalar case, the invariant relation $\bfe=0$ and \eqref{eq:delta-ode} show that the trajectory remains on the sign branch containing $-c$ and converges to the corresponding separated equilibrium.  In the multivariate Gaussian case, on the same invariant subspace, \eqref{eq:multi-drift} gives, with $q=\|\delta\|_\Sigma$,
$
\dot q
=\frac12(c_0\sigma_\varepsilon-q)
\frac{\delta^\top\Sigma^2\delta}{\delta^\top\Sigma\delta}.
$
Since the ratio is bounded between the extreme eigenvalues of $\Sigma$, $q\to c_0\sigma_\varepsilon$ whenever $c\ne0$.  In both cases the half-gap norm remains bounded below by
\[
q_{\min}:=
\begin{cases}
\min\{|c|,e^*\}, & d=1,\\
\min\{\|c\|_\Sigma,c_0\sigma_\varepsilon\}, & d\ge2,
\end{cases}
\]
where $q_{\min}$ is positive. The entire limiting trajectory is therefore contained in a compact tube separated from the coincident set.

By \Cref{lem:local-diffusion-regularity}(i), the drift is locally Lipschitz on this tube, so the ODE solution is unique and depends continuously on its initial condition.  Covering the compact trajectory by finitely many such neighborhoods gives an open set $G$ containing $(0,-c)$ whose trajectories remain separated and enter the locally exponentially attracting neighborhood from \Cref{thm:eqm}.  Hence $G$ lies in the basin of the relevant scalar equilibrium or multivariate equilibrium manifold.  Since $(\bfe_0,\delta_0)\xrightarrow{p}(0,-c)$, we have
$\Pr\{(\bfe_0,\delta_0)\in G\}\to1$, which proves the proposition.
\hfill$\square$

\subsection{Proof of \Cref{cor:CI}}
Put $X_\eta=(S_T-S_\infty^{H_0})/\sqrt\eta$.  By \Cref{prop:null-distribution},
\[
X_\eta\xrightarrow{d}\sigma_*Z,
\qquad
\sigma_*^2=
\begin{cases}
2D/\sigma_\varepsilon^2, & d=1,\\
V(e^*(c)), & d\ge2,
\end{cases}
\]
where $\sigma_*>0$ and the multivariate convergence is unconditional because the fixed direction $c$ selects the deterministic phase $e^*(c)$.  Since the limiting Gaussian distribution is continuous and strictly increasing, convergence of quantiles gives
\[
\frac{q_{1-\alpha,\eta}-S_\infty^{H_0}}{\sqrt\eta}
\longrightarrow \sigma_*z_{1-\alpha}.
\]
Consequently, with
$W_\eta=(S_T-q_{1-\alpha,\eta})/\sqrt\eta$,
\[
W_\eta\xrightarrow{d}
\sigma_*(Z-z_{1-\alpha})=:W,
\qquad \Pr(W=0)=0.
\]

It remains to show that preliminary centering and plug-in estimation are negligible on the $\sqrt\eta$ scale.  Conditional on the preliminary sample, centering replaces $(y,x)$ by $(y-\hat\mu_y,x-\hat\mu_x)$.  On the compact attracting tube used in \Cref{prop:localized-diffusion}, the bounded-density and moment bounds in \Cref{lem:local-diffusion-regularity} make the associated drift and covariance perturbations of order $O_p(r_\eta)$, where $r_\eta:=|\hat\mu_y|+\|\hat\mu_x\|$.  Exponential attraction makes the normal drift response $O_p(r_\eta)=o_p(\sqrt\eta)$, while continuity of the innovation covariance makes its effect on the $\sqrt\eta$ fluctuation $o_p(\sqrt\eta)$.  In the multivariate case the centering-induced tangent drift accumulates to $O_p(r_\eta\eta T)=o_p(1)$, while the tangent martingale remains $o_p(1)$ because $\eta^2T\to0$.  Hence the coefficients of the local terminal approximation converge to those at the phase selected by the oracle trajectory.  By \Cref{prop:null-distribution}(iii), $S_T\xrightarrow{p}S_\infty^{H_0}>0$.  Since $\Sigma\succ0$ and $\sigma_\varepsilon>0$, the Mahalanobis norm and scale normalization are smooth near the null limit.  Combining these centering bounds with the plug-in rates in \Cref{cor:CI} gives
\[
	\hat S_T-S_T=o_p(\sqrt\eta).
\]
For notational convenience, define
$\Delta_T=(\hat S_T-S_T)
-(\hat q_{1-\alpha,\eta}-q_{1-\alpha,\eta})$.
Then $\Delta_T=o_p(\sqrt\eta)$ and
\[
	\Pr(\hat S_T>\hat q_{1-\alpha,\eta})
	=\Pr\!\left(
	W_\eta
	+
	\frac{\Delta_T}{\sqrt\eta}>0
	\right).
\]

By Slutsky's theorem and $\Pr(W=0)=0$, the right-hand side converges to $\Pr(W>0)$.  It remains to identify this limit with $\alpha$.  Because the null law of $S_T$ is continuous and $q_{1-\alpha,\eta}$ is its exact $(1-\alpha)$ quantile, $\Pr(W_\eta>0)=\Pr(S_T>q_{1-\alpha,\eta})=\alpha$ for every $\eta$; since $\Pr(W=0)=0$ makes zero a continuity point of the limit law, $\Pr(W>0)=\lim_{\eta\to0}\Pr(W_\eta>0)=\alpha$, proving asymptotic size.

We next verify the scalar analytic calibration.  By \Cref{prop:null-distribution}(i),
$
S_T=\mu_0+\sigma_TZ+o_p(\sqrt\eta),
$
where $
\mu_0=2\kappa_\varepsilon\kappa_x$ and $
\sigma_T=\sqrt{2\eta D/\sigma_\varepsilon^2}.
$
It follows that
$q_{1-\alpha,\eta}
=\mu_0+z_{1-\alpha}\sigma_T+o(\sqrt\eta)$.
Therefore,
$\hat\mu_0-\mu_0=o_p(\sqrt\eta)$ and
$\hat\sigma_T/\sigma_T\to_p1$ imply the critical-value condition in
\Cref{cor:CI}.
\hfill$\square$

\subsection{Proof of \Cref{prop:H1-gap}}\label{app:H1-results}

Under $H_1$, there is no common $\theta^{\mathrm{true}}$ and hence no single
error midpoint $\bfe$.  We instead define the fitted parameter midpoint
$m:=(\theta_1+\theta_2)/2$, the parameter half-gap $\delta_\theta:=(\theta_1-\theta_2)/2$,
and $\alpha_z:=\theta_z^{\mathrm{true}}-m$.
For any integrable random variable $U$, write
$\E_z[U]:=\E[U\mid j=z]$ for expectation conditional on latent group~$z$.
Under $H_0$, this reduces to $\alpha_z=\theta^{\mathrm{true}}-m=\bfe$ and $\delta_\theta=-\delta$, where $\delta$ is the error-coordinate half-gap in \Cref{tab:notation}.
We first derive the ODE for $\delta_\theta$, then prove parts~(i)--(iii).

For group $z$, the two piece residuals are $r_{1,z}=\alpha_z^\top x+\varepsilon-\delta_\theta^\top x$ and $r_{2,z}=\alpha_z^\top x+\varepsilon+\delta_\theta^\top x$, so that $r_{2,z}^2-r_{1,z}^2=4(\delta_\theta^\top x)(\alpha_z^\top x+\varepsilon)$: piece~1 wins exactly when $(\delta_\theta^\top x)(\alpha_z^\top x+\varepsilon)>0$.  As in~\eqref{eq:ode-drift-expanded}, the uniform tie rule makes the two selection weights sum to one and differ by $\sgn\{(\delta_\theta^\top x)(\alpha_z^\top x+\varepsilon)\}$, with $\sgn(0)=0$.  Grouping the weights by their sum and difference, the winner-weighted residual difference entering $2\dot\delta_\theta=\dot\theta_1-\dot\theta_2$ collapses to
\[
\sgn\{(\delta_\theta^\top x)(\alpha_z^\top x+\varepsilon)\}\,(\alpha_z^\top x+\varepsilon)-\delta_\theta^\top x
=|\alpha_z^\top x+\varepsilon|\,\sgn(\delta_\theta^\top x)-\delta_\theta^\top x.
\]
Averaging over the joint distribution of $(x,j)$ and using
$\sum_z\pi_z\E_z[xx^\top]=\E[xx^\top]=\Sigma$ yields
\begin{equation}\label{eq:half-gap-H1}
	2\dot\delta_\theta = \sum_{z=1}^{K^*} \pi_z\,\E_z\bigl[|\alpha_z^\top x + \varepsilon|\,\operatorname{sgn}(\delta_\theta^\top x)\,x\bigr] - \Sigma\,\delta_\theta.
\end{equation}

We next prove the instability claim.  Let $(m_0,0)$ be a coincident
equilibrium and, for $\|v\|=1$, define
\[
\mathcal R(m,v):=\sum_z\pi_z\E_z\!\left[
|v^\top x|\,|x^\top(\theta_z^{\mathrm{true}}-m)+\varepsilon|
\right].
\]
Since $\Sigma\succ0$, $\Pr(v^\top x\ne0)>0$ for every unit $v$; conditional on $x$, the existence of the density in \Cref{ass:diffusion}(b) gives
$\Pr\{x^\top(\theta_z^{\mathrm{true}}-m_0)+\varepsilon=0\mid x,j=z\}=0$.
Hence $\mathcal R(m_0,v)>0$.  The moment assumptions and dominated convergence make $\mathcal R$ continuous in $(m,v)$, so compactness yields constants $c,\rho>0$ such that
$\mathcal R(m,v)\ge c$ whenever $\|m-m_0\|\le\rho$ and $\|v\|=1$.
For $\delta_\theta=rv$, $r>0$, projecting~\eqref{eq:half-gap-H1} onto $v$
gives
\[
2\dot r=\mathcal R(m,v)-r v^\top\Sigma v.
\]
Thus $\dot r\ge c/4$ whenever $\|m-m_0\|\le\rho$ and
$0<r\le c/(2\lambda_{\max}(\Sigma))$.  Any sufficiently small nonzero
antisymmetric perturbation must therefore leave a fixed neighborhood of
$(m_0,0)$ in finite time.  Hence no coincident equilibrium is stable.

We now prove part~(i).  In the scalar case, label the two pieces so that
$\delta_\theta>0$.  Then $\operatorname{sgn}(\delta_\theta x)x=|x|$, and the
equilibrium condition in \eqref{eq:half-gap-H1} becomes
\begin{equation}\label{eq:H1-scalar-eqm}
	\delta_\theta\,\sigma_x^2 = \sum_{z=1}^{K^*}\pi_z\,\E_z\bigl[|x|\cdot|\alpha_z x + \varepsilon|\bigr],
\end{equation}
where $\sigma_x^2=\E[x^2]$.  Define $g(\mu)=\E|\mu+W|$, where
$W=\varepsilon/\sigma_\varepsilon$.  Symmetry makes $g$ convex and even,
while continuity of $f_\varepsilon$ at zero and $f_\varepsilon(0)>0$ imply that zero
is its unique minimizer.  Hence $g(\mu)\ge g(0)=\kappa_\varepsilon$, with
equality only at $\mu=0$.  Conditioning jointly on $(x,j)$ and using that
$\varepsilon$ is independent of $(x,j)$ gives
\begin{align}
\sum_{z=1}^{K^*}\pi_z\E_z\bigl[|x|\cdot|\alpha_z x + \varepsilon|\bigr]
&=\sigma_\varepsilon\E\!\left[|x|\,
g\!\left(\frac{\alpha_jx}{\sigma_\varepsilon}\right)\right] \notag\\
&\ge \sigma_\varepsilon\kappa_\varepsilon\E|x|.
\label{eq:group-bound}
\end{align}
Consequently,
\[
\frac{S_\infty^{H_1}}2
=\frac{\delta_\theta\sigma_x}{\sigma_\varepsilon}
\ge\kappa_\varepsilon\kappa_x.
\]
The predictor-separation condition in~\Cref{ass:diffusion}(d), evaluated at
$b=m$, gives $\Pr(\alpha_jx\ne0)>0$.  Since $g$ is uniquely minimized at
zero, the inequality in~\eqref{eq:group-bound} is strict, proving
$S_\infty^{H_1}>2\kappa_\varepsilon\kappa_x$.

We next prove part~(ii).  Let $x\sim N(0,\Sigma)$ and define
$\tilde x=\Sigma^{-1/2}x$, $\tilde\delta_\theta=\Sigma^{1/2}\delta_\theta$,
$\hat v=\tilde\delta_\theta/\|\tilde\delta_\theta\|$,
$u=\hat v^\top\tilde x$, and
$\tilde\alpha_z=\Sigma^{1/2}\alpha_z$.  Then
$\tilde x\sim N(0,I)$ and $u\sim N(0,1)$.  Left-multiplying the
equilibrium equation by $\Sigma^{-1/2}$ and projecting onto $\hat v$
gives
\begin{equation}\label{eq:H1-multi-eqm}
	\frac{S_\infty^{H_1}}{2}
	= \frac{\|\delta_\theta\|_\Sigma}{\sigma_\varepsilon}
	= \frac{1}{\sigma_\varepsilon}
	\E\bigl[|u|\cdot|\tilde \alpha_j^\top\tilde x + \varepsilon|\bigr].
\end{equation}
Conditioning jointly on $(\tilde x,j)$ and applying $g(\mu)\ge g(0)$ gives
\[
\E\bigl[|u|\cdot|\tilde \alpha_j^\top\tilde x + \varepsilon|\bigr]
=\sigma_\varepsilon\E\!\left[|u|\,
g\!\left(\frac{\tilde\alpha_j^\top\tilde x}{\sigma_\varepsilon}\right)\right]
\ge\sigma_\varepsilon\kappa_\varepsilon\E|u|
=\sigma_\varepsilon\kappa_\varepsilon\sqrt{\frac2\pi}.
\]
The predictor-separation condition in~\Cref{ass:diffusion}(d) gives
$\Pr(\tilde\alpha_j^\top\tilde x\ne0)>0$.  The marginal Gaussian law of
$\tilde x$ also gives $\Pr(u=0)=0$.  The inequality is therefore strict,
and~\eqref{eq:H1-multi-eqm} yields
$S_\infty^{H_1}>2\kappa_\varepsilon\sqrt{2/\pi}$.
Finally, we prove part~(iii).  Let
$a^*=(\theta_1^*,\theta_2^*)$ be the separated attracting equilibrium
selected by the initialization and define
\[
G(a^*):=\frac{\|\theta_1^*-\theta_2^*\|_\Sigma}
{\sigma_\varepsilon}.
\]
Parts~(i)--(ii) give $G(a^*)>S_\infty^{H_0}$.  The terminal-concentration
condition in the proposition and nuisance consistency imply
$\hat S_T=G(a^*)+o_p(1)$, whereas the assumed critical-value convergence
gives $\hat q_{1-\alpha,\eta}=S_\infty^{H_0}+o_p(1)$.  Consequently,
\[
\hat S_T-\hat q_{1-\alpha,\eta}
=G(a^*)-S_\infty^{H_0}+o_p(1)>0
\]
with probability tending to one.  Therefore,
$\Pr(\textup{reject}\mid H_1)\to1$.
\hfill$\square$

\subsection{Proof of \Cref{prop:intercept}}
Put $a=\beta^{\mathrm{true}}-\hat\beta_{\mathrm{OLS}}$ and condition on the independent preliminary sample, so that $a$ is fixed.  The partial residual has the exact representation
\begin{equation}\label{eq:intercept-effective-noise}
r_t=\alpha_{j_t}^{\mathrm{true}}+\varepsilon_{a,t},
\qquad
\varepsilon_{a,t}:=\varepsilon_t+a^\top x_t.
\end{equation}
Thus the feasible algorithm is exactly a constant-design intercept recursion with effective noise $\varepsilon_a$.

We first verify the effective-noise conditions.  Central symmetry of $x$, symmetry of $\varepsilon$, and their independence imply that $\varepsilon_a$ is symmetric conditional on the preliminary sample.  Its conditional density is
\[
f_a(u)=\E\{f_\varepsilon(u-a^\top x)\mid a\}.
\]
It is bounded by $\|f_\varepsilon\|_\infty$, inherits the modulus of continuity of $f_\varepsilon$, and satisfies $f_a(0)\to_p f_\varepsilon(0)$.  With $M=\|f_\varepsilon\|_\infty$, symmetry also gives
\[
\big|\E(|\varepsilon+v|)-\E|\varepsilon|\big|\le Mv^2.
\]
Consequently,
\begin{align}
\sigma_a^2-\sigma_\varepsilon^2
&=a^\top\Sigma a=o_p(\eta), \label{eq:intercept-effective-variance}\\
\big|\E(|\varepsilon_a|\mid a)-\E|\varepsilon|\big|
&\le M a^\top\Sigma a=o_p(\eta), \label{eq:intercept-effective-absolute}
\end{align}
and hence $\kappa_a-\kappa_\varepsilon=o_p(\eta)$.  Moreover,
$2f_a(0)\E(|\varepsilon_a|\mid a)<1$ with probability tending to one.

For the intercept recursion, set the predictor identically equal to one; the calculation uses only the effective-noise conditions verified above.  Under $H_0^{(\alpha)}$, absorb the common intercept into the two intercept errors and write their midpoint and half-gap as $(\bfe,\delta)$.  Conditional on the preliminary sample, their mean ODE is
\begin{align*}
\dot\bfe
&=-\tfrac12\bfe+\tfrac12|\delta|\{2F_a(\bfe)-1\},\\
\dot\delta
&=-\tfrac12\delta+\tfrac12\sgn(\delta)
\E(|\varepsilon_a+\bfe|\mid a).
\end{align*}
The separated equilibria are $(0,\pm e_a^*)$, where
$e_a^*=\E(|\varepsilon_a|\mid a)$, and the Jacobian at $(0,e_a^*)$ is
\[
\begin{pmatrix}
-\tfrac12\{1-2f_a(0)e_a^*\}&0\\
0&-\tfrac12
\end{pmatrix}.
\]
The strict stability condition above makes this matrix Hurwitz.  Boundedness and uniform continuity of $f_a$, together with the finite second moment of $\varepsilon_a$, give the local Fr\'echet expansion, covariance continuity, and uniform integrability of the update squares used by the terminal martingale-array argument.  Repeating that argument for this constant-design recursion therefore yields its long-horizon terminal CLT.

At equilibrium, the common piece-noise variance is
\begin{align*}
D_{0,a}
&=\E[(\varepsilon_a-e_a^*)^2\ind\{\varepsilon_a\ge0\}\mid a]\\
&=\frac{\sigma_a^2}{2}(1-\kappa_a^2),
\end{align*}
where conditional symmetry gives the corresponding half-moment identities.  If
$S_{T,a}^{(\alpha)}:=|\alpha_{1,T}-\alpha_{2,T}|/\sigma_a$, the terminal CLT is
\begin{equation}\label{eq:intercept-effective-clt}
\frac{S_{T,a}^{(\alpha)}-2\kappa_a}
{\sqrt{\eta(1-\kappa_a^2)}}
\xrightarrow{d}N(0,1)
\qquad\text{conditionally in probability}.
\end{equation}
The assumed plug-in rates give
$\hat S_T^{(\alpha)}-S_{T,a}^{(\alpha)}=o_p(\sqrt\eta)$ and
\[
\hat q_{1-\alpha_{\mathrm I},\eta}^{(\alpha)}
-\{2\kappa_a+z_{1-\alpha_{\mathrm I}}\sqrt{\eta(1-\kappa_a^2)}\}
=o_p(\sqrt\eta).
\]
Slutsky's theorem applied conditionally to \eqref{eq:intercept-effective-clt}, followed by bounded convergence over the preliminary sample, proves
$\Pr(\hat S_T^{(\alpha)}>\hat q_{1-\alpha_{\mathrm I},\eta}^{(\alpha)}\mid H_0^{(\alpha)})\to\alpha_{\mathrm I}$.

Finally, under $H_1^{(\alpha)}$, \eqref{eq:intercept-effective-noise} remains an exact intercept-mixture representation.  At the selected equilibrium, let $c_{\alpha,a}:=(\alpha_1+\alpha_2)/2$ and $h_{\alpha,a}:=|\alpha_1-\alpha_2|/2>0$.  The half-gap equilibrium equation, averaged over the joint distribution of $(x,j)$, gives
\[
h_{\alpha,a}
=\E\bigl(
|\alpha_j^{\mathrm{true}}-c_{\alpha,a}+\varepsilon+a^\top x|
\mid a\bigr).
\]
To compare this equilibrium with its matched null, put $g_0(v):=\E|v+\varepsilon|$ and define
\[
\Delta_0:=
\inf_{c\in\mathbb R}
\left\{\E\bigl[g_0(\alpha_j^{\mathrm{true}}-c)\bigr]-g_0(0)\right\}.
\]
The function $g_0$ is continuous, coercive, and uniquely minimized at zero.  Since the latent intercepts are not all equal and their marginal probabilities are positive, the displayed objective attains a strictly positive minimum, so $\Delta_0>0$.  The elementary inequality $||u+v|-|u||\le|v|$ yields, uniformly in $c$,
\begin{align*}
\left|\E\bigl(|\alpha_j^{\mathrm{true}}-c+\varepsilon+a^\top x|\mid a\bigr)
-\E g_0(\alpha_j^{\mathrm{true}}-c)\right|
&\le \E(|a^\top x|\mid a),\\
\left|\E(|\varepsilon_a|\mid a)-g_0(0)\right|
&\le \E(|a^\top x|\mid a).
\end{align*}
Because $a=o_p(\sqrt\eta)$, Cauchy--Schwarz gives $\E(|a^\top x|\mid a)=o_p(1)$.  Hence, with probability tending to one,
\[
h_{\alpha,a}-\E(|\varepsilon_a|\mid a)\ge\frac{\Delta_0}{2},
\qquad
\frac{2h_{\alpha,a}}{\sigma_a}
\ge 2\kappa_a+\frac{\Delta_0}{\sigma_a}.
\]
Terminal concentration and the assumed plug-in and matched-null calibration then give
$\Pr(\hat S_T^{(\alpha)}>\hat q_{1-\alpha_{\mathrm I},\eta}^{(\alpha)}\mid H_1^{(\alpha)})\to1$, which completes the proof.
\hfill$\square$

\subsection{Proof of \Cref{cor:pooled-power}}
The concentration condition and the stated plug-in convergences give
\[
\hat S_T^{\mathrm{pool}}
=
\frac{\|\theta_1^*-\theta_2^*\|_\Sigma}
{\sigma_{\mathrm{pool}}}
+o_p(1)
=S_\infty^{\mathrm{pool}}+o_p(1),
\]
whereas
$\hat q_{1-\alpha,\eta}^{\mathrm{pool}}
=q_\infty^{\mathrm{pool}}+o_p(1)$.
The strict-margin condition therefore implies
$\Pr(\hat S_T^{\mathrm{pool}}>
\hat q_{1-\alpha,\eta}^{\mathrm{pool}}\mid H_1)=\Pr(\textup{reject}\mid H_1)\to1$.
\hfill$\square$

\section{Proofs for Results in \Cref{sec:extensions}}

\subsection{Proof of \Cref{prop:partial}}\label{app:partial}

Write $B=B_\ell$, $z=x_{-\ell}-Bx_\ell$,
$\vartheta_j=\theta_j+\psi^\top B$, and
$\bar\vartheta_w=\E(x_\ell^2\vartheta_j)/\E(x_\ell^2)$.  Then the partial-mixture model has the
exact reparametrization
\[
y=\psi^\top z+\vartheta_jx_\ell+\varepsilon.
\]
Let $\hat B$ and $\hat\psi$ be fixed at their preliminary-sample values,
put $d=B-\hat B$ and $a=\psi-\hat\psi$, and note that
$\hat z=x_{-\ell}-\hat Bx_\ell=z+dx_\ell$.  Direct substitution gives
\begin{equation}\label{eq:partial-exact-reduction}
\tilde y:=y-\hat\psi^\top\hat z
=\{\vartheta_j-\hat\psi^\top d\}x_\ell+\varepsilon_a,
\qquad
\varepsilon_a:=\varepsilon+a^\top z.
\end{equation}
Thus estimating $B$ changes every focal slope by the same amount and leaves
all slope contrasts unchanged.

Conditional on the independent preliminary sample, $a$ and $d$ are fixed.
By the assumptions of the proposition and \Cref{ass:diffusion},
$\varepsilon_a$ is independent of $(x_\ell,j)$ and is symmetric.  Its density
is
\[
f_a(u)=\E\{f_\varepsilon(u-a^\top z)\}.
\]
Consequently $\sup_u f_a(u)\leq\sup_u f_\varepsilon(u)$, and the modulus of
continuity of $f_a$ is bounded by that of $f_\varepsilon$.  Since $a=o_p(1)$,
dominated convergence gives $f_a(0)\to_p f_\varepsilon(0)>0$.  The stated
moment assumptions also give the scalar moment bounds uniformly for $a$ in
a neighborhood of zero.  Hence, with probability tending to one,
\eqref{eq:partial-exact-reduction} satisfies the scalar version of
\Cref{ass:diffusion}, with regularity constants uniform along the triangular
array $a=o_p(\sqrt\eta)$.

The effective noise moments are also equivalent at the required scale.
Symmetry and boundedness of $f_\varepsilon$ imply, for
$h(v):=\E|\varepsilon+v|$,
\[
|h(v)-h(0)|
\leq \|f_\varepsilon\|_\infty v^2.
\]
Because $z$ is centered and independent of $\varepsilon$,
\[
\E(\varepsilon_a^2)-\E(\varepsilon^2)
=a^\top\E(zz^\top)a,
\qquad
\bigl|\E|\varepsilon_a|-\E|\varepsilon|\bigr|
=O(\|a\|^2)=o_p(\eta).
\]
Thus the feasible residuals calibrate the effective scalar null center and
scale at the required order.

The preliminary regression on $(\hat z,x_\ell)$ targets $\psi$ and the
weighted focal slope $\bar\vartheta_w-\psi^\top d$.  The corresponding
$x_\ell^2$-weighted average slope in \eqref{eq:partial-exact-reduction} is
$\bar\vartheta_w-\hat\psi^\top d$; their difference is $a^\top d=o_p(\sqrt\eta)$.
Together with the assumed preliminary rate, this verifies the scalar
initialization condition.  We may therefore apply
\Cref{prop:null-distribution,cor:CI} conditionally to the feasible scalar
model in \eqref{eq:partial-exact-reduction}.  The uniform regularity just
verified makes the conditional convergence valid along the preliminary-
sample triangular array, and bounded convergence yields part~(i)
unconditionally.  Under $H_1$, independence and centering of $z$ give the orthogonal decomposition
\[
\inf_{q,b}\E\!\left[
\{z^\top(\psi-q)+x_\ell(\vartheta_j-b)\}^2
\right]
=\inf_b\E\{x_\ell^2(\vartheta_j-b)^2\}>0,
\]
where strict positivity follows from \Cref{ass:diffusion}(d).  The common shift $-\hat\psi^\top d$ leaves this scalar predictor separation unchanged.  Thus \Cref{prop:H1-gap}(i) gives the strict scalar
gap, while the stated terminal-concentration and calibration conditions give
part~(ii) by the same continuous-mapping argument used in the proof of
\Cref{prop:H1-gap}(iii).
\hfill$\square$

\subsection{Proof of \Cref{prop:endogeneity}}

To prove the result, it suffices to show that the best-linear-projection transformation reduces the endogenous model to one satisfying \Cref{ass:diffusion}, and then transfer the null and alternative conclusions to the transformed model.  For notational convenience, let $b=\Sigma^{-1}\gamma$.  Under either hypothesis,
\[
y_t=x_t^\top(\theta_{j_t}^{\mathrm{true}}+b)+\tilde\varepsilon_t.
\]
Thus endogeneity adds the same shift $b$ to every component coefficient.

We first verify that $(x_t,\tilde\varepsilon_t)$ satisfies the noise conditions in \Cref{ass:diffusion} with $\sigma_\varepsilon$ replaced by $\tilde\sigma_\varepsilon$.

(a)~\textit{Independence:}\;
  $\mathrm{Cov}(x_t,\tilde\varepsilon_t)
  = \mathrm{Cov}(x_t,\varepsilon_t) - \mathrm{Cov}(x_t,x_t)\Sigma^{-1}\gamma
  = \gamma - \gamma
  = 0.$
  Since $(x_t, \varepsilon_t)$ is jointly Gaussian and $\tilde\varepsilon_t$ is an affine
  function of $(x_t, \varepsilon_t)$, uncorrelatedness implies that
  $\tilde\varepsilon_t$ is independent of $x_t$.
  Under $H_1$, the additional assumption that $j_t$ is independent of
  $(x_t,\varepsilon_t)$ also makes $\tilde\varepsilon_t$ independent of $j_t$.

(b)~\textit{Distribution and symmetry:}\;
  $\Var(\tilde\varepsilon_t)
    = \sigma_\varepsilon^2
      -2\gamma^\top\Sigma^{-1}\gamma
      +\gamma^\top\Sigma^{-1}\Sigma\Sigma^{-1}\gamma
    = \sigma_\varepsilon^2 - \gamma^\top\Sigma^{-1}\gamma
    = \tilde\sigma_\varepsilon^2 > 0$,
  and $\tilde\varepsilon_t \sim N(0, \tilde\sigma_\varepsilon^2)$ is symmetric with
  $f_{\tilde\varepsilon}(0) > 0$.

(c)~\textit{Moment conditions:}\;
  All moments of jointly Gaussian random variables are finite,
  so \Cref{ass:diffusion}(c) holds.

We next prove part~(i).  Under $H_0$, define
$\theta_{\mathrm{BLP}}=\theta^{\mathrm{true}}+b$.
Since the transformed system satisfies \Cref{ass:diffusion}, by applying \Cref{thm:eqm} and \Cref{prop:null-distribution} with
$(\varepsilon,\sigma_\varepsilon,\theta^{\mathrm{true}})
\mapsto(\tilde\varepsilon,\tilde\sigma_\varepsilon,\theta_{\mathrm{BLP}})$, the null dynamics carry over unchanged.
The standard initialization (\Cref{prop:init}) also carries over, since
$\hat\theta_{\mathrm{OLS}}\xrightarrow{p}\theta_{\mathrm{BLP}}$.
Note that under joint Gaussianity, $\tilde\kappa_\varepsilon = \sqrt{2/\pi}$ and $\tilde\kappa_x = \sqrt{2/\pi}$, so $2\tilde\kappa_\varepsilon\tilde\kappa_x = 4/\pi$.
Moreover, either the OLS-residual variance or the component-scale estimator consistently estimates $\tilde\sigma_\varepsilon^2$ under $H_0$, and both have the usual root-$n_0$ rate under the stated Gaussian moment conditions; since $n_0\asymp T$ and $T\eta\to\infty$, their errors are $o_p(\sqrt\eta)$, as required by \Cref{cor:CI}.  By applying Slutsky's theorem, we obtain part~(i).

We then verify the noise scale under $H_1$.  For notational convenience, let
$\bar\theta=\E[\theta_j^{\mathrm{true}}]$ and
$\Omega_\theta=\Var(\theta_j^{\mathrm{true}})$.
Independence of $j_t$ from $(x_t,\varepsilon_t)$ implies that the population pooled BLP coefficient is $\bar\theta+b$.  Its residual is
\[
R=y-x^\top(\bar\theta+b)
  =x^\top(\theta_j^{\mathrm{true}}-\bar\theta)+\tilde\varepsilon.
\]
Conditioning on $x$ and using the independence and centering of
$\theta_j^{\mathrm{true}}-\bar\theta$ gives
\[
\E[R^2\mid x]
=\tilde\sigma_\varepsilon^2+x^\top\Omega_\theta x.
\]
Hence, under the identification and moment conditions in
\Cref{sec:scale-estimation}, the intercept in the cross-fitted variance function
regression converges to $\tilde\sigma_\varepsilon^2$ under $H_1$ as well.

We finally prove part~(ii).  Note that the transformed component coefficients are
$\theta_j^{\mathrm{true}}+b$, whose pairwise contrasts are identical to those
of the original coefficients.  Since $x$ and $\tilde\varepsilon$ are Gaussian
and $j_t$ is independent of $x_t$, by applying \Cref{prop:H1-gap}(i)--(ii), we obtain, at any separated
attracting equilibrium selected by the initialization,
\[
\frac{\|\theta_1-\theta_2\|_\Sigma}{\tilde\sigma_\varepsilon}
>\frac4\pi.
\]
By applying \Cref{prop:H1-gap}(iii) with the consistent projected-noise scale and
the stated concentration condition, we obtain
$\Pr(\textup{reject}\mid H_1)\to1$, which completes the proof.
\hfill$\square$

\subsection{Proof of \Cref{prop:heterosk}}

We first work under $H_0$, where $y_t=x_t\theta^{\mathrm{true}}+\varepsilon_t$.  Write $\sigma_x^2=\E[x^2]$ and $\sigma_\varepsilon^2=\E[\varepsilon^2]$ (the marginal noise variance).

We first prove Part~(i).  On the antisymmetric regime $e_1 = -e_2 = \delta$ with $\delta > 0$, the tie-aware winner weight is $\ind\{\varepsilon x<0\}+\tfrac12\ind\{\varepsilon x=0\}$.  It is expectation-equivalent in the drift to $\ind\{\varepsilon x\le0\}$: the conditional density rules out $\varepsilon=0$, while observations with $x=0$ contribute zero to the drift.  Hence
\[
\mu_1(\delta,-\delta) = -\E[\varepsilon x\,\ind\{\varepsilon x \le 0\}] - \delta\,\E[x^2\,\ind\{\varepsilon x \le 0\}].
\]
We evaluate the two terms separately.  For the first term, conditioning on $x$: for $x > 0$, $\ind\{\varepsilon x \le 0\} = \ind\{\varepsilon \le 0\}$, and conditional symmetry gives $\E[\varepsilon\,\ind\{\varepsilon \le 0\}\mid x] = -\tfrac{1}{2}\E[|\varepsilon|\mid x]$, so $\E[\varepsilon x\,\ind\{\varepsilon x\le 0\}\mid x] = -\tfrac{1}{2}|x|\,\E[|\varepsilon|\mid x]$.  The same identity holds for $x < 0$ (both the indicator and the sign of $x$ flip).  By the tower property,
\[
-\E[\varepsilon x\,\ind\{\varepsilon x \le 0\}] = \tfrac{1}{2}\E\bigl[|x|\,\E[|\varepsilon|\mid x]\bigr] = \tfrac{1}{2}\E[|\varepsilon|\,|x|].
\]
For the second term, since $\Pr(\varepsilon \le 0 \mid x) = \tfrac{1}{2}$ by conditional symmetry, $\E[x^2\,\ind\{\varepsilon x \le 0\}\mid x] = \tfrac{1}{2}x^2$, so $\E[x^2\,\ind\{\varepsilon x \le 0\}] = \tfrac{1}{2}\sigma_x^2$.  Combining: $\mu_1(\delta,-\delta) = \tfrac{1}{2}\E[|\varepsilon|\,|x|] - \tfrac{1}{2}\sigma_x^2\,\delta$.  Setting the drift to zero yields the equilibrium
\[
\delta^* = e^* = \frac{\E[|\varepsilon|\,|x|]}{\sigma_x^2}.
\]
The derivative of the antisymmetric half-gap drift at $\delta=e^*$ is $-\sigma_x^2/2$.  To determine full-state stability, consider the midpoint $\bfe$.  Conditional symmetry gives the midpoint drift
\[
\dot\bfe
=-\frac12\sigma_x^2\bfe
+\frac12\E\!\left[|\delta x|\,\sgn(\varepsilon+\bfe x)x\right].
\]
We verify the required derivative from the primitive conditions.  Let $M$ uniformly bound the conditional densities and, using conditional symmetry, define
\[
h_x(t):=\E[\sgn(\varepsilon+t)\mid x]
=2F_{\varepsilon\mid x}(t\mid x)-1.
\]
For $P_x$-almost every $x$,
$h_x(t)/t\to2f_{\varepsilon\mid x}(0\mid x)$ and
$|h_x(t)|\le2M|t|$.  The resulting difference quotient for the midpoint drift is dominated, uniformly for $\delta$ near $e^*>0$, by a constant multiple of $|x|^3$, which is integrable under \Cref{ass:diffusion}(c).  Dominated convergence therefore gives the symmetric-mode eigenvalue
\[
\lambda_{\bfe}
=-\frac12\sigma_x^2
+e^*\E\!\left[f_{\varepsilon\mid x}(0\mid x)|x|^3\right]
=-\frac12\sigma_x^2(1-\rho_{\mathrm{het}}),
\]
where $\rho_{\mathrm{het}}:=2\E[|\varepsilon|\,|x|]\E[f_{\varepsilon\mid x}(0\mid x)|x|^3]/\{\E[x^2]\}^2$.  To verify the remaining derivative, put $a_x(t):=\E[|\varepsilon+t|\mid x]$.  Conditional symmetry and the preceding density bound give $a_x'(0)=0$ and
$|a_x(t)-a_x(0)|\le M t^2$.  Hence the perturbation of the half-gap drift due to $\bfe$ is $O(\bfe^2\E|x|^3)$, uniformly for $\delta$ near $e^*$; the dependence of the midpoint derivative on $\delta$ contributes only $O(|\bfe|\,|\delta-e^*|)$.  These bounds establish a joint Fr\'echet expansion at $(0,e^*)$ and show that both cross-derivatives vanish.  The full Jacobian in $(\bfe,\delta)$ coordinates is therefore diagonal with eigenvalues $\lambda_{\bfe}$ and $\lambda_\delta=-\sigma_x^2/2$.  The differentiable-at-the-point Lyapunov argument used in the proof of \Cref{thm:eqm}(i) then shows that $\rho_{\mathrm{het}}<1$ is sufficient for local exponential attraction.  Under the stated concentration condition, the continuous mapping theorem gives $S_T\xrightarrow{p}S_\infty^{H_0}$, where
\[
S_\infty^{H_0}
=\frac{2e^*\sigma_x}{\sigma_\varepsilon}
=\frac{2\,\E[|\varepsilon|\,|x|]}{\sigma_\varepsilon\,\sigma_x}.
\]
When the conditional law does not depend on $x$, factorization gives $S_\infty^{H_0}=2\kappa_\varepsilon\kappa_x$.  The definition of $\rho_{\mathrm{het}}$ likewise reduces to $\alpha_\varepsilon\alpha_x$.

Decomposing $\E[|\varepsilon|\,|x|]$ into the product of its marginal means and their covariance gives
\[
S_\infty^{H_0}-2\kappa_\varepsilon\kappa_x
=\frac{2\,\mathrm{Cov}(|\varepsilon|,|x|)}{\sigma_\varepsilon\,\sigma_x}.
\]
Consistent fitting of the homoskedastic reference means $\hat S_T\xrightarrow{p}S_\infty^{H_0}$ and $\hat q_{1-\alpha,\eta}^{\mathrm{hom}}\xrightarrow{p}2\kappa_\varepsilon\kappa_x$.  A fixed positive covariance therefore makes their difference converge to a positive constant, so $\Pr(\textup{reject}\mid H_0)\to1$; a fixed negative covariance gives $\Pr(\textup{reject}\mid H_0)\to0$.  When the covariance is zero, the leading centers agree and size depends on the first-order fluctuation laws.

Finally, we prove part~(ii).  Since $j$ is independent of $(x,\varepsilon)$, consider the fitted midpoint $m$, parameter half-gap $\delta_\theta$ (labeled so that $\delta_\theta>0$), and group offsets $\alpha_z$ of the selected separated attracting equilibrium (\Cref{tab:notation}).  Its population half-gap equation is
\[
\delta_\theta\,\sigma_x^2 = \sum_{z=1}^{K^*}\pi_z\,\E_z\bigl[|x|\cdot|\alpha_z x + \varepsilon|\bigr].
\]
For each group $z$, condition on $x$ and define $g_x(\mu) = \E[|\mu + \varepsilon|\mid x]$.  By \Cref{ass:heterosk}, $g_x(\mu)>g_x(0)=\E[|\varepsilon|\mid x]$ whenever $\mu\ne0$ for almost every~$x$.  Therefore $g_x(\alpha_z x) \ge g_x(0)$ for almost every~$x$, with strict inequality when $\alpha_z x \neq 0$.  Multiplying by $|x|$ and taking the outer expectation:
\[
\E_z\bigl[|x|\cdot|\alpha_z x + \varepsilon|\bigr] = \E\bigl[|x|\,g_x(\alpha_z x)\bigr] \ge \E\bigl[|x|\,\E[|\varepsilon|\mid x]\bigr] = \E[|\varepsilon|\,|x|],
\]
with equality if and only if $\alpha_zx=0$ almost surely.  Defining $S_\infty^{H_1}:=2\delta_\theta\sigma_x/\sigma_\varepsilon$, summing over groups gives $S_\infty^{H_1}\ge S_\infty^{H_0}$.  Under $H_1$, at least one $\alpha_{z_0}$ is nonzero; since $\sigma_x^2>0$ and every group has the same covariate law, $\Pr_{z_0}(x\ne0)>0$.  The inequality is therefore strict for group $z_0$, proving $S_\infty^{H_1}>S_\infty^{H_0}$.  Concentration gives $\hat S_T\xrightarrow{p}S_\infty^{H_1}$, while consistent heteroskedastic-null calibration gives $\hat q_{1-\alpha,\eta}^{\mathrm{het}}\xrightarrow{p}S_\infty^{H_0}$.  Their difference converges to this fixed positive gap, so $\Pr(\textup{reject}\mid H_1)\to1$.
\hfill$\square$

\section{Technical Foundations}

This section collects the supporting technical results invoked by the proofs above: the localized diffusion approximation underlying the analysis outlined in \Cref{app:diffusion-approximation}, followed by the auxiliary lemmas.

\subsection{Localized Diffusion Approximation}\label{app:localized-diffusion-results}

\begin{lemma}\label{lem:local-diffusion-regularity}
Under $H_0$ and Parts~(a)--(c) of \Cref{ass:diffusion}, write the recursion in midpoint--half-gap coordinates $z=(\bfe,\delta)$ and let $\mathcal C:=\{(\bfe,\delta):\delta=0\}$ be the coincident set.  Consider either the scalar setting of \Cref{thm:eqm}(i) or the Gaussian multivariate setting of \Cref{thm:eqm}(ii).

\textup{(i)} On every compact $K\subset\mathbb R^{2d}\setminus\mathcal C$, the averaged drift $\mu_z$ is locally Lipschitz.  The innovation covariance $Q_z(z)$ is continuous on $K$, and the squared innovation norms are uniformly integrable over $z\in K$.

\textup{(ii)} At the scalar separated equilibria, and uniformly along the compact multivariate equilibrium manifold $\{(0,e):e\in\mathcal E_2\}$, the drift has the Fr\'echet expansion
\[
\mu_z(z+h)=\mu_z(z)+J_z(z)h+r_z(h),
\qquad
\sup_{z}\frac{\|r_z(h)\|}{\|h\|}\longrightarrow0
\quad\text{as }\|h\|\to0.
\]
The matrices $J_z(z)$ are the block Jacobians computed in \Cref{lem:equiv,lem:ebar-multi,lem:jacobian-block} and the proof of \Cref{thm:eqm}.

\textup{(iii)} For every fixed $c\ne0$, the ODE trajectory starting from $(\bfe_0,\delta_0)=(0,-c)$ is unique, remains in $\mathbb R^{2d}\setminus\mathcal C$, and stays a positive distance from $\mathcal C$.  It converges to the corresponding scalar separated equilibrium or to the multivariate equilibrium manifold.
\end{lemma}

\textit{Proof of \Cref{lem:local-diffusion-regularity}.}
The midpoint and half-gap drifts can be written, for both $d=1$ and $d\ge2$, as
\begin{align}
\dot{\bfe}
&=-\tfrac12\Sigma\bfe+\tfrac12\Psi(\bfe,\delta),
&\Psi(\bfe,\delta)
&:=\E[|\delta^\top x|\,\sgn(\varepsilon+\bfe^\top x)x],
\label{eq:local-midpoint-drift}\\
\dot\delta
&=-\tfrac12\Sigma\delta+\tfrac12\Phi(\bfe,\delta),
&\Phi(\bfe,\delta)
&:=\E[|\varepsilon+\bfe^\top x|\,\sgn(\delta^\top x)x].
\label{eq:local-gap-drift}
\end{align}
These identities follow from $p_1+p_2=1$ and
$p_1-p_2=-\sgn\{(\varepsilon+\bfe^\top x)(\delta^\top x)\}$, with $\sgn(0)=0$.

For the scalar case, \Cref{lem:equiv} gives the equivalent explicit formulas.  Let $M=\|f_\varepsilon\|_\infty$.  If $u,v$ have the same sign, conditioning on $x$ gives
\[
|\psi(u)-\psi(v)|
\le 2M|u-v|\E|x|^3;
\]
the same bound across zero follows by oddness and the triangle inequality.  Also,
\[
\big|\E[|\varepsilon+ux||x|]-\E[|\varepsilon+vx||x|]\big|
\le |u-v|\E[x^2].
\]
Because $|\delta|$ and $\sgn(\delta)$ are Lipschitz on compact sets separated from zero, the scalar drift is locally Lipschitz there.

For Gaussian $x$, the map $\Psi$ is locally Lipschitz directly from
$|2F_\varepsilon(s)-2F_\varepsilon(t)|\le2M|s-t|$ and the Gaussian moment bounds.  The map $\Phi$ is Lipschitz in $\bfe$ because $a(t):=\E|\varepsilon+t|$ is one-Lipschitz.  Its dependence on $\delta$ is also locally Lipschitz away from zero.  Indeed, write $x=\Sigma^{1/2}Z$ with $Z\sim N(0,I_d)$ and let $\vartheta$ be the angle between $\Sigma^{1/2}\delta_1$ and $\Sigma^{1/2}\delta_2$.  The signs can differ only in two spherical wedges of total angular measure proportional to $\vartheta$.  Radial--angular independence of $Z$, followed by integration of the polynomial radial weight, gives
\[
\E\!\left[(1+\|x\|^2)
\ind\{\sgn(\delta_1^\top x)\ne\sgn(\delta_2^\top x)\}\right]
\le C\vartheta
\le C_K\|\delta_1-\delta_2\|,
\]
where the last inequality uses $\inf_{\delta\in K}\|\Sigma^{1/2}\delta\|>0$.
This proves local Lipschitzness of the multivariate drift.

We next verify the first-order expansion.  With
$h(t):=2F_\varepsilon(t)-1$, continuity of $f_\varepsilon$ at zero gives
$h(t)=2f_\varepsilon(0)t+o(t)$, while boundedness of the density gives
$|h(t)|\le2M|t|$.  Dominated convergence, using
$\E\|x\|^4<\infty$, therefore yields, uniformly for $\delta$ in a compact set,
\[
\Psi(\bfe,\delta)
=2f_\varepsilon(0)\E[|\delta^\top x|xx^\top]\bfe
+o(\|\bfe\|).
\]
The coefficient matrix is locally Lipschitz in $\delta$ on compact sets, so replacing $\delta$ by an equilibrium value $e$ contributes only
$O(\|\delta-e\|\,\|\bfe\|)=o(\|(\bfe,\delta-e)\|)$.
Moreover, symmetry gives $a'(0)=0$, and
$|a(t)-a(0)|\le M t^2$ by integrating $h$.  Hence the derivative of $\Phi$ with respect to $\bfe$ at $(0,\delta)$ is zero, with a remainder bounded by
$M\|\bfe\|^2\E\|x\|^3$.  In the scalar case the remaining $\delta$ dependence is linear on each separated branch.  In the Gaussian multivariate case,
\[
\Phi(0,\delta)
=\E|\varepsilon|\sqrt{\frac{2}{\pi}}
\frac{\Sigma\delta}{\|\delta\|_\Sigma},
\]
which is continuously differentiable for $\delta\ne0$.  Compactness of $\mathcal E_2$ makes these remainders uniform along the equilibrium manifold.  This proves the uniform Fr\'echet expansion in part~(ii).

Let $H_z(z;x,\varepsilon,w)$ denote one unscaled midpoint--half-gap update and let $Q_z(z)=\operatorname{Cov}\{H_z(z;x,\varepsilon,w)\}$.  If $z_n\to z\notin\mathcal C$, the winner indicators converge almost surely except on the switching boundary.  For Gaussian $x$ this boundary has probability zero; in the scalar case, the possible event $x=0$ contributes a zero update, while all other ties have probability zero because the noise has a density.  On a compact $K$,
\[
\|H_z(z;x,\varepsilon,w)\|^2
\le C_K\{\varepsilon^2\|x\|^2+\|x\|^4\},
\]
whose expectation is finite by \Cref{ass:diffusion}(c).  Dominated convergence proves continuity of $Q_z$, and the same envelope gives uniform integrability.

Finally, $\bfe=0$ is invariant.  In the scalar case, the explicit solution of \eqref{eq:delta-ode} starting from $-c\ne0$ stays on its initial sign branch and its absolute value remains between $|c|$ and $e^*$ while converging to $e^*$.  In the Gaussian multivariate case, \eqref{eq:multi-drift} gives
\[
\dot q=\tfrac12(c_0\sigma_\varepsilon-q)
\frac{\delta^\top\Sigma^2\delta}{\delta^\top\Sigma\delta},
\qquad q=\|\delta\|_\Sigma.
\]
Thus $q$ remains between $\|c\|_\Sigma$ and $c_0\sigma_\varepsilon$ and converges monotonically to the latter.  The local Lipschitz result gives uniqueness along both trajectories.  This proves part~(iii). \hfill$\square$

\begin{proposition}\label{prop:localized-diffusion}
Under the conditions of \Cref{lem:local-diffusion-regularity}, let the initial state be measurable with respect to an initial sigma-field independent of the detection innovations, and consider compact neighborhoods separated from $\mathcal C$.

\textup{(i)} For any such neighborhood $U$, on every fixed time interval the interpolation stopped at its first exit from $U$ converges weakly to the unique stopped solution of the mean ODE.

\textup{(ii)} If $z^*$ is a separated equilibrium and
$\eta^{-1/2}(z_0^\eta-z^*)\Rightarrow\zeta_0$, then the stopped fluctuation process converges on every fixed time interval to
\[
d\zeta=J_z(z^*)\zeta\,d\tau+Q_z(z^*)^{1/2}dW_\tau.
\]
For the multivariate equilibrium manifold, the same conclusion holds in local normal--tangent coordinates, uniformly in $z^*$ on the manifold.

\textup{(iii)} Suppose the initialization is measurable with respect to an initial sigma-field independent of the detection-stream innovations and converges to $(0,-c)$ for a fixed $c\ne0$.  Suppose also that the corresponding stability condition in \Cref{thm:eqm} holds: $\alpha_\varepsilon\alpha_x<1$ in the scalar case or $\alpha_\varepsilon<\pi/4$ in the Gaussian multivariate case.  Let $T=T_\eta$ satisfy
\[
\frac{\eta T}{\log(1/\eta)}\longrightarrow\infty,
\qquad
\eta^2T\longrightarrow0.
\]
For some compact tube $U$ separated from $\mathcal C$ and containing the ideal trajectory and its attractor, $\Pr\{z_t\notin U\text{ for some }t\le T\}\to0$.  The contracting coordinates attain the stationary law of their limiting OU process.  In the scalar case this applies to the full fluctuation.  In the multivariate case the unscaled tangent displacement is $o_p(1)$, and the conclusion applies to every continuously differentiable functional that is constant on the equilibrium manifold, including $\|\delta\|_\Sigma$.
\end{proposition}

\textit{Proof of \Cref{prop:localized-diffusion}.}
Write the recursion as
$z_{t+1}=z_t+\eta\{\mu_z(z_t)+\zeta_{t+1}\}$, where
$\E(\zeta_{t+1}\mid\mathcal F_t)=0$.  Part~(i) follows from the martingale weak-convergence theorem for constant-step stochastic approximation \citep[Section~8.2.1]{kushner2003}: local Lipschitzness identifies a unique stopped ODE, uniform integrability gives the conditional Lindeberg condition, and continuity of the drift identifies the limiting integral equation.

For part~(ii), substitute $z_t=z^*+\sqrt\eta\,u_t$.  The uniform Fr\'echet expansion gives
\[
u_{t+1}-u_t
=\eta J_z(z^*)u_t+\sqrt\eta\,\zeta_{t+1}
+\eta\rho_t^\eta,
\qquad
\sup_{t\le L/\eta}\|\rho_t^\eta\|\xrightarrow{p}0
\]
after stopping, for every fixed $L$.  Continuity of $Q_z$ identifies the predictable quadratic variation, and uniform integrability verifies the martingale Lindeberg condition.  The martingale functional central limit theorem and uniqueness of the linear SDE give part~(ii), as in the constant-step rate theorem of \citet[Section~10.1.1]{kushner2003}.

For part~(iii), \Cref{lem:local-diffusion-regularity}(iii) and finite-time ODE tracking first bring the recursion into a normally attracting neighborhood that remains separated from $\mathcal C$.  In the scalar case there are only contracting coordinates.  In the multivariate case, choose $n$ normal to the equilibrium manifold and let $s$ be the asymptotic phase of the smooth restricted half-gap field
$g(\delta)=\tfrac12(c_0\sigma_\varepsilon/\|\delta\|_\Sigma-1)\Sigma\delta$ on $\bfe=0$; thus $s$ is constant along each deterministic stable fiber of this field.  The uniform transverse Lyapunov estimate established in the proof of \Cref{thm:eqm}, together with the uniform expansion, gives
\[
n_{t+1}=(I+\eta J_N(s_t))n_t+\eta\zeta^N_{t+1}
+\eta r^N_t,
\]
where the eigenvalues of $J_N(s)$ are uniformly bounded in the open left half-plane and
$r^N_t=o(\|n_t\|)+O(\|n_t\|^2)$.  Uniform transverse stability supplies a quadratic function $L(n,s)$ and constants $a,C>0$ such that, up to exit from a sufficiently small tube,
\[
\E[L(n_{t+1},s_{t+1})-L(n_t,s_t)\mid\mathcal F_t]
\le-a\eta L(n_t,s_t)+C\eta^2.
\]
Iteration gives $\E L(n_t,s_t)=O(\eta)$ after the contracting burn-in.  Applying the same inequality at the stopped exit time, together with the martingale maximal inequality, uniform integrability of the update squares, and accumulated quadratic variation $O(\eta^2T)=o(1)$, shows that the recursion remains in the fixed tube with probability tending to one and that
$\operatorname{dist}(z_T,\mathcal M)=O_p(\sqrt\eta)$.  (The distance need not be small during the initial deterministic approach.)  By \eqref{eq:local-gap-drift} and the quadratic bound in the proof of \Cref{lem:local-diffusion-regularity}(ii), the full half-gap drift differs from $g(\delta)$ by $O(\|\bfe\|^2)$.  A Taylor expansion of $s(\delta_{t+1})-s(\delta_t)$ therefore leaves a tangent martingale with second moment $O(\eta^2T)=o(1)$, an accumulated $O(\eta\sum_{t\le T}\|\bfe_t\|^2)=O_p(\eta^2T)=o_p(1)$ drift after burn-in, and a second-order update remainder of the same order; truncation plus the uniform-integrability bound controls large updates.  The martingale maximal inequality gives the corresponding uniform bound, so for the phase $s_*$ selected by the ideal trajectory,
\[
\sup_{t_b\le t\le T}\|s_t-s_*\|=o_p(1)
\]
after an entry time $t_b=o(T)$.

For the long-horizon terminal law, iterate the normal recursion from the burn-in time $t_b$ to $T$.  Set
$A:=J_N(s_*)$, $Q_N:=\operatorname{Cov}(\zeta^N_{t+1}\mid z_t\in\mathcal M,s_t=s_*)$, and $\Phi_\eta(j):=(I+\eta A)^j$.  Iteration of the normal recursion from $t_b$ to $T$, freezing the coefficients at $s_*$ and collecting the coefficient and nonlinear errors, gives
\begin{equation}\label{eq:terminal-normal-array}
\eta^{-1/2}n_T
=\eta^{-1/2}\Phi_\eta(T-t_b)n_{t_b}
+\sqrt\eta\sum_{j=0}^{T-t_b-1}\Phi_\eta(j)\zeta^N_{T-j}
+R_\eta,
\end{equation}
where $R_\eta=o_p(1)$.  Indeed, the phase bound above and continuity of $J_N$ and $Q_N$ justify freezing the coefficients, while the uniform Fr\'echet remainder and $\E L(n_t,s_t)=O(\eta)$ make the geometrically weighted nonlinear contribution $o_p(1)$ after normalization.

Because $A$ is uniformly Hurwitz, $\|\Phi_\eta(j)\|\le C\exp(-c\eta j)$ for all sufficiently small $\eta$.  Consequently,
\[
\eta^{-1/2}\Phi_\eta(T-t_b)n_{t_b}
=O_p\!\left\{\frac{\exp[-c\eta(T-t_b)]}{\sqrt\eta}\right\}
=o_p(1),
\]
where $\eta T/\log(1/\eta)\to\infty$ is used.  The predictable covariance of the martingale array in \eqref{eq:terminal-normal-array} converges as
\begin{align}
\eta\sum_{j=0}^{T-t_b-1}
\Phi_\eta(j)Q_N\Phi_\eta(j)^\top
&\longrightarrow
\int_0^\infty e^{As}Q_Ne^{A^\top s}\,ds
=:\Sigma_N. \label{eq:terminal-lyapunov-covariance}
\end{align}
Here the finite-step sum is a Riemann sum, $\eta(T-t_b)\to\infty$ extends its upper limit to infinity, and exponential stability controls the tail.  The matrix $\Sigma_N$ is the unique solution of
\[
A\Sigma_N+\Sigma_NA^\top+Q_N=0.
\]
Uniform integrability of the innovation squares gives the martingale-array Lindeberg condition.  Hence the martingale central limit theorem applied to \eqref{eq:terminal-normal-array} yields
\[
\eta^{-1/2}n_T\ \xrightarrow{d}\ N(0,\Sigma_N).
\]
The limiting covariance $\Sigma_N$ is the stationary covariance of the frozen normal OU process.  If a smooth functional is constant on $\mathcal M$, its derivative annihilates the tangent space, so the delta method retains only this normal fluctuation.  This proves part~(iii). \hfill$\square$

\subsection{Auxiliary Technical Lemmas}\label{app:auxiliary-lemmas}

We first record two Gaussian moment identities used in the multivariate equilibrium proof.
\begin{lemma}\label{lem:sign-proj}
	Let $x \sim N(0,\Sigma)$ with $\Sigma \succ 0$ and $e \in \R^d \setminus \{0\}$.  Define $v=\Sigma^{1/2}e$ and $\hat v = v/\|v\|$.  Then:

	\textup{(i)} $\E\bigl[x\,\sgn(e^\top x)\bigr]
	= \sqrt{\dfrac{2}{\pi}}\;\dfrac{\Sigma e}{\|e\|_{\Sigma}}$.

	\textup{(ii)} $\E\bigl[|e^\top x|\,xx^\top\bigr]
	= \|e\|_\Sigma\,\E[|Z|]\left(\Sigma + \dfrac{\Sigma e\,e^\top\Sigma}{e^\top\Sigma\,e}\right)$, where $\E[|Z|]=\sqrt{2/\pi}$.
\end{lemma}
\textit{Proof of \Cref{lem:sign-proj}.}
Let $\tilde x=\Sigma^{-1/2}x \sim N(0,I_d)$, so $e^\top x = v^\top \tilde x$.  Decompose $\tilde x = (\hat v^\top \tilde x)\hat v + \tilde x_\perp$ with $\hat v^\top \tilde x \sim N(0,1)$ independent of $\tilde x_\perp$.

\textit{Part~(i).}  Since $\E[\tilde x_\perp\,\sgn(\hat v^\top \tilde x)] = 0$ and
$\E[(\hat v^\top \tilde x)\sgn(\hat v^\top \tilde x)] = \E[|Z|] = \sqrt{2/\pi}$,
we obtain $\E[\tilde x\,\sgn(v^\top \tilde x)] = \sqrt{2/\pi}\,\hat v$.
Multiplying by $\Sigma^{1/2}$ and noting $\Sigma^{1/2}\hat v = \Sigma e/\|e\|_\Sigma$ yields~(i).

\textit{Part~(ii).}  We have $\E[|e^\top x|\,xx^\top] = \Sigma^{1/2}\E[|v^\top \tilde x|\,\tilde x\tilde x^\top]\Sigma^{1/2}$ and $|v^\top \tilde x| = \|v\|\,|\hat v^\top \tilde x|$.  Cross terms vanish since $\E[\tilde x_\perp] = 0$, giving
$\E[|v^\top \tilde x|\,\tilde x\tilde x^\top]
= \|v\|\bigl(\E[|Z|^3]\,\hat v\hat v^\top + \E[|Z|]\,(I-\hat v\hat v^\top)\bigr)
= \|v\|\,\E[|Z|](I+\hat v\hat v^\top)$,
where the last step uses $\E[|Z|^3]=2\E[|Z|]$.
Multiplying by $\Sigma^{1/2}$ on both sides and noting $\Sigma^{1/2}\hat v\hat v^\top\Sigma^{1/2} = \Sigma e\,e^\top\Sigma/(e^\top\Sigma e)$ yields~(ii).
\hfill$\square$

\begin{lemma} \label{lem:drift-symmetry}
	Under $H_0$ and Parts~(a)--(c) of \Cref{ass:diffusion}, when $K=2$, let $p_k=\Pr(w=k\mid x,\varepsilon,e_1,e_2)$ denote the selection weight induced by the uniform tie rule, and define $\mu_k(e_1,e_2)=-\E[p_k(\varepsilon+e_k^\top x)x]$.  Then:

	\textup{(a)} {Exchange symmetry:} $\mu_2(e_1,e_2) = \mu_1(e_2,e_1)$ for all $(e_1,e_2)$.

	\textup{(b)} {Negation antisymmetry:} $\mu_1(-e_1,-e_2) = -\mu_1(e_1,e_2)$ for all $(e_1,e_2)$.
\end{lemma}

\textit{Proof of \Cref{lem:drift-symmetry}.}  
The integrand $p_k r_kx$ is integrable.  Indeed, since $0\le p_k\le1$, $r_k=\varepsilon+e_k^\top x$, and $|e_k^\top x|\le\|e_k\|\|x\|$, the Cauchy--Schwarz inequality gives
\[
	\E[\|p_k r_kx\|]
	\le \E[|\varepsilon|\|x\|]+\|e_k\|\E[\|x\|^2]
	\le \sqrt{\E[\varepsilon^2\|x\|^2]}
		+\|e_k\|\sqrt{\E[\|x\|^4]}<\infty,
\]
where the final inequality follows from \Cref{ass:diffusion}(c).  Part~(a) holds by construction: swapping labels $1\leftrightarrow2$ interchanges both the residuals and their selection weights.  For Part~(b), at the negated profile the residuals are $\tilde r_j:=\varepsilon-e_j^\top x$.  By symmetry and independence, $(x,\varepsilon)\overset d=(x,-\varepsilon)$.  Under the substitution $\varepsilon\mapsto-\varepsilon$, we have $\tilde r_j\mapsto-r_j$, so the squared-residual ordering is unchanged.  The corresponding selection weight, including its value $1/2$ on a tie, therefore maps to $p_j$.  Hence
$
	\mu_1(-e_1,-e_2)
	=-\E[p_1(-r_1)x]
	=\E[p_1r_1x]
	=-\mu_1(e_1,e_2).
$
$\hfill\square$

\begin{lemma}\label{lem:equiv}
	Under $H_0$ and Parts~(a)--(c) of \Cref{ass:diffusion} in the scalar case, consider the midpoint--half-gap coordinates $(\bfe,\delta)$ of \Cref{tab:notation}.  The antisymmetric regime $\bfe=0$ is invariant under the ODE~\eqref{eq:ode-drift}.  More precisely, the $(\bfe,\delta)$ system is
	\begin{align}
		\dot{\bfe}
		&= -\tfrac{1}{2}\E[x^2]\,\bfe
		+ \tfrac{1}{2}|\delta|\,\psi(\bfe),
		\label{eq:ebar-dyn} \\
		\dot{\delta}
		&= \tfrac{1}{2}\sgn(\delta)\,
		\E\!\bigl[|\varepsilon+\bfe x|\,|x|\bigr]
		- \tfrac{1}{2}\E[x^2]\,\delta,
		\label{eq:delta-dyn-full}
	\end{align}
	where $\psi(\bfe) := \sgn(\bfe)\,\E[x^2\,\ind\{|\varepsilon|<|\bfe|\,|x|\}]$ satisfies $\psi(0)=0$ (giving invariance), $0\le\psi(\bfe)\le\E[x^2]$ for $\bfe>0$, and $\psi'(0)=2f_\varepsilon(0)\,\E[|x|^3]$.
\end{lemma}

\textit{Proof of \Cref{lem:equiv}.}
Let $p_k=\Pr(w=k\mid x,\varepsilon,e_1,e_2)$ be the tie-aware selection weights.  Then $p_1+p_2=1$, and, with $r_1^2-r_2^2=4\delta x(\varepsilon+\bfe x)$ and $\sgn(0)=0$,
\[
	p_1-p_2=-\sgn\bigl(\delta x(\varepsilon+\bfe x)\bigr).
\]
Recall that $\dot\bfe = \tfrac{1}{2}(\mu_1+\mu_2)$ and $\dot\delta = \tfrac{1}{2}(\mu_1-\mu_2)$. We derive the expressions for $\dot\bfe$ and $\dot\delta$ by substituting $e_1 = \bfe +\delta,e_2=\bfe-\delta$ into $\mu_k = -\E[p_k(\varepsilon+e_k x)x]$:
\begin{align*}
	\mu_1 + \mu_2
	&= -\E\!\left[(p_1+p_2)\varepsilon x\right]
	- \E\!\left[(p_1 e_1 + p_2 e_2)x^2\right]\\
	& = -\E\!\left[(p_1+p_2)\varepsilon x\right]
	- \bfe\,\E\!\left[(p_1+p_2)x^2\right]
	- \delta\,\E\!\left[(p_1-p_2)x^2\right] \\
	&= -\bfe\,\E[x^2]
	- \delta\,\E\!\left[(p_1-p_2)x^2\right],
	\\[6pt]
	\mu_1 - \mu_2
	&= -\E\!\left[(p_1-p_2)\varepsilon x\right]
	- \E\!\left[(p_1 e_1 - p_2 e_2)x^2\right] = -\E\!\left[(p_1-p_2)(\varepsilon+\bfe x)x\right]
	- \delta\,\E[x^2].
\end{align*}
Using $p_1-p_2 = -\sgn(\delta x(\varepsilon+\bfe x))$ and the identity
$\sgn(abc)\,a\,c = |a|\,|c|\,\sgn(b)$, we obtain:
$
-\E\!\left[(p_1-p_2)(\varepsilon+\bfe x)\,x\right]
= \sgn(\delta)\,\E\!\left[|(\varepsilon+\bfe x)\,x|\right]
= \sgn(\delta)\,\E\!\left[|\varepsilon+\bfe x|\,|x|\right].
$
Substituting the same identity into $-\delta\,\E[(p_1-p_2)x^2]$ gives
$
-\delta\E[(p_1-p_2)x^2]
= \delta\E[\sgn(\delta x(\varepsilon+\bfe x))x^2]
= |\delta|\E[\sgn(x(\varepsilon+\bfe x))x^2].
$
Let $F_\varepsilon$ be the noise cdf.  Conditional on $x$, symmetry of $\varepsilon$ gives
\[
	\E_\varepsilon[\sgn\{x(\varepsilon+\bfe x)\}\mid x]
	=\sgn(x)\{2F_\varepsilon(\bfe x)-1\}
	=\sgn(\bfe)\Pr(|\varepsilon|<|\bfe|\,|x|\mid x).
\]
Consequently,
\[
\E\bigl[\sgn(x(\varepsilon+\bfe x))\,x^2\bigr]
= \sgn(\bfe)\,\E\bigl[x^2\,\ind\{|\varepsilon|<|\bfe|\,|x|\}\bigr]
=: \psi(\bfe).
\]
A combination of these equations yields~\eqref{eq:ebar-dyn}--\eqref{eq:delta-dyn-full}.
The properties of $\psi$ follow from its definition: $\psi(0)=0$ since $\ind\{|\varepsilon|<0\}=0$;
the bound $0\le\psi(\bfe)\le\E[x^2]$ for $\bfe>0$ is immediate; and, as $\bfe\to0$,
\[
	\frac{\psi(\bfe)}{\bfe}
	=\E\!\left[x^2\frac{\Pr(|\varepsilon|<|\bfe|\,|x|\mid x)}{|\bfe|}\right]
	\longrightarrow 2f_\varepsilon(0)\E[|x|^3].
\]
Indeed, the integrand is bounded by $2\|f_\varepsilon\|_\infty|x|^3$, which is integrable under \Cref{ass:diffusion}(c).  This proves the two-sided derivative at zero.
\hfill$\square$

\begin{lemma}\label{lem:ebar-multi}
	Under $H_0$ and Parts~(a)--(c) of \Cref{ass:diffusion} with $x\in\R^d$, consider the midpoint--half-gap coordinates $(\bfe,\delta)$ of \Cref{tab:notation}.  Then:
	
	(i) The symmetric ODE component $\dot{\bfe}=\frac{1}{2}(\mu_1+\mu_2)$ satisfies
	\begin{equation}\label{eq:ebar-multi}
		\dot{\bfe} = -\tfrac{1}{2}\Sigma\bfe + \tfrac{1}{2}\Psi(\bfe,\delta),
	\end{equation}
	where $\Psi(\bfe,\delta) := \E\bigl[|\delta^\top x|\,\sgn(\varepsilon+\bfe^\top x)\,x\bigr]$,
	and $\Psi(0,\delta)=0$ for every $\delta\in\R^d$.
	
	(ii) The Jacobian of $\dot{\bfe}$ with respect to $\bfe$, evaluated at $(\bfe,\delta)=(0,e^*)$, is
	\begin{equation}\label{eq:Jsym-app}
		J^{\mathrm{sym}} = -\tfrac{1}{2}\Sigma + f_\varepsilon(0)\,\E\bigl[|e^{*\top}x|\,xx^\top\bigr].
	\end{equation}
\end{lemma}

\textit{Proof of \Cref{lem:ebar-multi}.}
 {Part~(i).}
Let $p_k=\Pr(w=k\mid x,\varepsilon,e_1,e_2)$ be the tie-aware selection weights.  Since $r_1^2-r_2^2=4(\delta^\top x)(\varepsilon+\bfe^\top x)$, we have $p_1+p_2=1$ and, with $\sgn(0)=0$, $p_1-p_2=-\sgn((\varepsilon+\bfe^\top x)(\delta^\top x))$.
From $\mu_k=-\E[p_k(\varepsilon+e_k^\top x)x]$ and $\E[\varepsilon x]=0$:
$\mu_1+\mu_2 = -\Sigma\bfe - \E[(p_1-p_2)(\delta^\top x)x]$.
Using $\sgn(ab)\cdot b = |b|\sgn(a)$, the selection-weight difference gives $(\delta^\top x)\sgn((\varepsilon+\bfe^\top x)(\delta^\top x)) = |\delta^\top x|\sgn(\varepsilon+\bfe^\top x)$, so $\mu_1+\mu_2 = -\Sigma\bfe + \E[|\delta^\top x|\sgn(\varepsilon+\bfe^\top x)x]$.
Since $\sgn(\varepsilon)$ is independent of $x$ with $\E[\sgn(\varepsilon)]=0$, we have $\Psi(0,\delta)=0$.

 {Part~(ii).}
Condition on $x$: $\Psi_i(\bfe,\delta) = \E_x[|\delta^\top x|\,x_i\,h(\bfe^\top x)]$
where $h(t):=\E[\sgn(\varepsilon+t)]=2F_\varepsilon(t)-1$ satisfies $h(0)=0$, $h(t)/t\to2f_\varepsilon(0)$, and $|h(t)|\le2\|f_\varepsilon\|_\infty|t|$.  The difference quotient in direction $a\in\R^d$ is therefore dominated by
$2\|f_\varepsilon\|_\infty|\delta^\top x|\,|x_i|\,|a^\top x|$,
which is integrable by \Cref{ass:diffusion}(c).  Dominated convergence gives
$D_{\bfe}\Psi_i(0,\delta)[a]=2f_\varepsilon(0)\E[|\delta^\top x|x_i(a^\top x)]$.
Setting $\delta=e^*$ yields
$\partial\Psi_i/\partial\bfe_j|_{(0,e^*)}=2f_\varepsilon(0)\E[|e^{*\top}x|x_ix_j]$.
\hfill$\square$

\begin{lemma} \label{lem:jacobian-block}
	Under $H_0$ and Parts~(a)--(c) of \Cref{ass:diffusion}, at any separated antisymmetric equilibrium $(e_1,e_2)=(e^*,-e^*)$ with $e^*\ne0$ at which the drift is differentiable, the $2d\times 2d$ Jacobian of the ODE system~\eqref{eq:ode-drift} in the $(\bfe,\delta)$ coordinates is block-diagonal:
	\begin{equation}\label{eq:block-diag-main}
		J_{(\bfe,\delta)} = \begin{pmatrix} J^{\mathrm{sym}} & 0 \\ 0 & J^{\mathrm{anti}} \end{pmatrix},
	\end{equation}
	where $J^{\mathrm{sym}},J^{\mathrm{anti}}\in\R^{d\times d}$ are the Jacobians of the $\bfe$- and $\delta$-dynamics at $(\bfe,\delta)=(0,e^*)$, respectively.
\end{lemma}

\textit{Proof of \Cref{lem:jacobian-block}.}
Let $\nabla_1\mu_k(a,b)$ and $\nabla_2\mu_k(a,b)$ denote the $d\times d$ Jacobian matrices of $\mu_k$ with respect to its first and second arguments, respectively.  Define
$  A := \nabla_1\mu_1(e^*,-e^*)$ and $ B := \nabla_2\mu_1(e^*,-e^*)$. The full $2d\times 2d$ Jacobian of $(\dot e_1,\dot e_2) = (\mu_1,\mu_2)$ at $(e_1,e_2)=(e^*,-e^*)$ is
\[
J_{(e_1,e_2)} = \begin{pmatrix} \nabla_1\mu_1 & \nabla_2\mu_1 \\ \nabla_1\mu_2 & \nabla_2\mu_2 \end{pmatrix}\bigg|_{(e^*,-e^*)}.
\]
The first row is $(A,\;B)$ by definition.  We now show the second row is $(B,\;A)$.

Using $\mu_2(e_1,e_2) = \mu_1(e_2,e_1)$ (i.e.,   exchange symmetry (a)), we have:
\[
\nabla_1\mu_2(e_1,e_2) = \nabla_2\mu_1(e_2,e_1), \quad \nabla_2\mu_2(e_1,e_2) = \nabla_1\mu_1(e_2,e_1).
\]
At $(e_1,e_2)=(e^*,-e^*)$, we further obtain:
\begin{equation}\label{eq:exchange-deriv}
	\nabla_1\mu_2(e^*,-e^*) = \nabla_2\mu_1(-e^*,e^*), \qquad
	\nabla_2\mu_2(e^*,-e^*) = \nabla_1\mu_1(-e^*,e^*).
\end{equation}
These express the second-row blocks in terms of $\mu_1$-derivatives evaluated at $(-e^*,e^*)$. 

Using $\mu_1(-e_1,-e_2) = -\mu_1(e_1,e_2)$ (i.e., negation antisymmetry (b)), we have:
\[
\frac{\partial}{\partial e_1}\mu_1(-e_1,-e_2) = -\nabla_1\mu_1(-e_1,-e_2)= -\nabla_1\mu_1(e_1,e_2),
\]
which implies
$\nabla_1\mu_1(-e_1,-e_2) = \nabla_1\mu_1(e_1,e_2)$
for all $(e_1,e_2)$.  Similarly, $\nabla_2\mu_1(-e_1,-e_2) = \nabla_2\mu_1(e_1,e_2)$.
At $(e_1,e_2) = (e^*,-e^*)$, we then have:
\begin{equation}\label{eq:negation-deriv}
	\nabla_1\mu_1(-e^*,e^*) = \nabla_1\mu_1(e^*,-e^*) = A, \qquad
	\nabla_2\mu_1(-e^*,e^*) = \nabla_2\mu_1(e^*,-e^*) = B.
\end{equation}

By combining~\eqref{eq:exchange-deriv} and~\eqref{eq:negation-deriv}, we also have:
$\nabla_1\mu_2(e^*,-e^*) = B$ and $\nabla_2\mu_2(e^*,-e^*) = A$.
Therefore the full $2d\times 2d$ Jacobian in $(e_1,e_2)$ coordinates is
\begin{equation}\label{eq:J-e1e2}
	J_{(e_1,e_2)} = \begin{pmatrix} A & B \\ B & A \end{pmatrix}.
\end{equation}

For analytical convenience, we show an equivalent transformation of the Jacobian.  In the midpoint--half-gap coordinates, $e_1 = \bfe+\delta$ and $e_2 = \bfe-\delta$.
Applying the chain rule, we transform the Jacobian as $J_{(\bfe,\delta)} = P\,J_{(e_1,e_2)}\,P^{-1}$, where $P = \frac{\partial(\bfe,\delta)}{\partial(e_1,e_2)} = \frac{1}{2}\bigl(\begin{smallmatrix} I & I \\ I & -I \end{smallmatrix}\bigr)$.  Carrying out the multiplication:
\begin{align*}
	P\begin{pmatrix} A & B \\ B & A \end{pmatrix}P^{-1}
&	= \frac{1}{2}\begin{pmatrix} I & I \\ I & -I \end{pmatrix}
	\begin{pmatrix} A & B \\ B & A \end{pmatrix}
	\begin{pmatrix} I & I \\ I & -I \end{pmatrix} \\
&	= \frac{1}{2}\begin{pmatrix} A+B & A+B \\ A-B & -(A-B) \end{pmatrix}
	\begin{pmatrix} I & I \\ I & -I \end{pmatrix}  \\
&	= \begin{pmatrix} A+B & 0 \\ 0 & A-B \end{pmatrix}.
\end{align*}
Therefore
\begin{equation*} 
	J_{(\bfe,\delta)} = \begin{pmatrix} J^{\mathrm{sym}} & 0 \\ 0 & J^{\mathrm{anti}} \end{pmatrix},
	\quad\text{where}\quad
	J^{\mathrm{sym}} = A+B,\quad J^{\mathrm{anti}} = A-B.
\end{equation*}
\hfill$\square$

\section{Detection with Higher-Order Competitive Loss}\label{sec:K-general}

The formal detection results in the main text use the $2$-competitive loss ($K=2$).  A natural question, raised at the end of \Cref{sec:detection} of the main text, is what happens under a higher-order competitive loss with more than two pieces; in particular, whether adding a third competing piece ($K=3$) can improve power.  The following result characterizes the performance of the $3$-competitive loss in the scalar case.  

\begin{proposition}\label{prop:K3-scalar}
	Consider the scalar $3$-competitive loss with
	$x\sim N(0,\sigma_x^2)$ and
	$\varepsilon\sim N(0,\sigma_\varepsilon^2)$.
	
	(i) Under $H_0$, every nonzero symmetric ODE path converges to
	$\theta^{\mathrm{true}}+(\sigma_\varepsilon/\sigma_x)(-c_3,0,c_3)$; its normalized outer gap is $2c_3$, where $c_3\approx 1.320$ is a constant.
	
	(ii) Under $H_1$, suppose the two true slopes are
	$\theta_0-a$ and $\theta_0+a$ with equal probabilities, independently of $x$, and let
	$\delta_*=2a\sigma_x/\sigma_\varepsilon$.
	For all sufficiently small $\delta_*>0$, the ODE has a locally attracting symmetric equilibrium
	$\theta_0+(\sigma_\varepsilon/\sigma_x)(-A(\delta_*),0,A(\delta_*))$, where
	$
	2A(\delta_*)
	=2c_3-0.2502\,\delta_*^2+O(\delta_*^4).
	$
	Suppose, as $\eta\to0^+$ with $T\asymp\eta^{-p}$ for some fixed $p\in(1,2)$, that the terminal iterates satisfy
	$
	\bigl\|(\theta_{1,T},\theta_{2,T},\theta_{3,T})
	-\{\theta_0\mathbf 1_3+(\sigma_\varepsilon/\sigma_x)(-A(\delta_*),0,A(\delta_*))\}\bigr\|=o_p(1),
	$
	$\hat\sigma_\varepsilon\xrightarrow{p}\sigma_\varepsilon$, and the fitted null critical value converges to $2c_3$.  Then
	$\Pr(\textup{reject}\mid H_1)\to0$.
\end{proposition}
\textit{Proof of \Cref{prop:K3-scalar}.} 
To prove the result, it suffices to reduce the three-candidate dynamics to a population three-means problem for a tilted ratio law, establish the null equilibrium and its stability, and then expand the outer equilibrium under the balanced alternative.  For notational convenience, set $s_0:=\sigma_\varepsilon/\sigma_x$, write the three ordered candidate slopes as $q_1<q_2<q_3$, and let $W=y/x$.

We first derive the reduced ODE.  Except on the null event $x=0$, the winning candidate is the one closest to $W$, so its population ODE can be written as
\[
\dot q_k
=\E\!\left[x^2(W-q_k)\ind\{W\in C_k(q)\}\right],
\]
where $C_k(q)$ is the Voronoi cell of $q_k$.  Under the probability measure obtained by tilting the data law by $x^2/\E[x^2]$, the normalized ratio
$R=\varepsilon/(s_0x)$ has density
$g(r)=2/\{\pi(1+r^2)^2\}$.  For notational convenience, let $G(u)=\int_{-\infty}^u g(v)\,dv$ and $M(u)=\int_{-\infty}^u v g(v)\,dv$.  Hence, after centering and dividing slopes by $s_0$, the ODE equilibria are exactly the population three-means centroids of a shifted copy, or a mixture of shifted copies, of $g$.

We next prove part~(i).  For a symmetric candidate profile $(-A,0,A)$, the right-cell boundary is $A/2$, and a component shifted by $s$ contributes
\begin{equation}\label{eq:K3-right-drift}
	\mathcal R(A,s)
	=(s-A)\{1-G(A/2-s)\}-M(A/2-s)
\end{equation}
to the right-candidate drift.  Under $H_0$, the symmetric equilibrium equation is
$\mathcal R(A,0)=0$, or
$A\{1-G(A/2)\}+M(A/2)=0$, whose unique positive solution is
$A=c_3\approx1.32014$.  Note that the right drift points outward below this root and inward above it, proving convergence of every nondegenerate symmetric trajectory.  By directly differentiating the three centroid drifts at $(-c_3,0,c_3)$, we obtain
\[
J_0=
\begin{pmatrix}
	\tfrac{c_3g(c_3/2)}{4}-\{1-G(c_3/2)\} & \tfrac{c_3g(c_3/2)}{4} & 0\\[2pt]
	\tfrac{c_3g(c_3/2)}{4} & \tfrac{c_3g(c_3/2)}{2}-\{2G(c_3/2)-1\} & \tfrac{c_3g(c_3/2)}{4}\\[2pt]
	0 & \tfrac{c_3g(c_3/2)}{4} & \tfrac{c_3g(c_3/2)}{4}-\{1-G(c_3/2)\}
\end{pmatrix}.
\]
Its eigenvalues are approximately
$-0.50733$, $-0.06601$, and $-0.01892$, so the equilibrium is locally attracting in the full ordered three-candidate ODE.  Rescaling to the original slopes gives part~(i).

We then expand the outer equilibrium under the balanced alternative.  The normalized component shifts are
$s=\pm\lambda$, where $\lambda=a/s_0=\delta_*/2$.  Symmetry keeps the middle candidate at zero, and the outer equilibrium solves
\begin{equation}\label{eq:K3-H1-equilibrium}
	Q(A,\lambda)
	:=\frac{\mathcal R(A,\lambda)+\mathcal R(A,-\lambda)}2=0.
\end{equation}
For notational convenience, let $b=c_3/2$.  At $(A,\lambda)=(c_3,0)$,
\[
Q_A=-\{1-G(b)\}+\frac{c_3g(b)}4
\approx-0.066013,
\qquad
Q_{\lambda\lambda}=g(b)+b g'(b)
\approx-0.066060.
\]
By applying the implicit-function theorem together with the symmetry in $\lambda$, we obtain
\[
A(\lambda)
=c_3-\frac{Q_{\lambda\lambda}}{2Q_A}\lambda^2+O(\lambda^4)
=c_3-0.50036\,\lambda^2+O(\lambda^4).
\]
Moreover, because the three eigenvalues of $J_0$ are strictly negative and the drift Jacobian varies continuously with $\lambda$, this equilibrium remains locally attracting for all sufficiently small $\lambda$.

We finally prove the blindness conclusion.  Substituting $\lambda=\delta_*/2$ yields
\[
2A(\delta_*)
=2c_3-0.25018\,\delta_*^2+O(\delta_*^4).
\]
For each sufficiently small fixed $\delta_*>0$, this limiting outer gap lies a fixed distance below the null center $2c_3$.  Under the concentration and vanishing calibration-error conditions in the proposition, the statistic converges to $2A(\delta_*)$ while the upper-tail critical value converges to $2c_3$.  Therefore, $\Pr(\textup{reject}\mid H_1)\to0$.  Note that numerically solving the exact equation~\eqref{eq:K3-H1-equilibrium} subject to $A=c_3$ gives the nonzero crossing
$\lambda\approx0.67069$, or $\delta_*\approx1.34138$; the full drift Jacobian there has eigenvalues approximately
$-0.35142$, $-0.16080$, and $-0.01294$, confirming that this blind-point equilibrium remains locally attracting, which completes the proof.
\hfill$\square$

\begin{figure}[h!]
	\centering
	\begin{subfigure}{0.35\textwidth}
		\centering
		\includegraphics[width=\linewidth]{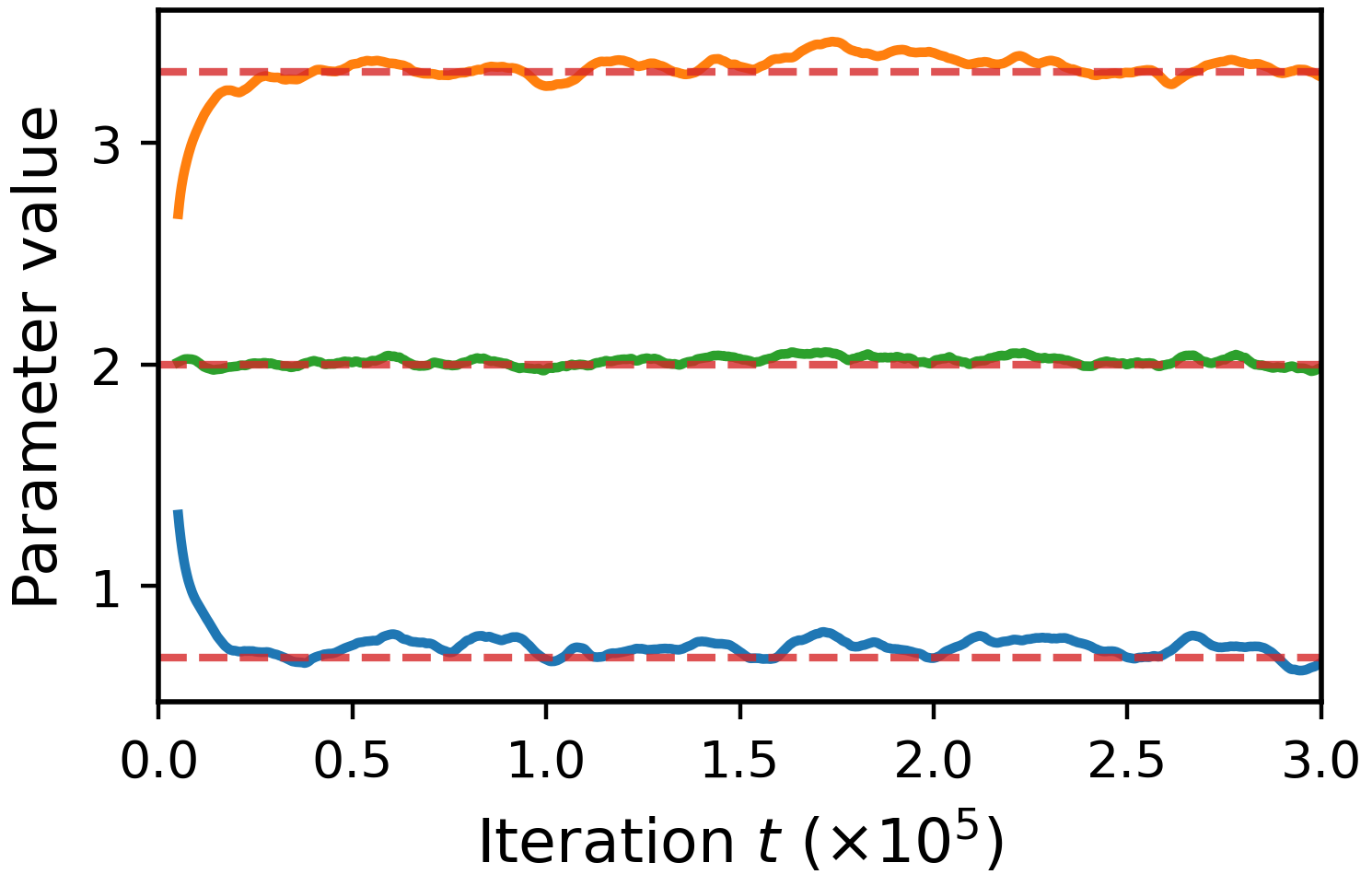}
		\caption{$H_0$: trajectories}
		\label{fig:K3-H0-traj}
	\end{subfigure}
	\qquad
	\begin{subfigure}{0.35\textwidth}
		\centering
		\includegraphics[width=\linewidth]{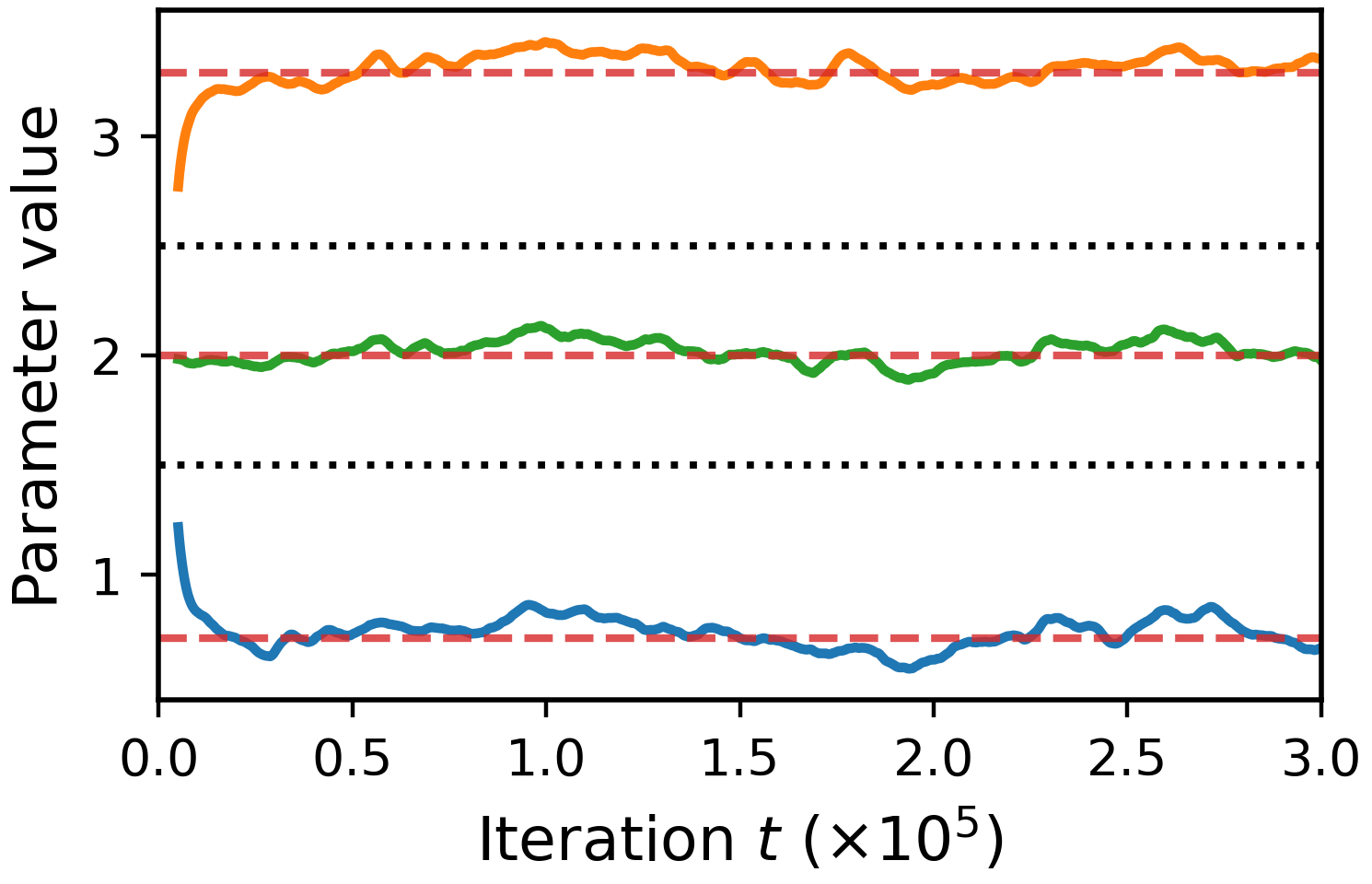}
		\caption{$H_1$: trajectories}
		\label{fig:K3-H1-traj}
	\end{subfigure}
	\par\medskip
	\begin{subfigure}{0.35\textwidth}
		\centering
		\includegraphics[width=\linewidth]{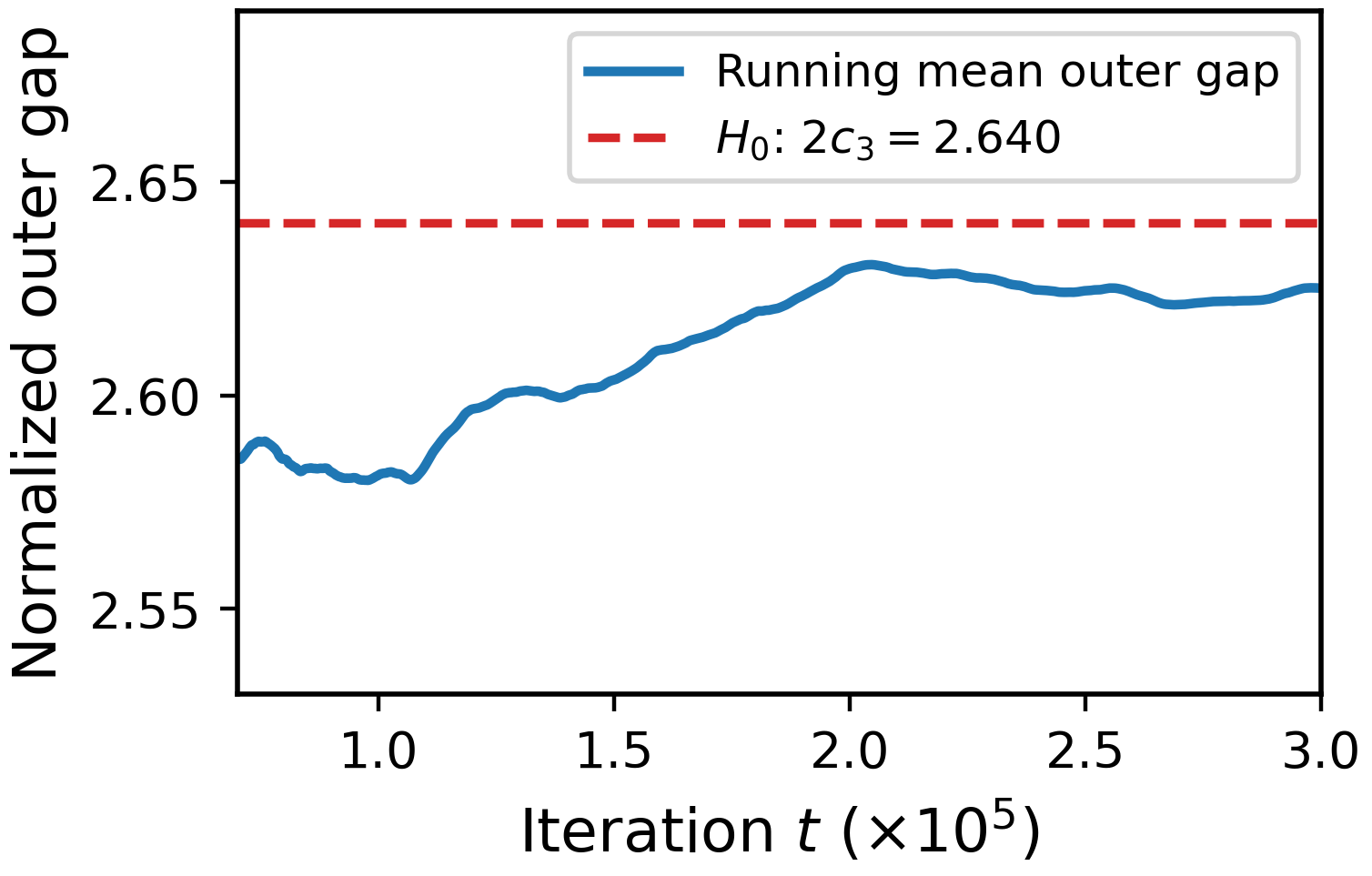}
		\caption{$H_0$: outer gap}
		\label{fig:K3-H0-gap}
	\end{subfigure}
	\qquad
	\begin{subfigure}{0.35\textwidth}
		\centering
		\includegraphics[width=\linewidth]{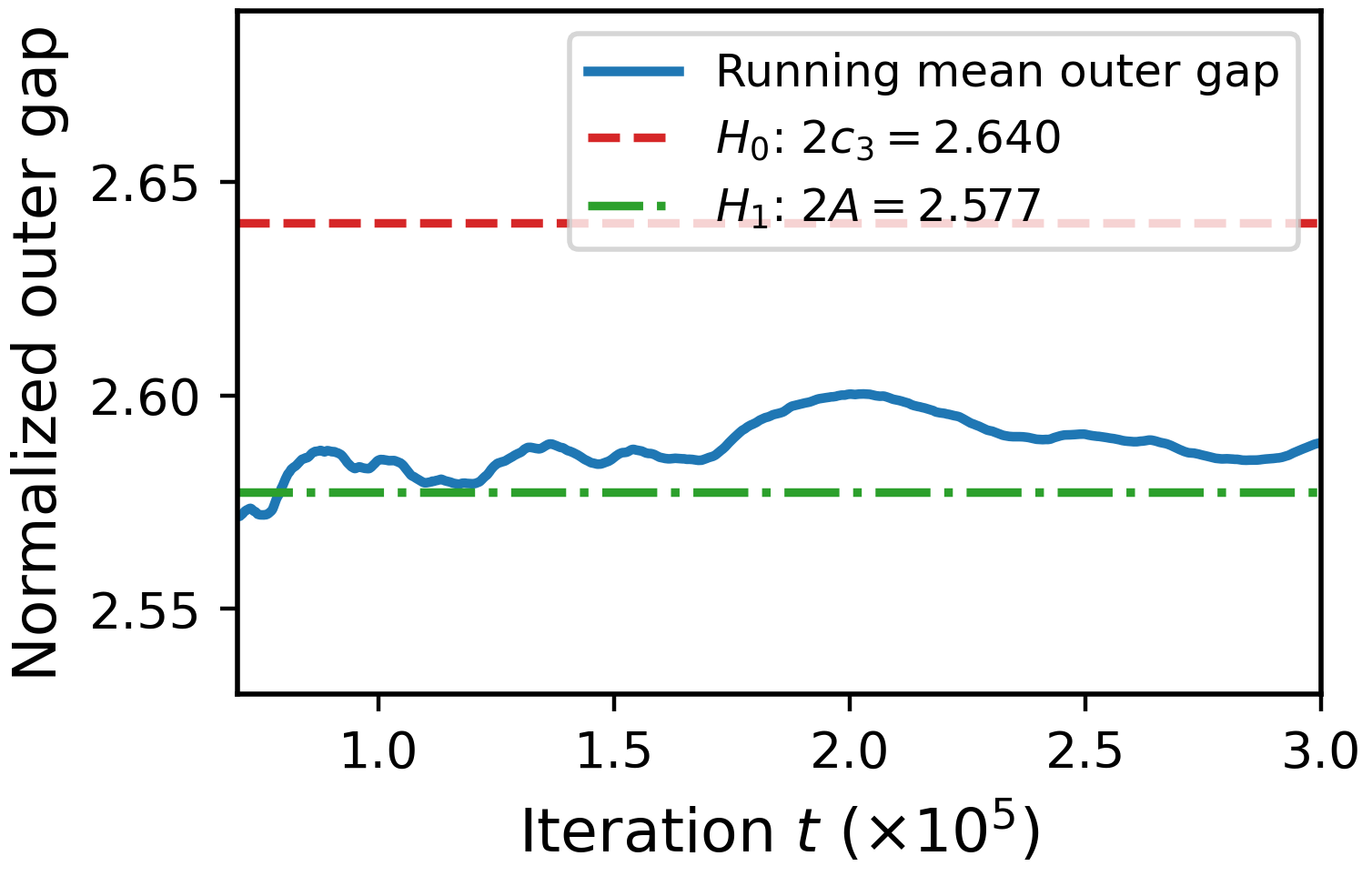}
		\caption{$H_1$: outer gap}
		\label{fig:K3-H1-gap}
	\end{subfigure}
	\caption{$K=3$ SGD under $H_0$ and $H_1$ ($\theta_0=2$, $\sigma_\varepsilon=\sigma_x=1$, $\eta=0.002$, $T=300{,}000$).  Under $H_0$, panels~(a) and~(c) show convergence toward $\theta_0+(-c_3,0,c_3)$ and outer gap $2c_3\approx2.640$.  Under the balanced alternative with true slopes $1.5$ and $2.5$ ($\delta_*=1$), panels~(b) and~(d) show convergence toward $\theta_0+(-A,0,A)$ and outer gap $2A\approx2.577<2c_3$.  Red dashed lines mark the ODE equilibria in panels~(a)--(b) and the null gap in panels~(c)--(d); black dotted lines in panel~(b) mark the true slopes, and the green dash-dotted line in panel~(d) marks the alternative gap.  Gap curves are post-burn-in running means.}
	\label{fig:K3-H01}
\end{figure}

\Cref{fig:K3-H01} illustrates the contrast in \Cref{prop:K3-scalar}: the $K=3$ outer gap is smaller under the displayed alternative than under the null.  This is a structural failure rather than merely a loss of finite-sample precision.  Under the balanced alternative, both adjacent gaps equal $A(\delta_*)$ and the outer gap equals $2A(\delta_*)$, so every natural one-sided summary of the pairwise gaps initially moves in the wrong direction.  The exact equilibrium equation also has a nonzero numerical crossing at $\delta_*\approx1.341$, where the locally attracting alternative equilibrium has the same leading gap pattern $\{c_3,c_3,2c_3\}$ as $H_0$.  Thus neither reversing the tail nor using another pairwise-gap summary yields an omnibus $K=3$ test.  This failure holds with known $\sigma_\varepsilon$ and therefore cannot be repaired by changing the scale estimator.

The third competitor truncates the winner regions of the outer pieces, so the $K=3$ alternative outer gap need not exceed its null value.  Consequently, a one-sided test based on the $K=3$ outer gap is not an omnibus detector.

\Cref{fig:K4-H01} provides an exploratory numerical comparison for $K=4$.  Its four pieces generate several pairwise-gap levels that move in different directions under $H_1$, illustrating that the geometry becomes still richer as $K$ increases; formal analysis of the higher-order loss is limited to $K=3$.

\paragraph*{Numerical $K=4$ Comparison.}

\begin{figure}[H]
	\centering
	\begin{subfigure}{0.35\textwidth}
		\centering
		\includegraphics[width=\linewidth]{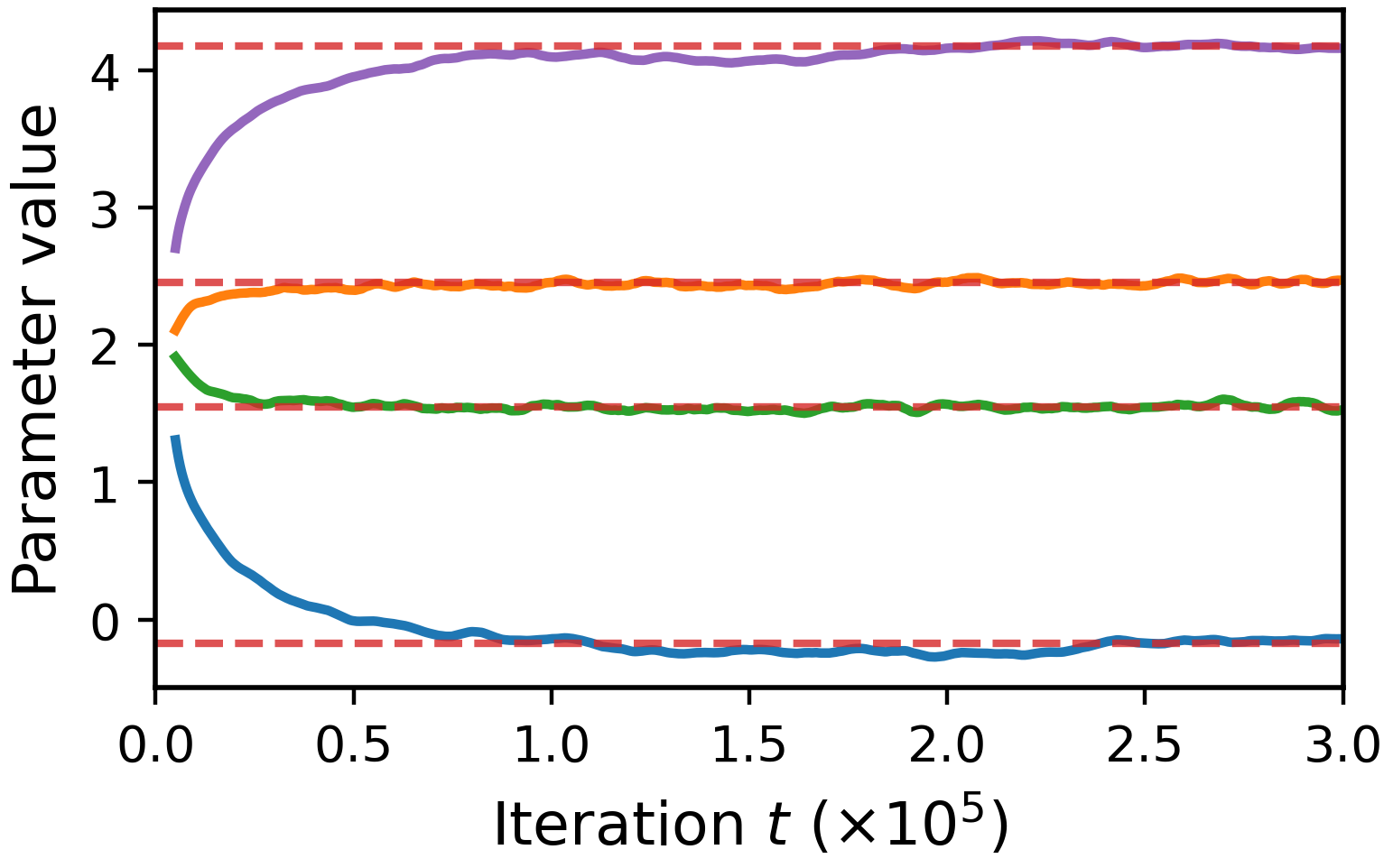}
		\caption{$H_0$: trajectories}
		\label{fig:K4-H0-traj}
	\end{subfigure}
	\qquad
	\begin{subfigure}{0.35\textwidth}
		\centering
		\includegraphics[width=\linewidth]{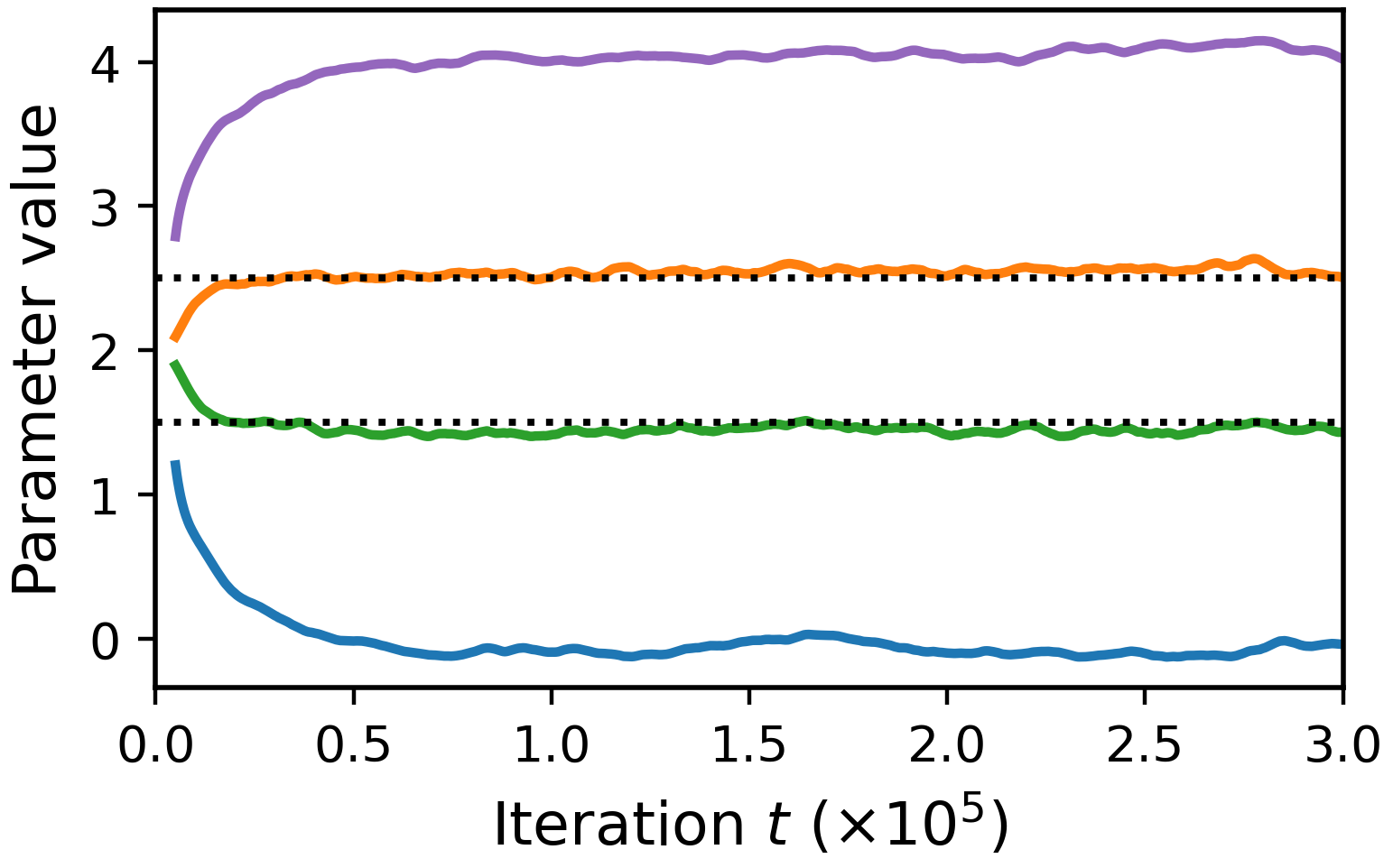}
		\caption{$H_1$: trajectories}
		\label{fig:K4-H1-traj}
	\end{subfigure}
	\par\medskip
	\begin{subfigure}{0.35\textwidth}
		\centering
		\includegraphics[width=\linewidth]{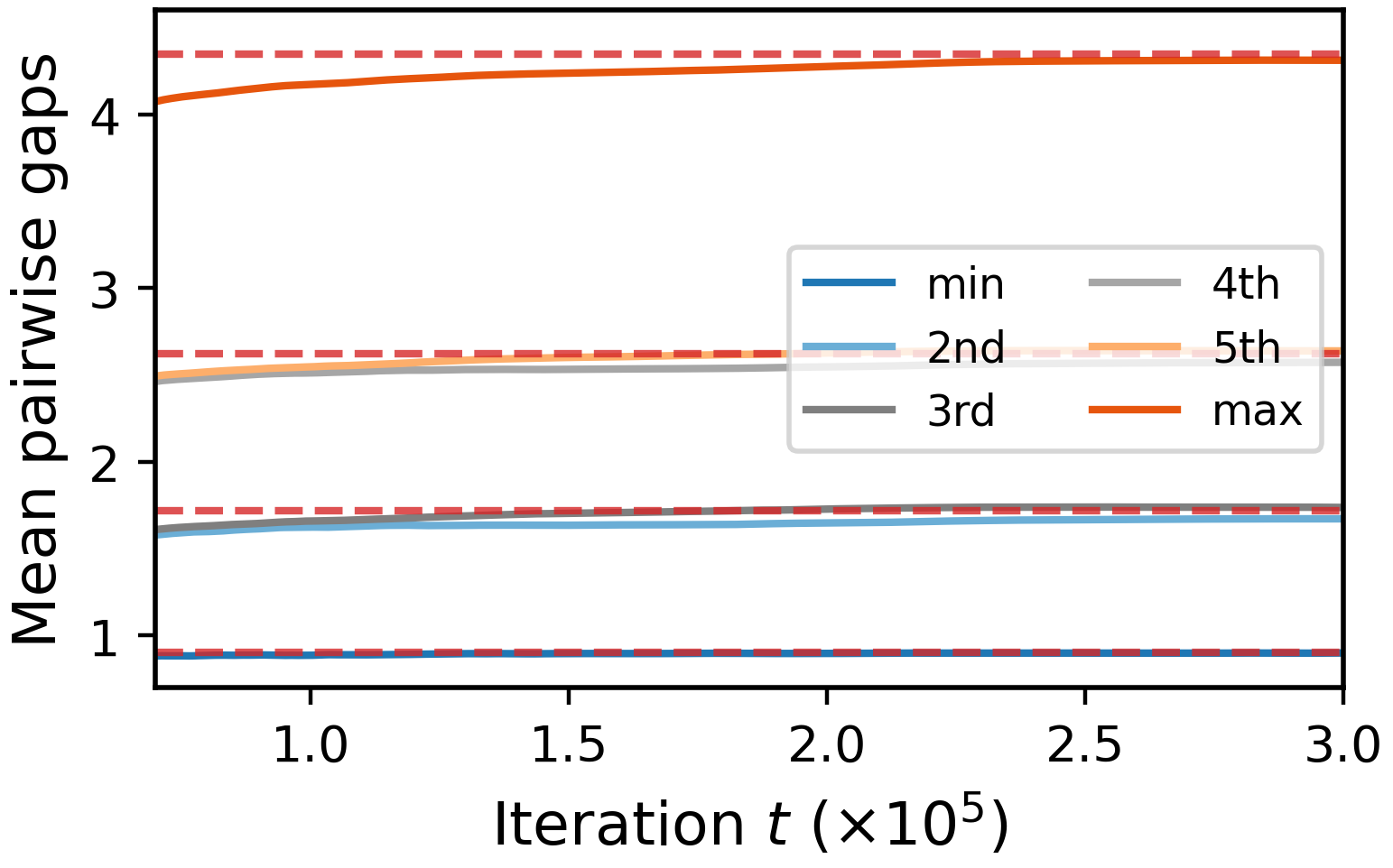}
		\caption{$H_0$: pairwise gaps}
		\label{fig:K4-H0-gaps}
	\end{subfigure}
	\qquad
	\begin{subfigure}{0.35\textwidth}
		\centering
		\includegraphics[width=\linewidth]{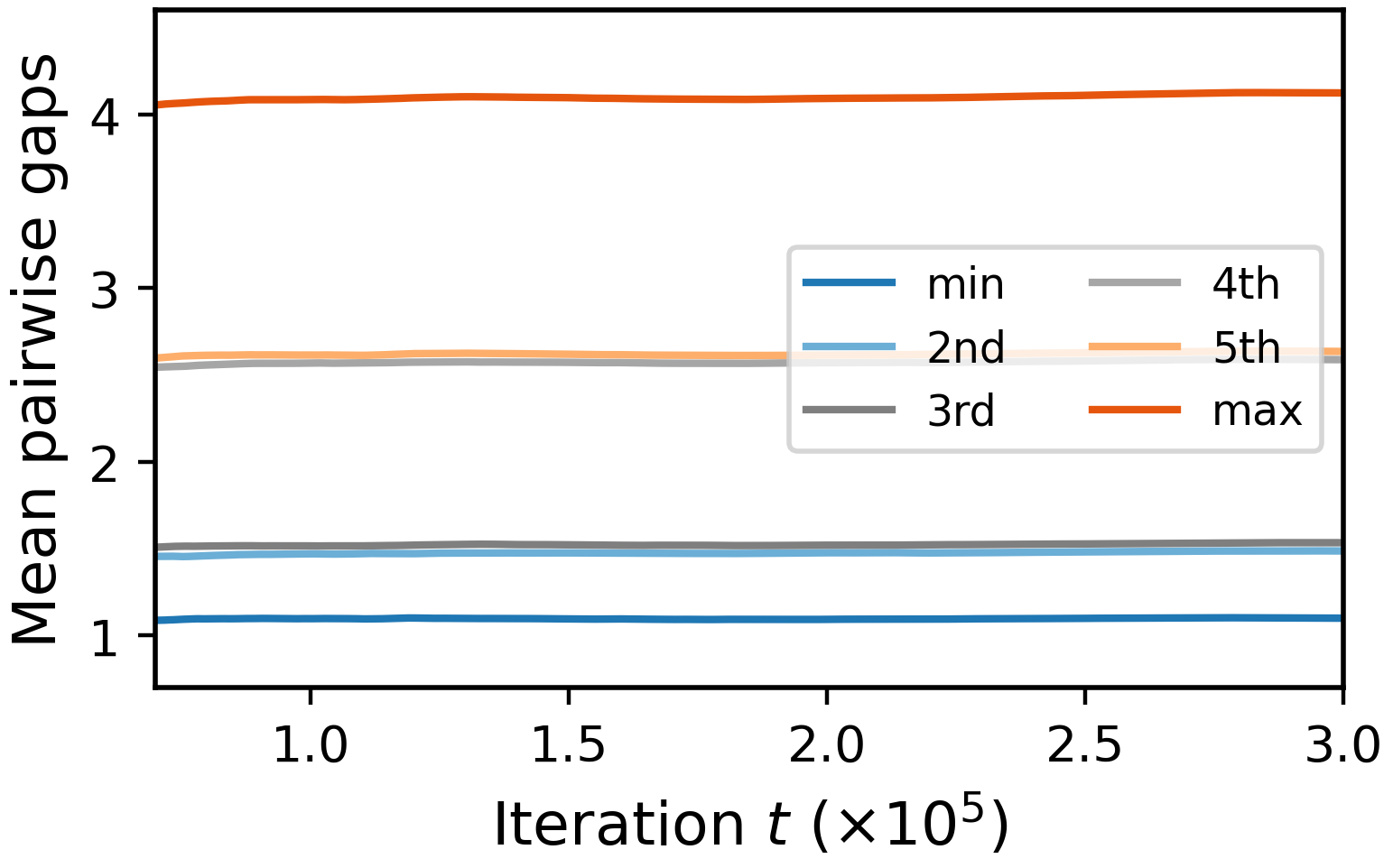}
		\caption{$H_1$: pairwise gaps}
		\label{fig:K4-H1-gaps}
	\end{subfigure}
	\caption{Exploratory scalar $K=4$ SGD under $H_0$ and balanced $H_1$, using the settings of \Cref{fig:K3-H01}.  Red dashed lines in panels~(a) and~(c) mark the numerical $H_0$ equilibrium levels; black dotted lines in panel~(b) mark the two true slopes under $H_1$.  Panels~(c)--(d) show post-burn-in running means of the six ordered pairwise gaps.  Some gap ranks increase under $H_1$ while others decrease, so no single directional comparison captures the change.}
	\label{fig:K4-H01}
\end{figure}

\Cref{fig:K4-H01} uses the same scalar Gaussian null and balanced alternative as \Cref{fig:K3-H01}, but fits four competing pieces.  Under $H_0$, the four estimates settle near two radii and produce four pairwise-gap levels.  Under $H_1$, the smallest gaps increase while the largest gap decreases, so the gap ranks do not share a common direction of departure from their null values.

\section{Additional Synthetic Robustness Results}\label{app:robustness-experiments}

\subsection{Finite-Sample Audit of the Repeated-Data Protocol}

The asymptotic size result in \Cref{cor:CI} uses an independent preliminary sample of size $n_0\asymp T$ and a fresh detection stream.  The real-data implementation instead has only $n$ distinct observations and recycles them for 1{,}000 passes.  Increasing the number of updates in this way does not increase the effective sample size of the preliminary regression or scale estimator.  To measure the resulting finite-sample discrepancy directly, we generated 300 independent datasets under the scalar Gaussian null at four sample sizes spanning the smaller real applications.  For each dataset we applied the literal real-data settings: $\eta=0.0005$, 1{,}000 passes, two-fold cross-fitted component-scale estimation, the Gaussian analytic scalar critical value, and averaging over 10 SGD runs.  We also evaluated the matched pooled-residual, distribution-adaptive calibration under a Taylor-style binary design with $\Pr(X=1)=0.484$.

\begin{table}[ht]
\centering
\caption{Finite-sample null rejection rates under the repeated-data real-application protocol.}
\label{tab:finite-n-reuse}
\footnotesize
\setlength{\tabcolsep}{4pt}
\begin{tabular}{@{}l r l r r r@{}}
\toprule
Design & $n$ & Calibration & Rejections & Rate & MC s.e. \\
\midrule
Gaussian scalar & 150   & component/Gaussian analytic & 176/300 & 0.587 & 0.028 \\
Gaussian scalar & 333   & component/Gaussian analytic & 166/300 & 0.553 & 0.029 \\
Gaussian scalar & 777   & component/Gaussian analytic & 138/300 & 0.460 & 0.029 \\
Gaussian scalar & 1{,}840 & component/Gaussian analytic &  91/300 & 0.303 & 0.027 \\
Binary, $\Pr(X=1)=0.484$ & 777 & pooled/distribution-adaptive & 8/300 & 0.027 & 0.009 \\
\bottomrule
\end{tabular}
\end{table}

The nominal level is 0.05.  The continuous-scalar analytic rule therefore over-rejects severely at all four audited sample sizes, while the binary matched rule is conservative in the audited design.  These are protocol-level null experiments rather than dataset-specific nulls, so they do not by themselves supply corrected $p$-values for \Cref{tab:confounder-all}.  They do show that the fresh-stream experiment cannot validate inference after aggressive reuse of a finite dataset.  Until a full-pipeline finite-sample calibration (for example, a design-matched parametric bootstrap that repeats nuisance estimation, data reuse, averaging, and statistic selection) is completed for each application, the affected Stage~1 analytic $p$-values should be treated as provisional.

We first assess rejection power using the one-sided rule rather than the null center alone.  For $\Delta=0.5$, each of 13 dimension--endogeneity settings uses 500 independent $H_0$ runs to estimate its upper critical value and 100 independent $H_1$ runs to estimate power, with the component-scale estimator recomputed in every run.

The complementary signal-strength experiment fixes $K=K^*=2$, uses balanced slopes $\pm(\Delta,0,\ldots,0)$, and varies $d\in\{1,2,5,10\}$.  \Cref{fig:signal-strength} shows that the component-normalized mean gap remains above $4/\pi$ and nearly invariant to $d$: it ranges from $1.378$ to $1.405$ for $\Delta=0.3$, from $1.572$ to $1.597$ for $\Delta=0.5$, and from $2.273$ to $2.285$ for $\Delta=1$.

\begin{figure}[ht]
	\centering
	\includegraphics[width=0.62\textwidth]{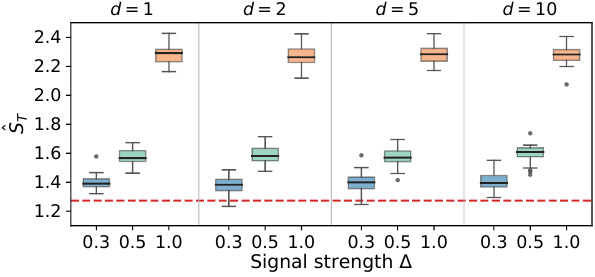}
	\caption{Alternative separation as signal strength varies.  Boxplots show seed-level distributions of $\hat S_T$ for $K=K^*=2$, and the dashed line marks the Gaussian null center $4/\pi\approx1.273$.}
	\label{fig:signal-strength}
\end{figure}

To assess robustness to endogeneity (\Cref{prop:endogeneity}), we draw $(x_t,\varepsilon_t)$ jointly Gaussian with $\mathrm{Cov}(x_t,\varepsilon_t)=\rho\,\sigma_x\sigma_\varepsilon$ and vary~$\rho$, using the default synthetic setup in the main text.  \Cref{fig:univ-endo} shows the resulting seed-level distributions after projected-noise calibration.

\begin{figure}[ht]
	\centering
	\includegraphics[width=0.58\textwidth]{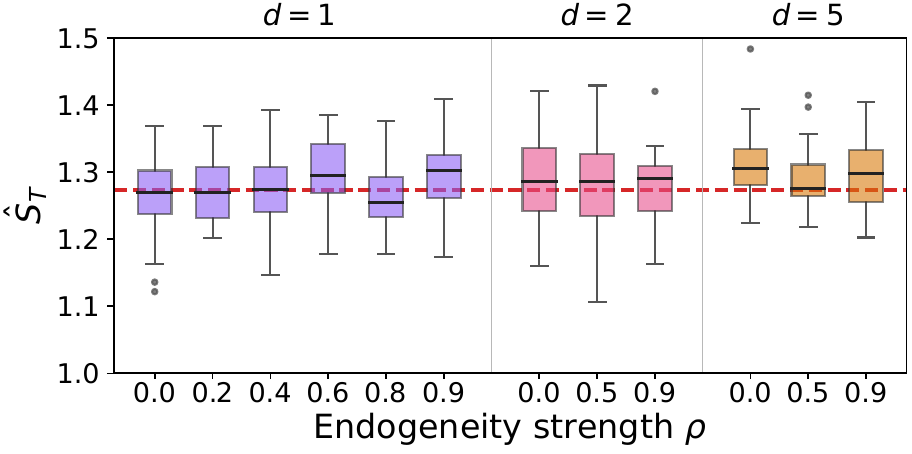}
	\caption{Null calibration under joint Gaussian endogeneity.  Boxplots show seed-level distributions of $\hat S_T$, and the dashed line marks the Gaussian null center $4/\pi\approx1.273$.}
	\label{fig:univ-endo}
\end{figure}

We next explore the impact of heteroskedasticity and quadratic mean misspecification using a scalar setup with $x\sim N(0,1)$, $\sigma_\varepsilon=1$, step size $\eta=0.005$, and $T=200{,}000$.  Each setting uses 200 independent trials.  Under $H_0$, the heteroskedastic designs set $\Var(\varepsilon\mid x)$ to $1+\lambda_{\mathrm{het}}|x|$ or $1+\lambda_{\mathrm{het}}x^2$ for $\lambda_{\mathrm{het}}\in\{0.5,1,2\}$, $1+2/(1+x^2)$ for the inverse-quadratic design, and $x^2$ for the multiplicative design.  For quadratic mean misspecification, the data follow $y=x^\top\theta^{\mathrm{true}}+\beta x^2+\varepsilon$ with $\beta\in\{0.2,0.5,1.0\}$, while the detector fits only a linear mean.  The $H_1$ checks use balanced slopes $1\pm\Delta/2$: $\Delta\in\{0.5,1,2\}$ under homoskedasticity and $\Delta\in\{1,2\}$ for the linear and quadratic variance designs with $\lambda_{\mathrm{het}}=1$.  The principal numerical findings and rejection-rate panels appear in \Cref{fig:synthetic-robustness} of the main text.

\section{Supplementary Results for Real-World Datasets}\label{app:realworld-details}

 \subsection{Implementation Details for the DSV Toolkit and Benchmarks}\label{app:implementation}

For Stage~1, we use constant step size $\eta = 0.0005$ with 1{,}000 passes and average $\hat S_T$ over 10 independent runs.  For continuous focal covariates, the component-noise scale is estimated by the two-fold cross-fitted variance function regression in \Cref{sec:scale-estimation}; scalar critical values use the Gaussian-component analytic approximation specified there and based on \Cref{prop:null-distribution}(i), while 499 fitted Gaussian-component bootstrap runs resample centered observed covariates and rerun the complete scale-estimation and SGD pipeline when $d\ge2$.  The originally reported rule declares rejection when the selected statistic exceeds its fitted $\hat q_{0.95}$.  Taylor and Atlanta, whose focal covariates are binary, instead use the pooled-residual statistic and its matched distribution-adaptive scalar null calibration.  Their group-assignment probabilities vary with the focal covariate, a setting covered by the generalized alternative in \eqref{eq:dgp}.  More generally, the use of 1{,}000 passes recycles the same $n$ observations and does not satisfy the independent preliminary-sample and fresh-stream protocol of \Cref{cor:CI}.  The finite-sample null audit in \Cref{tab:finite-n-reuse} finds severe over-rejection for the continuous scalar analytic implementation at $n=150$, 333, 777, and 1{,}840.  Accordingly, the Stage~1 $p$-values and decisions in the real-data table are retained as reproducible outputs of the original implementation but are provisional rather than validated 5\%-level inference.
For Stage~2, we run two-piece SGD on the original data with an intercept, using constant step size $\eta = 0.002$ for 100 passes over 10 independent runs.

For the SP-discovery comparison, we supply the nominated focal response--covariate relationship and the same set of observed candidate covariates to all procedures, including the reference confounder without identifying it during fitting.  In the Xu-style implementation of \citet{xu2018detecting}, each candidate is considered as a possible stratification, continuous candidates are discretized into quartiles, and the within-stratum slopes are compared with the pooled slope.  Motivated by the practical relaxation of the all-strata condition in \citet{alipourfard2018simpsons}, the reference reversal is counted as surfaced when a strict majority of its eligible strata reverse the pooled sign.  Eligible strata contain at least five observations, except that the application-specific verification thresholds are 200 for Power, 100 for Nitrogen deposition, and 30 for Nematodes.

For the tree benchmark, we implement the X-terminal ordering rule of \citet{shmueli2018forest} with CART.  We one-hot encode categorical candidates and fit regression trees of depths two through five using the focal covariate and all candidates, with minimum leaf size $\max\{10,\lfloor n/100\rfloor\}$.  Each path is stopped at its first split on the focal covariate.  After fitting, the reference reversal is counted as surfaced only if the reference candidate occurs above that focal split and the slope in the corresponding node reverses the pooled sign.  Thus, the reference identity is used to score the output, not to construct the tree.

For the likelihood-based detection benchmark, we implement the $m=1$ versus $m=2$ modified penalized-EM test of \citet{kasahara2015testing}.  The detector receives the same response and full Stage~1 predictor design as SGD, without candidate confounders.  We use the recommended initial split $\mathcal T=\{0.5\}$, $K=2$ modified-EM statistic, penalty constants $a_n=2.2$ for one-dimensional designs and $a_n=8.3$ for the four-dimensional California Housing design, and the restriction $\sigma_j\geq0.01\widehat\sigma$.  The restricted penalized MLE is computed from four structured EM starts followed by analytic-gradient refinement.  We calibrate each test by 199 fixed-design parametric bootstrap samples from the fitted homogeneous regression and use the plus-one bootstrap $p$-value.

As a computationally lighter EM benchmark, we also fit homogeneous and homoscedastic two-component normal mixture regressions by conventional EM~\citep{dempster1977,mclachlan2000finite}, using the same standardized response and Stage~1 predictor design.  The two-component fit uses a variance shared across components, 10 starts, and at most 250 iterations.  EM--BIC selects $m=2$ when $2(\ell_2-\ell_1)>(p+1)\log n$, where $p$ is the number of regression coefficients, including the intercept.  This comparison requires only the observed-data fits and no bootstrap calibration.

\begin{table}[!t]
	\centering
	\caption{Mixture detection and Simpson's-paradox discovery across eight datasets.}
	\label{tab:confounder-all}
	\footnotesize
	\begin{tabular}{@{}l cc@{\hspace{3pt}}c@{\hspace{3pt}}cc@{\hspace{7pt}}c ccc@{}}
		\toprule
		& \multicolumn{5}{c}{\textbf{Mixture detection}} & \multicolumn{4}{c}{\textbf{SP discovery}} \\
		\cmidrule(lr){2-6}\cmidrule(lr){7-10}
		Dataset & $\hat S_T$ & $\hat q_{0.95}$ & $p$ & KS & EM--BIC & Ref.~$Z$ (top/10) & DSV & Xu-style$^\ddagger$ & X-term.$^\ddagger$ \\
		\midrule
		Taylor$^\dagger$   & $1.880$ & $1.707$ & ${<}10^{-15}$  & \checkmark & \checkmark & Age (10/10)  & \checkmark & \checkmark & \checkmark \\
		Iris$^\dagger$     & $1.444$ & $1.356$ & ${<}10^{-12}$  & $\times$ & $\times$ & species (10/10) & \checkmark & \checkmark & $\times$ \\
		Penguins$^\dagger$ & $1.643$ & $1.371$ & ${<}10^{-15}$  & \checkmark & \checkmark & species (10/10) & \checkmark & \checkmark & \checkmark \\
		Atlanta$^\dagger$  & $1.357$ & $1.305$ & $3.7\times10^{-6}$  & \checkmark & \checkmark & org.\ (10/10)   & \checkmark & \checkmark & $\times$ \\
		Power$^\star$      & $2.110$ & $0.838$ & ${<}10^{-15}$ & \checkmark & \checkmark & country (10/10) & \checkmark & \checkmark & \checkmark \\
		CA Hous.$^\star$   & $1.683$ & $1.532$ & $0.002$  & \checkmark & \checkmark & MedInc (10/10)  & \checkmark & \checkmark & \checkmark \\
		Nematodes$^\star$  & $1.419$ & $1.335$ & ${<}10^{-10}$ & \checkmark & \checkmark & biome (10/10) & \checkmark & \checkmark & $\times$ \\
		\shortstack[l]{Nitrogen\\deposition}$^\star$ & $1.505$ & $1.358$ & ${<}10^{-15}$ & \checkmark & \checkmark & project (10/10) & \checkmark & \checkmark & \checkmark \\
		\bottomrule
	\end{tabular}
	\\[3pt]
	\begin{minipage}{\textwidth}
		\footnotesize\raggedright
		$\dagger$: established SP benchmark; $\star$: SP not previously documented to our knowledge.  The first three columns report SGD detection outputs; the finite-sample audit in \Cref{tab:finite-n-reuse} makes these decisions provisional pending full-pipeline calibration.  KS reports the bootstrap-calibrated $m=1$ versus $m=2$ homogeneity decision, while EM--BIC reports whether conventional EM selects $m=2$; neither denotes SP discovery.  ``Ref.~$Z$'' is the scientifically interpretable reference grouping, and ``top/10'' is the number of DSV Screen runs in which it is top-ranked.  A DSV checkmark requires the computed Stage~1 rejection, recovery of a candidate grouping, and verification of its reversal.  $\ddagger$: Xu-style enumeration and X-terminal CART require the focal relationship and observed candidate confounders from the outset; they are SP-discovery procedures, not direct benchmarks for candidate-free mixture detection.\par
	\end{minipage}
\end{table}

\paragraph*{Comparison results.}
The originally used SGD detection rule exceeds its fitted critical value on all eight datasets, comprising seven scalar and one multivariate design.  In light of the finite-sample audit, this statement describes the computed outputs rather than eight validated 5\%-level rejections.  KS rejects homogeneity on seven datasets and does not reject on Iris ($p=0.210$); the recommended implementation with $K=3$ modified-EM iterations gives the same decisions.  EM--BIC also selects two components on seven datasets, with Iris remaining homogeneous by BIC.  Turning to SP discovery, DSV Screen top-ranks the reference grouping on all eight datasets, and the complete procedure yields eight computed end-to-end discoveries.  Xu-style enumeration also surfaces all eight reference reversals, whereas X-terminal CART does so in five.  Iris is counted as a computed DSV discovery because species is recovered and its reversal is verified, although its SGD detection decision remains subject to the finite-sample calibration caveat.

The two blocks of \Cref{tab:confounder-all} answer different questions.  KS supplies calibrated likelihood-based evidence for two-component normal mixture regression, and EM--BIC supplies likelihood-based component selection without bootstrap calibration.  By contrast, the SP-discovery procedures use observed candidates: DSV continues from detection to an interpretable grouping and direct verification, Xu-style enumeration examines each supplied stratification, and X-terminal CART searches for candidates that precede the focal variable in a reversed tree path.  The latter two therefore are not direct comparators for Stage~1, which operates without candidate confounders.  Within the discovery block, the equal 8/8 totals for DSV and Xu-style enumeration should also be interpreted in light of their different outputs: DSV ranks a candidate linked to a detected partition, whereas enumeration reports qualifying stratifications.  The tree result depends on which variables are selected before the focal split.  Accordingly, the block records whether each procedure surfaces the reference reversal under its stated inputs, rather than comparing like-for-like mixture detectors.


 \subsection{Real-World Dataset Details and Confounding Mechanisms}\label{app:realworld-datasets}
Here, we provide the dataset sources and dataset-level context omitted from the compact main-text comparison.  We describe all eight datasets, their focal regressions, and the mechanisms represented by the reference groupings.
 
\paragraph*{Taylor healthcare expenditures.}
 The Taylor dataset~\citep{taylor2014simpson} contains 777 observations; the focal covariate is a Hispanic indicator and the outcome is annual healthcare expenditure.  The originally used pooled-residual statistic exceeds its fitted threshold ($\hat S_T^{\mathrm{pool}}=1.880$ and $\hat q_{0.95}^{\mathrm{pool}}=1.707$).  The pooled slope is negative, whereas the ethnicity slope is positive within five of six age cohorts.  Age is the confounder because the Hispanic sample is younger on average and younger individuals have lower expenditures, so differences in age composition dominate the within-cohort association.
 
\paragraph*{Iris.}
 Fisher's Iris dataset~\citep{fisher1936use} contains 150 flowers; the focal regression relates sepal width to sepal length, with species as the reference grouping.  The originally used statistic exceeds its fitted threshold ($\hat S_T=1.444$ and $\hat q_{0.95}=1.356$).  Iris species have distinct centers in both sepal dimensions, so the pooled association combines between-species differences with the within-species relationship and can have the opposite sign from the species-specific slopes.
 
\paragraph*{Palmer Penguins.}
 The Palmer Penguins dataset~\citep{gorman2014penguins} contains 333 complete observations; the focal regression relates bill length to bill depth, with species as the reference grouping.  The originally used statistic exceeds its fitted threshold ($\hat S_T=1.643$ and $\hat q_{0.95}=1.371$).  Adelie, Chinstrap, and Gentoo penguins occupy distinct regions of the bill-dimension plane, so pooling the species mixes morphological differences between species with the association within each species.
 
\paragraph*{Atlanta CES payroll.}
 The Atlanta payroll benchmark used by \citet{wang2023simnet} contains 8{,}246 employees; the focal covariate is gender and the outcome is salary.  The originally used pooled-residual statistic exceeds its fitted threshold ($\hat S_T^{\mathrm{pool}}=1.357$ and $\hat q_{0.95}^{\mathrm{pool}}=1.305$).  Organization or job category is the confounder because units differ in both gender composition and salary scale.  The pooled gender--salary association therefore reflects the allocation of employees across units as well as within-unit salary differences.
 
\paragraph*{Power Generation.}
 The CarbonMonitor-Power data~\citep{zhu2023carbonmonitor} record daily renewable and fossil generation by country; the focal regression relates fossil generation to renewable generation, with country as the reference grouping.  The originally used statistic exceeds its fitted threshold ($\hat S_T=2.110$ and $\hat q_{0.95}=0.838$).  The pooled slope is positive ($+2.06$), whereas the within-country slope is negative in 33 of 45 countries with near-complete daily series.  Country captures power-system size: large systems produce more of both sources, and this between-country scale effect overwhelms their within-country substitution.
 
\paragraph*{California Housing.}
 The California Housing dataset~\citep{pace1997sparse} contains 20{,}640 census-block groups; the regression uses house age, average rooms, latitude, and longitude to explain median house value, with median income as the reference grouping.  The originally used statistic exceeds its fitted threshold ($\hat S_T=1.683$ and $\hat q_{0.95}=1.532$).  In the focal room-count comparison, the pooled association is positive, but the within-income slope is negative in 9 of 10 income deciles.  Across income groups, wealthier neighborhoods tend to have both larger homes and higher prices; within an income group, high room counts often identify inland or rural locations where land is cheaper.

\paragraph*{Soil nematodes.}
 The source study of \citet{vandenhoogen2019soil} assembled 6{,}759 georeferenced samples to quantify and map the global abundance and functional-group composition of soil nematodes and to characterize the environmental drivers of their biogeography.  Our focal regression relates log bacterivore abundance to the 2009 human-footprint index, with WWF biome as the reference grouping.  The originally used statistic exceeds its fitted threshold ($\hat S_T=1.419$ and $\hat q_{0.95}=1.335$).  The pooled slope is positive ($+0.0064$), whereas the within-biome slope is negative in 7 of 9 biomes.  Biomes differ in both baseline abundance and typical human pressure, so between-biome variation overwhelms the negative within-biome association.
 
\paragraph*{Nitrogen deposition and plant richness.}
 The nitrogen deposition data~\citep{simkin2016conditional} contain 15{,}136 vegetation plots from multiple source projects; the focal regression relates log species richness to critical-load exceedance, with source project as the reference grouping.  The originally used statistic exceeds its fitted threshold ($\hat S_T=1.505$ and $\hat q_{0.95}=1.358$).  The pooled slope is positive ($+0.014$), whereas the within-project slope is negative in 7 of 9 projects, with six bootstrap intervals entirely below zero.  Source project captures geographic composition: projects sample regions with different baseline species pools and nitrogen exposures, so the national pooled regression combines between-region differences with the within-region response.

\end{document}